\RequirePackage{ifthen}

\newcommand{\typeof}{1} %
\newcommand{\longversion}[1]{\ifthenelse{\equal{\typeof}{0}}{}{#1}}
\newcommand{\shortversion}[1]{\ifthenelse{\equal{\typeof}{0}}{#1}{}}
\newcommand{\longshortversion}[2]{\ifthenelse{\equal{\typeof}{0}}{#2}{#1}}

\longshortversion{
	\documentclass[a4paper,10pt]{article}
	
	\title{On Dynamic Lifting and Effect Typing\\in Circuit Description Languages (Extended Version)}
	\author{Andrea Colledan and Ugo Dal Lago}
	\date{}
	
	\usepackage{a4wide}
	\usepackage{amsthm}
	\theoremstyle{plain}
	\newtheorem{theorem}{Theorem}[section]
	\newtheorem{proposition}[theorem]{Proposition}
	\newtheorem{lemma}[theorem]{Lemma}
	
	\theoremstyle{definition}
	\newtheorem{definition}[theorem]{Definition}
	\newtheorem{example}{Example}
	
	\usepackage{amsmath}
	\usepackage{amsfonts}
	\usepackage{amssymb}
	\usepackage{color}
	\usepackage{proof}
	\usepackage{mathrsfs}  
	\usepackage{bm}
	\usepackage{stmaryrd}
	\usepackage{mathtools}
	\usepackage{bbm}
	\usepackage{tikz}\usetikzlibrary{quantikz}
	\usepackage{forest}
	\usepackage{wrapfig}
	\usepackage{empheq}
	\usepackage{url}
	\usepackage[inference]{semantic}\setpremisesend{0pt}\setpremisesspace{10pt}\setnamespace{0pt}
	\usepackage{csquotes}
	\usepackage[strings]{underscore}
	\usepackage{eucal}
	\usepackage{subcaption}
	\usepackage{url}
	
	\bibliographystyle{plain}	
	
\newcommand{\PQ}{\textsf{Proto-Quipper}}
\newcommand{\PQM}{\textsf{Proto-Quipper-M}}
\newcommand{\PQL}{\textsf{Proto-Quipper-L}}
\newcommand{\PQS}{\textsf{Proto-Quipper-S}}
\newcommand{\PQD}{\textsf{Proto-Quipper-D}}
\newcommand{\PQK}{\textsf{Proto-Quipper-K}}
\newcommand{\QCL}{\textsf{QCL}}
\newcommand{\QML}{\textsf{QML}}
\newcommand{\Haskell}{\texttt{Haskell}}
\newcommand{\Quipper}{\texttt{Quipper}}

\newcommand{\wtype}[1]{\fname{#1}}
\newcommand{\gateid}[1]{\fname{#1}}
\newcommand{\fname}[1]{\operatorname{\mathsf{#1}}}
\newcommand{\setname}[1]{\mathit{#1}}

\newcommand{\labelfresh}{\fname{fresh}}
\newcommand{\cto}{\triangleright}
\newcommand{\bseval}{\Downarrow}
\newcommand{\bsdiverge}{\Uparrow}
\newcommand{\simpbseval}{\bseval}
\newcommand{\simpbsdiverge}{\bsdiverge}
\newcommand{\qubitt}{\wtype{Qubit}}
\newcommand{\bitt}{\wtype{Bit}}
\newcommand{\bang}{\;!}
\newcommand{\circt}[1]{\operatorname{\mathsf{Circ}}_{#1}}
\newcommand{\unitt}{\mathbbm{1}}

\newcommand{\genericSetOne}{X} \newcommand{\genericSetTwo}{Y}
\newcommand{\indexsetOne}{I}
\newcommand{\gateset}{\mathscr{G}}

\newcommand{\labelset}{\mathscr{L}}
\newcommand{\dlvalset}{\mathscr{V}}
\newcommand{\wtypeset}{\mathscr{W}}
\newcommand{\lcset}{\mathscr{Q}}
\newcommand{\dltreeset}{\mathscr{T}}
\newcommand{\assset}{\mathscr{A}}
\newcommand{\cassset}{\mathscr{P}}
\newcommand{\contextset}{\setname{CONTEXT}}
\newcommand{\termset}{\setname{TERM}}
\newcommand{\valset}{\setname{VAL}}
\newcommand{\mvalset}{\setname{MVAL}}
\newcommand{\typeset}{\setname{TYPE}}

\newcommand{\mtypeset}{\setname{MTYPE}}
\newcommand{\genericObjectOne}{x} \newcommand{\genericObjectTwo}{y}
\newcommand{\renamingOne}{\rho}
\newcommand{\alpharenamingOne}{\pi}
\newcommand{\gateOne}{g} 
\newcommand{\circuitOne}{C} \newcommand{\circuitTwo}{D}
\newcommand{\labOne}{\ell} \newcommand{\labTwo}{k} 
\newcommand{\dlvalOne}{u} \newcommand{\dlvalTwo}{s}
\newcommand{\labsubOne}{\EuScript{L}} 
\newcommand{\dlvalsubOne}{\EuScript{V}} 
\newcommand{\wtypeOne}{w}
\newcommand{\lcOne}{Q} \newcommand{\lcTwo}{L}
\newcommand{\varOne}{x} \newcommand{\varTwo}{y} \newcommand{\varThree}{z}
\newcommand{\assOne}{a} \newcommand{\assTwo}{b} \newcommand{\assThree}{c}
\newcommand{\contextOne}{\Gamma}
\newcommand{\pcontextOne}{\Phi}
\newcommand{\termOne}{M} \newcommand{\termTwo}{N}   
\newcommand{\valOne}{V} \newcommand{\valTwo}{W} \newcommand{\valThree}{Z}
\newcommand{\typeOne}{A} \newcommand{\typeTwo}{B}  
\newcommand{\ptypeOne}{P} \newcommand{\ptypeTwo}{R}
\newcommand{\mtypeOne}{T} \newcommand{\mtypeTwo}{U}
\newcommand{\MongenericObjectOne}{\xi}
\newcommand{\MongenericObjectTwo}{\theta}
\newcommand{\MonlcOne}{\Delta} \newcommand{\MonlcTwo}{\Lambda}
\newcommand{\MontermOne}{\mu}  
\newcommand{\MonvalOne}{\phi}
\newcommand{\MonvalTwo}{\psi}
\newcommand{\MonmvalOne}{\lambda} 
\newcommand{\MontypeOne}{\alpha} \newcommand{\MontypeTwo}{\beta}  
 \newcommand{\MonmtypeTwo}{\upsilon}
\newcommand{\Moncontext}{\gamma}

\newcommand{\emptyassignment}{\emptyset}
\newcommand{\emptylc}{\emptyset}
\newcommand{\emptycontext}{\emptyset}

\newcommand{\condmonad}{\operatorname{\EuScript{K}}}
\newcommand{\powerset}{\operatorname{\EuScript{P}}}
\newcommand{\invert}[1]{{#1}^{-1}}

\newcommand{\unitv}{*}
\newcommand{\tuple}[1]{(#1)}
\newcommand{\lift }{\operatorname{\mathsf{lift}}}
\newcommand{\force}{\operatorname{\mathsf{force}}}
\newcommand{\justboxt}{\operatorname{\mathsf{box}}}
\newcommand{\boxt}[2]{\justboxt_{#1}}
\newcommand{\boxedcirc}[4]{(#1,#2,#3)_{#4}}
\newcommand{\apply}[1]{\operatorname{\mathsf{apply}}_{#1}}
\newcommand{\letin}[3]{\mathsf{let}\; #1 = #2 \;\mathsf{in}\; #3}
\newcommand{\return}{\operatorname{\mathsf{return}}}
\newcommand{\cinput}[1]{\fname{input}(#1)}

\newcommand{\clift}{\fname{lift}}
\newcommand{\cond}{\operatorname{?}}
\newcommand{\lto}{\Rightarrow}

\newcommand{\dom}{\fname{dom}}

\newcommand{\assmerge}{\cup}
\newcommand{\renamed}[2]{#1|#2|}
\newcommand{\alpharenamed}[2]{#1\langle#2\rangle}

\newcommand{\gateinputs}{\fname{inType}}
\newcommand{\gateoutputs}{\fname{outType}}
\newcommand{\cconcat}{::}
\newcommand{\freshlabels}{\operatorname{\mathsf{freshlabels}}}
\newcommand{\append}{\operatorname{\mathsf{append}}}
\newcommand{\munit}{\operatorname{\eta}}

\newcommand{\Monmerge}[1]{\stackrel{#1}{,}}
\newcommand{\namesof}[1]{\fname{var}_{#1}}

\newcommand{\Moncomp}[4]{{#1}[{#2}]_{#3}^{#4}}
\newcommand{\Monsplit}[1]{\Yleft^{#1}}
\newcommand{\Monunit}[1]{\treeLeaf{#1}}
\newcommand{\Monconst}[2]{#1^{#2}}
\newcommand{\Monflatten}[1]{\lfloor#1\rfloor}

\newcommand{\circjudgment}[4]{#1 \vdash^{#2} #3 \cto #4}
\newcommand{\mjudgment}[3]{#1 \vdash_m #2 : #3}
\newcommand{\compjudgment}[5]{#1;#2 \vdash_c^{#3} #4 : #5}
\newcommand{\valjudgment}[4]{#1;#2 \vdash_v #3 : #4}
\newcommand{\lbscfgjudgment}[6]{#1 \vdash_{#2}^{#3} #4 : #5 ; #6}
\newcommand{\rbscfgjudgment}[7]{#1 \vdash_{#2}^{#4} #5 : #6 ; #7}
\newcommand{\simpjudgment}[3]{\vdash^{#2}#1:#3}
\newcommand{\Moncompjudgment}[5]{#1;#2 \Vdash_c^{#3} #4 : #5}
\newcommand{\Monvaljudgment}[5]{#1;#2 \Vdash_v^{#3} #4 : #5}

\newcommand{\dlvaltreeOne}{\mathfrak{t}} \newcommand{\dlvaltreeTwo}{\mathfrak{r}}
\newcommand{\emptytree}{\epsilon}
\newcommand{\treeNode}[3]{#1\,\{#2\}\{#3\}}
\newcommand{\treeLeaf}[1]{\{#1\}}

\newcommand{\bscfgl}[3]{(#1,#2,#3)}
\newcommand{\bscfgr}[2]{(#1,#2)}
\newcommand{\unbox}{\operatorname{\mathsf{unbox}}}
\newcommand{\ite}[3]{\mathsf{if}\; #1 \;{\mathsf{then}}\; #2 \;{\mathsf{else}}\; #3}

\newcommand{\for}{\operatorname{\textnormal{ for }}}

\newcommand{\void}{\phantom{H}}
\newcommand{\rulespace}{\vspace{4pt}}
\newcommand{\forceindent}{\phantom{_}}

\newcommand{\measlift}{\mathit{ML}}
\newcommand{\when}[2]{\mathsf{when}\; #1 \;\mathsf{then} #2}
\newcommand{\lindex}{{\ell}}
\newcommand{\rindex}{{r}}
}{
	\documentclass[a4paper,UKenglish,cleveref, autoref, thm-restate]{lipics-v2021}
	\hideLIPIcs  %uncomment to remove references to LIPIcs series (logo, DOI, ...), e.g. when preparing a pre-final version to be uploaded to arXiv or another public repository
	
	\bibliographystyle{plainurl}% the mandatory bibstyle
	
	\title{On Dynamic Lifting and Effect Typing in Circuit Description 
	Languages}
	\author{Andrea Colledan}{University of Bologna, Italy \and INRIA Sophia Antipolis, France}{andrea.colledan@unibo.it}{https://orcid.org/0000-0002-0049-0391}{}	
	\author{Ugo {Dal Lago}}{University of Bologna, Italy \and INRIA Sophia Antipolis, France}{ugo.dallago@unibo.it}{https://orcid.org/
		0000-0001-9200-070X}{}
	\authorrunning{A. Colledan and U. Dal Lago}
	
	\Copyright{Andrea Colledan and Ugo Dal Lago} %TODO mandatory, please use full first names. LIPIcs license is "CC-BY";  http://creativecommons.org/licenses/by/3.0/
	
	 %TODO mandatory: Please choose ACM 2012 classifications from https://dl.acm.org/ccs/ccs_flat.cfm 
	
	 %TODO mandatory; please add comma-separated list of keywords
	
	\category{} %optional, e.g. invited paper
	
	\relatedversion{} %optional, e.g. full version hosted on arXiv, HAL, or other respository/website
	\EventEditors{John Q. Open and Joan R. Access}
	\EventNoEds{2}
	\EventLongTitle{42nd Conference on Very Important Topics (CVIT 2016)}
	\EventShortTitle{CVIT 2016}
	\EventAcronym{CVIT}
	\EventYear{2016}
	\EventDate{December 24--27, 2016}
	\EventLocation{Little Whinging, United Kingdom}
	\EventLogo{}
	\SeriesVolume{42}
	\ArticleNo{23}
	\usepackage{amsmath}
	\usepackage{amsfonts}
	\usepackage{amssymb}
	\usepackage{color}
	\usepackage{proof}
	\usepackage{mathrsfs}  
	\usepackage{bm}
	\usepackage{stmaryrd}
	\usepackage{mathtools}
	\usepackage{bbm}
	\usepackage{tikz}\usetikzlibrary{quantikz}
	\usepackage{forest}
	\usepackage{wrapfig}
	\usepackage{empheq}
	\usepackage{url}
	\usepackage[inference]{semantic}\setpremisesend{0pt}\setpremisesspace{10pt}\setnamespace{0pt}
	\usepackage{csquotes}
	\usepackage[strings]{underscore}
	\usepackage{eucal}
	\usepackage{subcaption}

}

\begin{document}

\maketitle

\begin{abstract}
In the realm of quantum computing, circuit description languages represent a valid alternative to traditional QRAM-style languages. They indeed allow for finer control over the output circuit, without sacrificing flexibility nor modularity. We introduce a generalization of the paradigmatic lambda-calculus \PQM, itself modeling the core features of the quantum circuit description language \Quipper. The extension, called \PQK, is meant to capture a very general form of dynamic lifting. This is made possible by the introduction of a rich type and effect system in which not only \textit{computations}, but also the very \textit{types} are effectful. The main results we give for the introduced language are the classic type soundness results, namely subject reduction and progress.
\end{abstract}

%%%%%%%%%%%%%%%%%%%%%%%%%%%
\section{Introduction}
%%%%%%%%%%%%%%%%%%%%%%%%%%%

Despite the undeniable fact that large-scale, error-free quantum hardware has  yet to be built \cite{nisq}, research into programming languages specifically  designed to be  compiled towards architectures including quantum hardware has  taken hold in recent years \cite{quantum-languages}. Most of the proposals in this sense (see~\cite{survey-gay,survey-selinger,ying}  for some surveys) concern languages that either express or can be  compiled into some form of \emph{quantum circuit}~\cite{nielsen-chuang}, which can then be executed by quantum hardware. This reflects the need to have tighter control over the kind and amount of quantum resources that programs require. 

In this scenario, the idea of considering high-level languages that are  specifically designed to \emph{describe} circuits and in which the latter are  treated as first-class citizens is particularly appealing. A typical example  of this class of languages is \Quipper~\cite{quipper-intro,quipper}, whose underlying design  principle is precisely that of enriching a very expressive and powerful  functional language like \Haskell\ with the possibility of manipulating quantum  circuits. In other words, programs do not just build circuits, but also treat them as data, as can be seen in the example from Figure \ref{fig: alice}. \Quipper's meta-theory has been studied in recent years through the  introduction of a family of research languages that correspond to suitable  \Quipper\ fragments and extensions, which usually take the form of linear  $\lambda$-calculi. We are talking about a family of languages whose members include \PQS~\cite{proto-quipper-s}, \PQM~\cite{proto-quipper-m}, \PQD~\cite{proto-quipper-d} and \PQL~\cite{proto-quipper-l}. In this paper, we define an extension of \PQM.

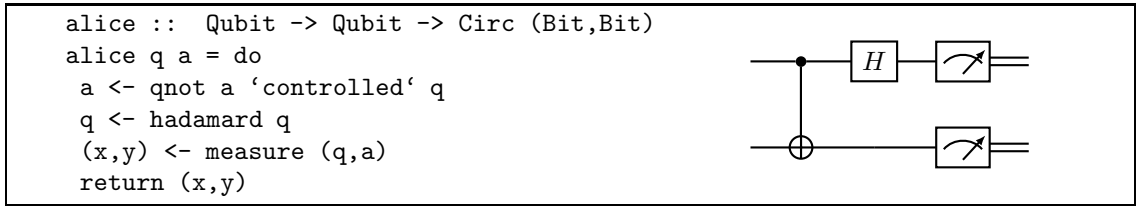
\begin{figure}[tb]
	\fbox{\begin{minipage}{.98\textwidth}
			\centering
			\begin{minipage}{.6\textwidth}
				\texttt{alice :: Qubit -> Qubit -> Circ (Bit,Bit)\\
					alice q a = do\\
					\forceindent a <- qnot a `controlled` q\\
					\forceindent q <- hadamard q\\
					\forceindent (x,y) <- measure (q,a)\\
					\forceindent return (x,y)}
			\end{minipage}
			\begin{minipage}{.3\textwidth}
				\begin{quantikz}
					& \ctrl{1}	& \gate{H}	& \meter{} & \cw \\
					& \targ{}	&  \qw 		& \meter{} & \cw
				\end{quantikz}
			\end{minipage}
	\end{minipage}}
	\caption{Alice's part of the quantum teleportation circuit. The \Quipper\ 
		program on the left builds the circuit on the right, but while doing so it also manipulates the smaller circuit \texttt{qnot a}, enriching it with control.}
	\label{fig: alice}
\end{figure}

An aspect that so far has only marginally been considered by the research community is 
the study of the meta-theory of so-called \emph{dynamic lifting}, i.e. the 
possibility of allowing the (classical) value flowing in 
one of the wires of the underlying circuit, naturally unknown at circuit 
building time, to be \emph{visible} in the host program for control flow. As 
an example, one could append a unitary to some of the wires \emph{only if} a previously performed measurement operation 
has yielded a certain outcome. This is commonly achieved in many quantum 
algorithms via classical control\longversion{ (see Figure \ref{fig: teleportation})}, but Quipper also offers a higher-level solution precisely in the form of dynamic lifting, 
as can be seen in the example program in Figure \ref{fig: teleportation-dl}. Notably, such a program cannot be captured by any of the calculi in the \PQ\ 
family, with the exception of Lee et al.'s \PQL~\cite{proto-quipper-l}, arguably the most 
recent addition to the family.

\longversion{
	\begin{figure}[tb]
		\fbox{\begin{minipage}{.98\textwidth}
				\centering
				\begin{minipage}{.8\textwidth}
					\texttt{teleport :: Qubit -> Qubit -> Qubit -> Circ Qubit\\
						teleport b q a = do\\
						\forceindent a <- qnot a `controlled` q\\
						\forceindent q <- hadamard q\\
						\forceindent (x,y) <- measure (q,a)\\
						\forceindent b <- gate_X b `controlled` y\\
						\forceindent b <- gate_Z b `controlled` x\\
						\forceindent return b}
				\end{minipage}
				\hspace{-10em}%
				\begin{minipage}{.4\textwidth}
					\begin{quantikz}[row sep = 2pt, column sep = 7pt]
						&\qw		&\qw		&\qw		&\gate{X}	&\gate{Z}	&\qw\\
						&\ctrl{1}	&\gate{H}	&\meter{} 	&\cw 		&\cwbend{-1}	&\cw\\
						&\targ{}	&\qw		&\meter{} 	&\cwbend{-2}	&\cw		&\cw
					\end{quantikz}
				\end{minipage}
		\end{minipage}}
		\caption{Quantum teleportation circuit with classical control}
		\label{fig: teleportation}
	\end{figure}
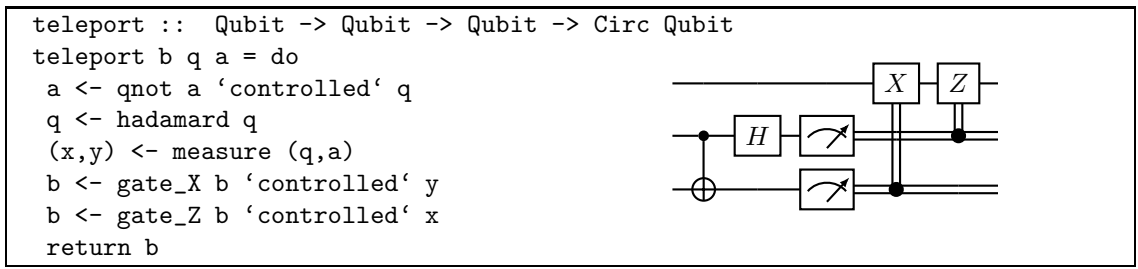
}

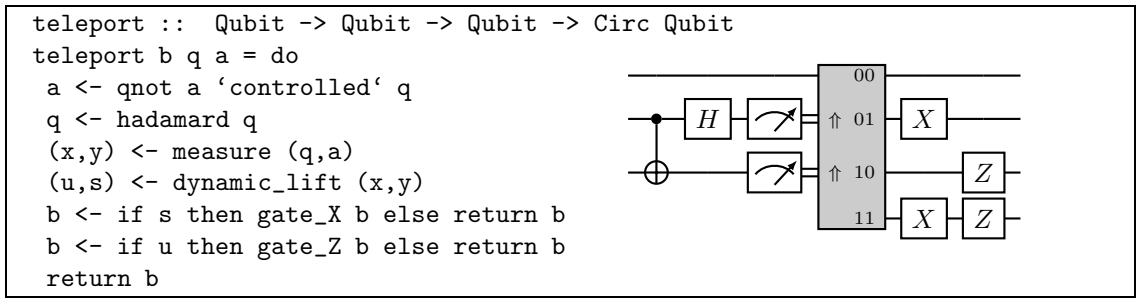
\begin{figure}[tb]
	\fbox{\begin{minipage}{.98\textwidth}
			\centering
			\begin{minipage}{.8\textwidth}
				\texttt{teleport :: Qubit -> Qubit -> Qubit -> Circ Qubit\\
					teleport b q a = do\\
					\forceindent a <- qnot a `controlled` q\\
					\forceindent q <- hadamard q\\
					\forceindent (x,y) <- measure (q,a)\\
					\forceindent (u,s) <- dynamic_lift (x,y)\\
					\forceindent b <- if s then gate_X b else return b\\
					\forceindent b <- if u then gate_Z b else return b\\
					\forceindent return b}
			\end{minipage}
			\longshortversion{\hspace{-12em}}{\hspace{-10em}}
			\begin{minipage}{.44\textwidth}
				\begin{quantikz}[row sep = 2pt, column sep = 6pt]
					& \qw		& \qw		&  \qw 		&\gate[wires={4}, style={fill=black!20},nwires={4},cwires={2,3}][25pt]{}\gateoutput{00}	& \qw		& \qw	& \qw\\
					& \ctrl{1}	& \gate{H}	&  \meter{} &\gateinput{$\Uparrow$}\gateoutput{01}		& \gate{X} & \qw	& \qw\\
					& \targ{}	&  \qw		& \meter{} 	&\gateinput{$\Uparrow$}\gateoutput{10}			& \qw	& \gate{Z} & \qw\\
					&			&			&			&\gateoutput{11}			& \gate{X} 	& \gate{Z} & \qw
				\end{quantikz}
			\end{minipage}
	\end{minipage}}
	\caption{Quantum teleportation circuit with dynamic lifting. The gray box is not a gate, but rather represents the dynamic lifting of the bit wires marked with $\Uparrow$ and the extension of the circuit with one of four possible continuations for the remaining wire, depending on the outcome of the intermediate measurements.}
	\label{fig: teleportation-dl}
\end{figure}

Looking at the \Quipper\ program in Figure \ref{fig: teleportation-dl}, one  immediately realizes that the two branches of both occurrences  of the  \texttt{if} operator change the underlying circuit in a uniform way, i.e. the  number and type of the wires are the same in either branch. What if, for instance, we wanted to condition the execution of a measurement on a lifted value, like in Figure~\ref{fig: one-way}? Situations such as the one just mentioned can occur, for example in one-way quantum computing, and the so-called measurement calculus~\cite{measurement-calculus} indeed allows the execution of a measurement to be conditioned on the result of a previously performed measurement. Unfortunately, \Quipper\ does \emph{not} allow the program in Figure \ref{fig: one-way} to be typed, and it is thus natural to wonder whether this is an intrinsic limitation, or if a richer type system can deal with a more general form of circuits.

\begin{figure}[tb]
	\fbox{\begin{minipage}{.98\textwidth}
			\centering
			\begin{minipage}{.62\textwidth}
				
				\texttt{oneWay :: Qubit -> Qubit -> ???\\
					oneWay q a = do\\
					\forceindent q <- hadamard q\\
					\forceindent x <- measure q\\
					\forceindent u <- dynamic_lift x\\
					\forceindent out <- if u then measure a else return a\\
					\forceindent return out}
			\end{minipage}
			\begin{minipage}{.33\textwidth}
				\begin{quantikz}[row sep = 2pt, column sep = 6pt]
					& \gate{H}	& 	\meter{}	& \gate[wires={2}, style={fill=black!20},cwires={1}][25pt]{}\gateinput{$\Uparrow$}\gateoutput{0}	& \qw	& \qw\\
					& \qw 		& \qw 			& \gateoutput{1}	& 	\meter{}	&\cw
				\end{quantikz}
			\end{minipage}
	\end{minipage}}
	\caption{An example of conditional measurement. This program is ill-typed in \Quipper.}
	\label{fig: one-way}
\end{figure}
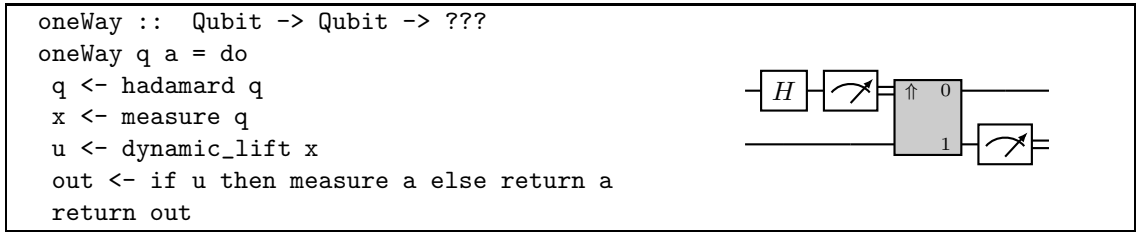

In this paper, we introduce a generalization of \PQM, called  \PQK, in which  dynamic lifting is available in a very  general form. This newly introduced language is capable of producing  not only circuits like the one in Figure \ref{fig: teleportation-dl}, but rather a more general class of circuits whose  structure and type \emph{essentially depend} on the values flowing through the  lifted channels, like the one in Figure~\ref{fig: one-way}. We show along the paper that this asks for a non-trivial  generalization of \PQM's type system, in which types reflect the different behaviors a circuit can have. This is achieved through a type and effect system \cite{types-and-effects} which assigns any \PQK\ computation (possibly) distinct types depending on the state of the lifted variables. The main results, beside the introduction of the language itself, its type  system, and its operational semantics, are the type soundness results of  subject reduction and progress, which together let us conclude that well-typed \PQK\ programs do not go wrong. \shortversion{A long version of this paper with more details and proofs is available \cite{extended-ver}.}

%%%%%%%%%%%%%%%%%%%%%%%%%%%%%%%%%%%%%%%%%%%%%%%%%%%%%%%%%%%%%%%
\section{Circuits and Dynamic Lifting: a Bird's-Eye View}
\label{sec: circuit-building languages}
%%%%%%%%%%%%%%%%%%%%%%%%%%%%%%%%%%%%%%%%%%%%%%%%%%%%%%%%%%%%%%%

This section is meant to provide an informal introduction to the peculiarities of the 
\PQ\ family of paradigmatic programming languages, for the non-specialists. The host language, namely \Haskell, is 
modeled as a linearly typed $\lambda$-calculus. Terms, in addition to manipulating the usual data structures and working with higher-order functions, 
act on an underlying circuit, which we usually refer to as $\circuitOne$. During program evaluation, $\circuitOne$ can be modified with the addition of wires, gates or entire sub-circuits. This is 
made possible through a dedicated operator called $\apply{}$. But how can the 
programmer specify \emph{where} in $\circuitOne$ these modifications have to be carried out? This is possible thanks to the presence of \emph{labels}, that is, names that identify distinct output wires of $\circuitOne$. 
These labels are ordinary terms which can be passed around, but have to be treated linearly. 
Among other things, they can be passed to $\apply{}$, together with the 
specification of which gate or sub-circuit is to be appended to the underlying circuit $\circuitOne$.

But is $\apply{}$ the only way of manipulating circuits? The answer is 
negative. Circuits, once built by means of a term, can be ``boxed'', potentially 
copied, and passed to other parts of the program, where they can be used, usually by appending them to multiple parts of the underlying circuit. From a linguistic point of view, this is possible thanks to an additional operator, called $\justboxt$, which is responsible for turning a circuit-building term $\termOne$ of type $\bang(\mtypeOne\multimap\mtypeTwo)$ -- the type of duplicable functions from label tuples to label tuples -- into a term of type $\circt{}(\mtypeOne,\mtypeTwo)$ -- the type of circuits. The term 
$\justboxt\termOne$ does not touch the underlying circuit, but rather evaluates $\termOne$ ``behind the scenes'', in a sandboxed environment, to obtain a standalone circuit which is then returned as a \emph{boxed circuit}.

\paragraph*{Measure as a Label-Lifting Operation.}
The above considerations are agnostic to the kind of circuits being built. In  fact, any type of circuit-like structure can be produced in output by programs of 
the \PQ\ family of languages, provided that it can be interpreted as a morphism in an underlying 
symmetric monoidal category \cite{proto-quipper-m}.  If we imagine that the produced structure 
is an \textit{actual} quantum circuit, however, it is only natural to wonder whether all the examples of 
circuits that we talked about informally in the introduction can be captured by an instance of \PQ. 
Unsurprisingly, the program in Figure~\ref{fig: alice} is not at all 
problematic, and can be handled easily by all languages in the \PQ\ family (see Figure \ref{fig: pqk alice}).

\begin{figure}[tb]
	\fbox{\begin{minipage}{.98\textwidth}
			\vspace{-1em}
			\begin{align*}
				\lambda q_{\qubitt} .\lambda a_{\qubitt}.
				&\letin{\tuple{q,a}}{\apply{}(\gateid{CNOT},\tuple{q,a})}{\\
					&\letin{q}{\apply{}(\gateid{H},q)}{\apply{}(\gateid{Meas2},\tuple{q,a})}}
			\end{align*}
	\end{minipage}}
	\caption{A \PQM\ program describing the circuit shown in Figure \protect\ref{fig: alice}. We assume that we have a constant boxed circuit $\gateid{CNOT},\gateid{H}$, etc. for every available native gate.}
	\label{fig: pqk alice}
\end{figure}

The program in Figure~\ref{fig: teleportation-dl}, on the other hand, can only be handled by \PQL\ (see Figure \ref{fig: pql teleportation-dl}). In this instance of the \PQ\ family, a measurement is not limited to consuming a qubit to return a bit. Rather, a \PQL\ measurement can return the \textit{Boolean value} corresponding to the outcome of the measurement. When this happens, the ongoing computation is split into two branches: one in which the Boolean output is true and one in which it is false. This way, further circuit-altering operations down the line can depend on the classical information produced by the intermediate measurement.

\begin{figure}[tb]
	\fbox{\begin{minipage}{.98\textwidth}
			\vspace{-1em}
			\begin{align*}
				\lambda b.\lambda q.\lambda a.&\letin{\tuple{q,a}}{(\unbox\gateid{CNOT})\,\tuple{q,a}}{\\
					&\letin{q}{(\unbox\gateid{H})\,q}{\\
						&\letin{\tuple{\dlvalOne,\dlvalTwo}}{(\unbox\gateid{MeasLift2})\,\tuple{q,a}}{\\
							&\letin{b}{\ite{\dlvalTwo}{(\unbox\gateid{X})\,b}{b}}{\\
								&\letin{b}{\ite{\dlvalOne}{(\unbox\gateid{Z})\,b}{b}}{b}}}}}
			\end{align*}
	\end{minipage}}
	\caption{A \PQL\ program describing the circuit shown in Figure \protect\ref{fig: teleportation-dl}. Informally, terms such as $(\unbox\gateid{H})\,q$ corresponds to $\apply{}(\gateid{H},q)$ in \PQM.}
	\label{fig: pql teleportation-dl}
\end{figure}

How about the program in Figure~\ref{fig: one-way}? Unfortunately, this case exceeds the expressiveness of even \PQL, in that it requires the distinct execution branches to yield values of different \textit{types}. \PQL, on the other hand, requires all branches of a computation to share the same type. This is where our contribution starts.

\paragraph*{The Basic Ideas Underlying \PQK}

The approach to dynamic lifting that we follow in this paper is radical. The 
evaluation of a term $\termOne$ involving the $\apply{}$ operator can give rise to 
the lifting of a bit value into a variable $\dlvalOne$ and consequently produce in output not a 
single result in the set $\valset$ of values, but possibly one distinct result for each possible 
value of $\dlvalOne$. Therefore, it is natural to think of $\termOne$ as a 
computation that results in an object of the set 
$\condmonad_{\{\dlvalOne\}}(\valset)$, where $ \condmonad_{\{\dlvalOne\}} = X 
\mapsto (\{\dlvalOne\} \rightarrow \{0,1\}) \rightarrow \genericSetOne$ is a 
functor such that for each possible assignment of a Boolean value to $\dlvalOne$, 
$\condmonad_{\{\dlvalOne\}}(X)$ returns an element of $X$.

What if \emph{more than one} variable is lifted? For example, a program could lift 
$\dlvalTwo$ after having lifted $\dlvalOne$, but \emph{only if} the latter has 
value $0$. This shows that one cannot just take 
$\condmonad_{\{\dlvalOne,\dlvalTwo\}}= X \mapsto (\{\dlvalOne,\dlvalTwo\} 
\rightarrow \{0,1\}) \rightarrow \genericSetOne$, simply because not all 
assignments in $\{\dlvalOne,\dlvalTwo\}\rightarrow \{0,1\}$ are relevant. 
Instead, one should just focus on the three assignments $(\dlvalOne=0,\dlvalTwo=0),(\dlvalOne=0,\dlvalTwo=1)$ and $(\dlvalOne=1)
$, namely those assignments which are consistent with the tree in Figure 
\ref{fig: lifting tree}, which we call a \emph{lifting tree}. This is a key 
concept in this work, which we will discuss in detail in Section \ref{sec: generalized circuits}. Our type system captures the lifting pattern of an
underlying well-typed program through a lifting tree $\dlvaltreeOne$ and the result of the corresponding computation is an element of $\condmonad_{\dlvaltreeOne}(\valset)$ where $\condmonad_{\dlvaltreeOne}=\genericSetOne\mapsto (\cassset_\dlvaltreeOne \to \genericSetOne)$ and $\cassset_\dlvaltreeOne$ is the 
set of all assignments of variables which describe a root-to-leaf path in $\dlvaltreeOne$.
Since by design we want to handle situations in which a circuit, and by 
necessity the term building it, can produce 
results which have distinct \emph{types} -- and not only distinct \emph{values} --
depending on the values of the lifted variables, we also employ (in the
spirit of the type and effects paradigm \cite{types-and-effects}) an effectful notion of \emph{type}, in which computations are typed according to an element of $\condmonad_{\dlvaltreeOne}(\typeset)$, where $\typeset$ is the set of \PQK\ types.

\begin{figure}[tb]
	\fbox{\begin{minipage}{.98\textwidth}
			\centering
			\begin{forest}
				[$\dlvalOne$,grow=east [$\dlvalTwo$,grow=east,edge label={node[midway,below,font=\scriptsize]{0}} [$\emptytree$,edge label={node[midway,below,font=\scriptsize]{0}}] [$\emptytree$,edge label={node[midway,above,font=\scriptsize]{1}}]] [$\emptytree$,edge label={node[midway,above,font=\scriptsize]{1}}]]
			\end{forest}
	\end{minipage}}
	\caption{An example of a lifting tree, where $\emptytree$ is the empty tree}
	\label{fig: lifting tree}
\end{figure}
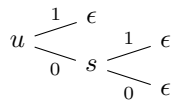

%%%%%%%%%%%%%%%%%%%%%%%%%%%%%%%%%%%%%%%%%
\section{Generalized Quantum Circuits}
\label{sec: generalized circuits}
%%%%%%%%%%%%%%%%%%%%%%%%%%%%%%%%%%%%%%%%%

A quantum circuit describes a quantum computation by means of the application 
of \emph{quantum gates}, which represent basic unitary operations, to a number 
of typed \textit{input wires}, to obtain a number of typed \textit{output wires}.
In this section, we introduce a general form of circuits in which the 
application of \emph{any} gate to one or more wires can be carried out
\emph{conditionally} on the classical value flowing in a lifted channel. This is possible 
even when the gate inputs are \textit{not} the same number and type of the gate 
outputs. This implies that the number and type of the outputs of a circuit can depend on the values flowing in its wires.
		
\paragraph*{A Syntax for Circuits}
We represent the inputs and outputs of circuits as \emph{label contexts}, that is, partial mappings from the set $\labelset$ of \emph{label names} to the set $\wtypeset = \{\wtype{Bit}, \wtype{Qubit}\}$ of wire types. Set $\labelset$ contains precisely the kind of labels that we mentioned in Section \ref{sec: circuit-building languages}, therefore a label context is a way to attach type information labeled wires. We write the set of all label contexts as $\lcset$.
		
Whereas the order of wires in a circuit as a whole is irrelevant, the order of wires in a gate application is crucial. For this reason, we perform gate applications not on label contexts, but rather on \emph{label tuples}, which imbue label contexts with a specific ordering via a simple form of typing judgment. The grammar and typing rules for label tuples, which we call \textit{M-values} following~\cite{proto-quipper-m}, are given in Figure \ref{table: mvalues and mtypes}, where $\labOne\in\labelset,\wtypeOne\in\wtypeset$, and $\lcOne$ and $\lcTwo$ are label contexts whose disjoint union is denoted by $\lcOne,\lcTwo$.
		
\begin{figure}[tb]
	\centering
	\fbox{\begin{minipage}{.98\textwidth}\centering
			
			\begin{tabular}{l l l}
				M-types & $\mtypeOne,\mtypeTwo$  &  $::= \unitt \mid \wtypeOne \mid \mtypeOne \otimes \mtypeTwo.$\\
				M-values & $\vec{\labOne},\vec{\labTwo}$ &  $::= \unitv \mid \labOne \mid \tuple{\vec{\labOne},\vec{\labTwo}}.$
			\end{tabular}
			$$
			\inference[unit]
			{}
			%---
			{\mjudgment{\emptycontext}{\unitv}{\unitt}}
			\qquad
			\inference[label]
			{}
			%---
			{\mjudgment{\labOne:\wtypeOne}{\labOne}{\wtypeOne}}
			\qquad
			\inference[tuple]
			{
				\mjudgment{\lcOne}{\vec\labOne}{\mtypeOne}
				&
				\mjudgment{\lcTwo}{\vec\labTwo}{\mtypeTwo}
			}
			%---
			{\mjudgment{\lcOne,\lcTwo}{\tuple{\vec\labOne,\vec\labTwo}}{\mtypeOne\otimes\mtypeTwo}}
			$$
	\end{minipage}}
	\caption{Syntax and rule system for M-types and M-values}
	\label{table: mvalues and mtypes}
\end{figure}

\begin{definition}[Gate Set]
	Let $\gateset$ be a set of \emph{gates}, equipped with two functions $\gateinputs: \gateset \to \mtypeset$ and $\gateoutputs: \gateset \to \mtypeset$. We denote by $\gateset(\mtypeOne,\mtypeTwo)$ the set of gates with input type $\mtypeOne$ and output type $\mtypeTwo$.\longversion{ More formally:
		\begin{equation}
			\gateset(\mtypeOne,\mtypeTwo)=\{\gateOne \in\gateset \mid \gateinputs(\gateOne)=\mtypeOne \wedge \gateoutputs(\gateOne)=\mtypeTwo\}.
	\end{equation}}
\end{definition}
Besides the set of labels $\labelset$, there is also another set of names, called $\dlvalset$, which is disjoint from it and contains the lifted variables. An \emph{assignment} of lifted variables is then simply a finite sequence of equalities $(\dlvalOne_1 = b_1,\ldots,\dlvalOne_n = b_n)$ which assign the values $b_1,\dots,b_n\in\{0,1\}$ to the distinct variables $\dlvalOne_1,\dots,\dlvalOne_n\in\dlvalset$, respectively. We usually indicate assignments with metavariables like $\assOne,\assTwo$ and $\assThree$.
		
\longversion{
\begin{definition}[Assignment]
	Given a subset of labels $\dlvalsubOne \subseteq \dlvalset$, an \emph{assignment of the variables in $\dlvalsubOne$} is a mapping $\dlvalsubOne \to \{0,1\}$.
	The empty assignment is denoted by $\emptyassignment$, whereas the union of two assignments $\assOne$ and $\assTwo$ with disjoint domains is written as $\assOne\assmerge\assTwo$.
\end{definition}
}
		
We now introduce a low-level language to describe quantum circuits at the gate level, which will serve as a target for circuit building in \PQK. We call it \textit{circuit representation language} (CRL) and define it via the following grammar:
\begin{equation}
	\circuitOne, \circuitTwo ::= \cinput{\lcOne}
	\mid \circuitOne; a \cond \gateOne(\vec\labOne) \to \vec\labTwo
	\mid \circuitOne; a \cond \lift(\labOne) \Rightarrow \dlvalOne.
\end{equation}
The base case $\cinput{\lcOne}$ corresponds to the trivial identity circuit that takes as 
input the wires represented by $\lcOne$ and returns them unchanged.
The notation $\assOne \cond \gateOne(\vec\labOne) \to \vec\labTwo$ denotes the 
application of a gate $\gateOne$ to the wires identified by
$\vec\labOne$ to obtain the wires in $\vec\labTwo$, 
provided that the condition expressed by $\assOne$ is met. We simply write 
$\gateOne(\vec\labOne) \to \vec\labTwo$ when a gate is applied unconditionally 
(i.e. when $\assOne = \emptyassignment$). Figure \ref{fig: CRL alice} shows a 
simple example of a CRL circuit consisting exclusively of gate applications.

On the other hand, $\assOne \cond \clift(\labOne) \Rightarrow \dlvalOne$ 
represents the dynamic lifting of the bit wire $\labOne$ if the 
condition expressed by $\assOne$ is met. When we perform dynamic lifting on a 
bit, we promote its contents to a Boolean value that is bound to the lifted variable $\dlvalOne$.
This variable can then be mentioned in subsequent assignments to control 
whether further operations in the circuit are executed or not.
The introduction of $\dlvalOne$ thus naturally leads to two distinct execution branches: one in which $\dlvalOne=0$ and one in which $\dlvalOne=1$. As we do with gate applications, we write $\clift(\labOne)\Rightarrow\dlvalOne$ when a lifting operation is unconditional. A CRL circuit that performs dynamic lifting is shown in Figure \ref{fig: CRL teleportation-dl}.
		
\begin{figure}[tb]
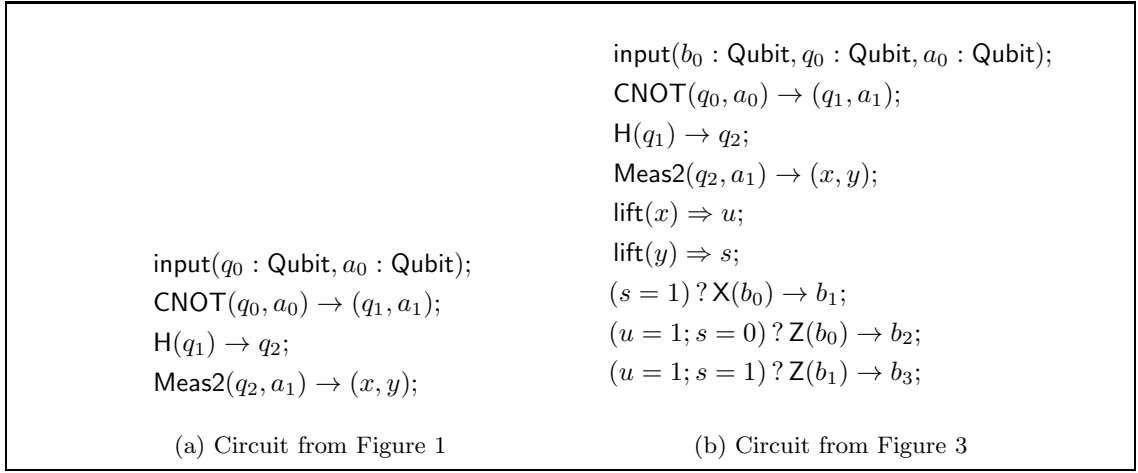

	\fbox{\begin{minipage}{.98\textwidth}\centering
			\begin{subfigure}[b]{.46\textwidth}
				\begin{equation*}
					\begin{aligned}
						&\cinput{q_0:\qubitt,a_0:\qubitt};\\
						&\gateid{CNOT}(q_0,a_0) \to (q_1,a_1);\\
						&\gateid{H}(q_1) \to q_2;\\
						&\gateid{Meas2}(q_2,a_1) \to (x,y);
					\end{aligned}
				\end{equation*}
				\caption{Circuit from Figure \protect{\ref{fig: alice}}}
				\label{fig: CRL alice}
			\end{subfigure}
			\begin{subfigure}[b]{.46\textwidth}
				\begin{equation*}
					\begin{aligned}
						&\cinput{b_0:\qubitt,q_0:\qubitt,a_0:\qubitt};\\
						&\gateid{CNOT}(q_0,a_0) \to (q_1,a_1);\\
						&\gateid{H}(q_1) \to q_2;\\
						&\gateid{Meas2}(q_2,a_1) \to (x,y);\\
						&\clift(x) \lto \dlvalOne;\\
						&\clift (y) \lto \dlvalTwo;\\
						&(\dlvalTwo = 1) \cond \gateid{X}(b_0) \to b_1;\\
						&(\dlvalOne = 1 ; \dlvalTwo = 0) \cond \gateid{Z}(b_0) \to b_2;\\
						&(\dlvalOne = 1 ; \dlvalTwo = 1)\cond \gateid{Z}(b_1) \to b_3;
					\end{aligned}
				\end{equation*}
				\caption{Circuit from Figure \protect{\ref{fig: teleportation-dl}}}
				\label{fig: CRL teleportation-dl}
			\end{subfigure}
	\end{minipage}}
	\caption{Two examples of CRL descriptions of a quantum circuit}
\end{figure}

\paragraph*{Lifting Trees} Naturally, not all CRL expressions denote reasonable circuits. For 
example, conditioning the application of a gate on the value of a lifted 
variable which has not yet been introduced should be avoided, for obvious 
reasons. Capturing this idea at the type level is non trivial, since 
the presence of a variable can itself depend on previous liftings. This is where the concept of lifting 
tree, which we introduced informally in Section~\ref{sec: circuit-building languages}, really comes into play.

\begin{definition}[Lifting Tree]
	We define the set $\dltreeset$ of \emph{lifting trees}, along with their \emph{variable set on branch $\assOne$}, seen as a function $\namesof{\assOne}:\dltreeset\to\powerset(\dlvalset)$, as the smallest set of expressions and functions such that
	\begin{itemize}
		\item $\emptytree\in\dltreeset$ with $\namesof{\assOne}(\emptytree)=\emptyset$.
		\item If $\dlvaltreeOne_0\in\dltreeset$ and $\dlvaltreeOne_1\in\dltreeset$, then for every $\dlvalOne$ not in $\namesof{\emptyassignment}(\dlvaltreeOne_0) \cup \namesof{\emptyassignment}(\dlvaltreeOne_1)$ we have $\treeNode{\dlvalOne}{\dlvaltreeOne_0}{\dlvaltreeOne_1}\in\dltreeset$ and
		\longshortversion
		%LONG
		{\begin{equation}
			\namesof{\assOne}(\treeNode{\dlvalOne}{\dlvaltreeOne_0}{\dlvaltreeOne_1})= \{\dlvalOne\} \cup \begin{cases}
				\namesof{\assOne}\left(\dlvaltreeOne_{0}\right) & \textnormal{if } \assOne(\dlvalOne)=0,\\
				\namesof{\assOne}\left(\dlvaltreeOne_{1}\right) & \textnormal{if } \assOne(\dlvalOne)=1,\\
				\namesof{\assOne}(\dlvaltreeOne_0)\cup\namesof{\assOne}(\dlvaltreeOne_1) & \textnormal{if $\assOne(\dlvalOne)$ is undefined}.
			\end{cases}
		\end{equation}}
		%SHORT
		{$\namesof{\assOne}(\treeNode{\dlvalOne}{\dlvaltreeOne_0}{\dlvaltreeOne_1})=\{\dlvalOne\}\cup\namesof{\assOne}\left(\dlvaltreeOne_{\assOne(\dlvalOne)}\right)$ if $\assOne(\dlvalOne)$ is defined, and $\{\dlvalOne\}\cup\namesof{\assOne}(\dlvaltreeOne_0)\cup\namesof{\assOne}(\dlvaltreeOne_1)$ otherwise.}
	\end{itemize}
\end{definition}
	
We often write $\namesof{}(\dlvaltreeOne)$ as shorthand for $\namesof{\emptyassignment}(\dlvaltreeOne)$ to denote all the variables mentioned in $\dlvaltreeOne$.
By way of lifting trees, we can keep track of whether an assignment, representing a condition, is consistent with the current state of the lifted variables.
Given a lifting tree $\dlvaltreeOne$, we call $\assset_\dlvaltreeOne$ the set of such assignments\longshortversion
%LONG
{.
	\begin{definition}[Assignment Set]
		Given a lifting tree $\dlvaltreeOne$, we define its \emph{assignment set}, written $\assset_\dlvaltreeOne$, as
		\begin{equation}
			\begin{aligned}
				\assset_\emptytree &= \{\emptyassignment\},\\
				\assset_{\treeNode{\dlvalOne}{\dlvaltreeOne_{0}}{\dlvaltreeOne_{1}}} &= \{\assOne,\assTwo,(\dlvalOne=0)\assmerge\assOne,(\dlvalOne=1)\assmerge\assTwo \mid \assOne\in\assset_{\dlvaltreeOne_{0}}, \assTwo\in\assset_{\dlvaltreeOne_{1}} \}.
			\end{aligned}
		\end{equation}
	\end{definition}
}
%SHORT
{ (which is easily defined by induction on $\dlvaltreeOne$ \cite{extended-ver}).}
Among these consistent assignments, there are some which are \emph{maximal}, i.e. that cannot be further extended: they describe root-to-leaf paths in $\dlvaltreeOne$ and correspond exactly to the elements of the set $\cassset_\dlvaltreeOne$ which we introduced in Section \ref{sec: circuit-building languages}.
\longversion{
	\begin{definition}[Path Set]
		Given a lifting tree $\dlvaltreeOne$, we define its \emph{path set}, written $\cassset_\dlvaltreeOne$, as
		\begin{equation}
			\begin{aligned}
				\cassset_\emptytree &= \{\emptyassignment\},\\
				\cassset_{\treeNode{\dlvalOne}{\dlvaltreeOne_{0}}{\dlvaltreeOne_{1}}} &= \{(\dlvalOne=0)\assmerge\assOne,(\dlvalOne=1)\assmerge\assTwo \mid \assOne\in\cassset_{\dlvaltreeOne_{0}}, \assTwo\in\cassset_{\dlvaltreeOne_{1}} \}.
			\end{aligned}
		\end{equation}
	\end{definition}
}

Finally, we can properly define one of the key notions of this paper, not only for 
circuits, but also for programs: given a generic set $\genericSetOne$, $\condmonad_{\dlvaltreeOne}(\genericSetOne)$ indicates the 
set $\cassset_{\dlvaltreeOne}\to\genericSetOne$, which we call the \emph{lifting} of $\genericSetOne$, and whose elements we refer to as 
\emph{lifted objects}. Despite the fact that lifted objects are formally mappings, seeing them as decorations of lifting trees whose leaves are labeled with objects from $\genericSetOne$ is perhaps more intuitive. In this perspective, given $\genericObjectOne\in\genericSetOne$, we indicate by $\treeLeaf{\genericObjectOne}$ the trivial lifted object in $\condmonad_\emptytree(\genericSetOne)$ defined as the mapping $\emptyassignment\mapsto\genericObjectOne$, and by $\treeNode{\dlvalOne}{\MongenericObjectOne_0}{\MongenericObjectOne_1}$ the object in $\condmonad_{\treeNode{\dlvalOne}{\dlvaltreeOne_0}{\dlvaltreeOne_1}}(\genericSetOne)$ defined as $\assOne\mapsto \MongenericObjectOne_{\assOne(\dlvalOne)}(\assOne')$, where $\MongenericObjectOne_0\in\condmonad_{\dlvaltreeOne_0}(\genericSetOne),\MongenericObjectOne_1\in\condmonad_{\dlvaltreeOne_1}(\genericSetOne)$ and $\assOne'\in\cassset_{\dlvaltreeOne_{\assOne(\dlvalOne)}}$ is obtained from $\assOne$ by excluding $\dlvalOne$ from its domain. A graphical representation of this intuition can be found in the form of the trees shown in figures \ref{fig: composition} and \ref{fig: flattening}.

\begin{example}
	A CRL circuit that does \textit{not} perform dynamic lifting has an output of the form $\treeLeaf{\lcOne}$ for some label context $\lcOne$. On the other hand, a CRL circuit like the one in Figure \ref{fig: CRL teleportation-dl} has output $\treeNode{\dlvalOne}{\treeNode{\dlvalTwo}{b_0:\qubitt}{b_1:\qubitt}}{\treeNode{\dlvalTwo}{b_2:\qubitt}{b_3:\qubitt}}$.
\end{example}

The same intuition informs the various operations that we define homogeneously on lifting trees and lifted objects. The first is a very natural one:
if $\dlvaltreeOne$ is a lifting tree and $\{\genericObjectOne_\assOne\}_{\assOne\in\cassset_\dlvaltreeOne}$ is 
a family of elements in $\genericSetOne$ indexed on the root-to-leaf paths of $\dlvaltreeOne$, then the expression $\Moncomp{\dlvaltreeOne}{\genericObjectOne_\assOne}{\assOne}{\cassset_\dlvaltreeOne}$ stands for the element of $\condmonad_{\dlvaltreeOne}(\genericSetOne)$ obtained by sticking each $\genericObjectOne_\assOne$ to the leaf of $\dlvaltreeOne$ identified by $\assOne$. Be aware that in general we employ this operation, which we call \emph{composition}, in a very flexible way. Namely, we generally perform composition on a \textit{subset} of leaves, in which case we write $\Moncomp{\dlvaltreeOne}{\genericObjectOne_\assOne}{\assOne}{\indexsetOne}$ for some $\indexsetOne\subseteq\cassset_\dlvaltreeOne$. When $\indexsetOne$ is a singleton $\{\assTwo\}$ we omit the subscript $\assOne$ and write $\Moncomp{\dlvaltreeOne}{\genericObjectOne}{}{\{\assTwo\}}$, whereas when $\indexsetOne=\cassset_\dlvaltreeOne$ we omit the index set and write $\Moncomp{\dlvaltreeOne}{\genericObjectOne_\assOne}{\assOne}{}$. We also allow already specified lifted objects to appear on the left of a composition, in which case we overwrite the interested branches. For instance, if $\MonlcOne$ is a lifted label context, we often write the expression $\Moncomp{\MonlcOne}{\lcOne}{}{\{\assOne\}}$ to denote the lifted label context that is $\lcOne$ on $\assOne$ and is otherwise equal to $\MonlcOne$ on all the other branches.
		
\longversion{
	\begin{definition}[Composition]
		Given $\MongenericObjectOne\in\condmonad_{\dlvaltreeOne}(\genericSetOne)$, an index set $\indexsetOne\subseteq\cassset_\dlvaltreeOne$ and a family $\{\genericObjectOne_\assOne\}_{\assOne\in\indexsetOne}$ of elements in $\genericSetOne$, we define the \emph{composition of $\MongenericObjectOne$ and $\{\genericObjectOne_\assOne\}_{\assOne\in\indexsetOne}$}, written $\Moncomp{\MongenericObjectOne}{\genericObjectOne_\assOne}{\assOne}{\indexsetOne}$, as
		\begin{equation}
			\begin{aligned}
				\Moncomp{\Monunit{\genericObjectTwo}}{\genericObjectOne}{}{\emptyset} &= \Monunit{\genericObjectTwo},\\
				\Moncomp{\Monunit{\genericObjectTwo}}{\genericObjectOne}{}{\{\emptyassignment\}} &= \Monunit{\genericObjectOne},\\
				\Moncomp{\treeNode{\dlvalOne}{\MongenericObjectOne_0}{\MongenericObjectOne_1}}{\genericObjectOne_\assOne}{\assOne}{\indexsetOne} &= \treeNode{\dlvalOne}{\Moncomp{\MongenericObjectOne_0}{\genericObjectOne_\assOne}{\assOne}{\indexsetOne_0}}{\Moncomp{\MongenericObjectOne_1}{\genericObjectOne_\assOne}{\assOne}{\indexsetOne_1}},
			\end{aligned}
		\end{equation}
		where $\indexsetOne_b=\{\assOne|_{\namesof{}(\dlvaltreeOne)\setminus\{\dlvalOne\}} \mid \assOne\in\indexsetOne \wedge \assOne(\dlvalOne)=b\}$ and $\assOne|_{\namesof{}(\dlvaltreeOne)\setminus\{\dlvalOne\}}$ denotes the exclusion of $\dlvalOne$ from the domain of $\assOne$.
	\end{definition}
}

\begin{example}
\label{ex: composition}
	Let $\dlvaltreeOne=\treeNode{\dlvalOne}{\treeNode{\dlvalTwo}{\emptytree}{\emptytree}}{\emptytree}$ and let $\MonlcOne = \treeNode{\dlvalOne}{\treeNode{\dlvalTwo}{q_0:\qubitt}{c_1:\bitt}}{c_2:\bitt} \in \condmonad_\dlvaltreeOne(\lcset)$ be a lifted label context, which graphically corresponds to the tree shown in Figure \ref{fig: composition (before)}. We have $\Moncomp{\MonlcOne}{c_0:\bitt}{}{\{\dlvalOne=0,\dlvalTwo=0\}} = \treeNode{\dlvalOne}{\treeNode{\dlvalTwo}{c_0:\bitt}{c_1:\bitt}}{c_2:\bitt} \in \condmonad_\dlvaltreeOne(\lcset)$, which corresponds to the tree shown in Figure \ref{fig: composition (after)}.
\end{example}

\begin{figure}[tb]\centering
	\fbox{\begin{minipage}{.98\textwidth}\centering
			\begin{subfigure}[c]{.4\textwidth}
			\begin{forest}
				[$\dlvalOne$,grow=east [$\dlvalTwo$,grow=east,edge label={node[midway,below,font=\scriptsize]{0}} [$q_0:\qubitt$,edge label={node[midway,below,font=\scriptsize]{0}}] [$c_1:\bitt$,edge label={node[midway,above,font=\scriptsize]{1}}]]
				[$c_2:\bitt$,edge label={node[midway,above,font=\scriptsize]{1}}]]
			\end{forest}
			\subcaption{Before composition}
			\label{fig: composition (before)}
			\end{subfigure}
			\begin{subfigure}[c]{.28\textwidth}
			\begin{forest}
				[$\dlvalOne$,grow=east [$\dlvalTwo$,grow=east,edge label={node[midway,below,font=\scriptsize]{0}} [$c_0:\bitt$,edge label={node[midway,below,font=\scriptsize]{0}}] [$c_1:\bitt$,edge label={node[midway,above,font=\scriptsize]{1}}]]
				[$c_2:\bitt$,edge label={node[midway,above,font=\scriptsize]{1}}]]
			\end{forest}
			\subcaption{After composition}
			\label{fig: composition (after)}
			\end{subfigure}
	\end{minipage}}
	\caption{An example of composition with overwriting}
	\label{fig: composition}
\end{figure}
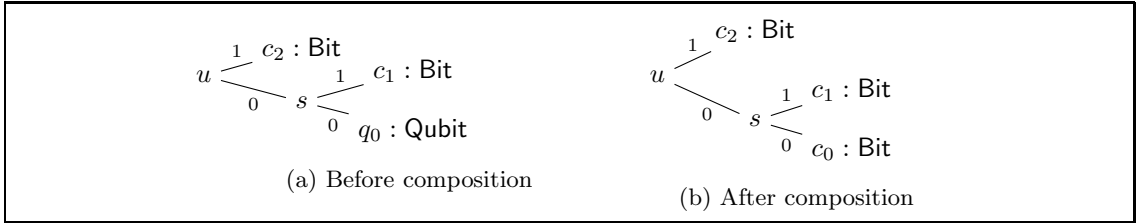
	
The second operation is a \textit{flattening} operation, which we indicate with $\Monflatten{\cdot}$. Intuitively, if we have a lifted object whose leaves are themselves trees, the flattening operation ``unfolds'' the trees in the leaves so that they become sub-trees of the original tree. Given a family of trees $\{\dlvaltreeTwo_\assOne\}_{\assOne\in\cassset_\dlvaltreeOne}$ the difference between $\Moncomp{\dlvaltreeOne}{\dlvaltreeTwo_\assOne}{\assOne}{}$ and $\Monflatten{\Moncomp{\dlvaltreeOne}{\dlvaltreeTwo_\assOne}{\assOne}{}}$ is therefore that the former is an element of $\condmonad_\dlvaltreeOne(\dltreeset)$, whereas the latter is a proper element of $\dltreeset$. This operation is well-defined when $\namesof{\assOne}(\dlvaltreeOne)\cap\namesof{}(\dlvaltreeTwo_\assOne)=\emptyset$ for every $\assOne\in\cassset_\dlvaltreeOne$. We can also flatten when we have have lifted objects on the right of a composition: if we have $\Moncomp{\dlvaltreeOne}{\MongenericObjectOne_\assOne}{\assOne}{}$, where $\MongenericObjectOne_\assOne\in\condmonad_{\dlvaltreeTwo_\assOne}(\genericSetOne)$ for every $\assOne\in\cassset_\dlvaltreeOne$, then $\Monflatten{\Moncomp{\dlvaltreeOne}{\MongenericObjectOne_\assOne}{\assOne}{}}\in\condmonad_{\Monflatten{\Moncomp{\dlvaltreeOne}{\dlvaltreeTwo_\assOne}{\assOne}{}}}(\genericSetOne)$.

\longversion{
	A formal definition of flattening requires us to consider a slightly more general operation $\Monflatten{\cdot}^\dlvalsubOne$, where $\dlvalsubOne$ is a finite set of lifted variables which are accumulated as a lifted object is traversed. At the leaf level, $\Monflatten{\Monunit{\genericObjectOne}}^\dlvalsubOne=\genericObjectOne$ only if $\genericObjectOne$ is a lifted object in which none of the lifted variables in $\dlvalsubOne$ occur. More formally:
	\begin{definition}[Flattening]
		Given $\MongenericObjectOne\in\condmonad_{\dlvaltreeOne}(\genericSetOne)$, we define the \emph{flattening of $\MongenericObjectOne$ under $\dlvalsubOne$}, written $\Monflatten{\MongenericObjectOne}^\dlvalsubOne$, as
		\begin{equation}
			\begin{aligned}
				\Monflatten{\Monunit{\genericObjectTwo}}^\dlvalsubOne &= \begin{cases}
					\genericObjectTwo & \textnormal{if } \exists\dlvaltreeTwo,\genericSetTwo.\genericObjectTwo\in\condmonad_{\dlvaltreeTwo}(\genericSetTwo) \textnormal{ and } \dlvalsubOne\cap\namesof{}(\dlvaltreeTwo)=\emptyset,\\
					\Monunit{\genericObjectTwo} & \textnormal{if } \nexists\dlvaltreeTwo,\genericSetTwo.\genericObjectTwo\in\condmonad_{\dlvaltreeTwo}(\genericSetTwo),
				\end{cases}\\
				\Monflatten{\treeNode{\dlvalOne}{\MongenericObjectOne_0}{\MongenericObjectOne_1}}^\dlvalsubOne &= \treeNode{\dlvalOne}{\Monflatten{\MongenericObjectOne_0}^{\dlvalsubOne\cup\{\dlvalOne\}}}{\Monflatten{\MongenericObjectOne_1}^{\dlvalsubOne\cup\{\dlvalOne\}}}.
			\end{aligned}
		\end{equation}
	\end{definition}

The actual flattening operation that we employ in the rest of the paper can then be defined as $\Monflatten{\cdot} = \Monflatten{\cdot}^\emptyset$.
}

\begin{example}
	\label{ex: flattening}
	Reconsider $\dlvaltreeOne$ and $\treeNode{\dlvalOne}{\treeNode{\dlvalTwo}{c_0:\bitt}{c_1:\bitt}}{c_2:\bitt}=\MonlcOne'$ from Example \ref{ex: composition}. Let $\MongenericObjectOne=\Moncomp{\MonlcOne'}{\treeNode{\dlvalTwo}{c_2:\bitt}{c_3:\bitt}}{}{\{\dlvalOne=1\}} = \treeNode{\dlvalOne}{\treeNode{\dlvalTwo}{c_0:\bitt}{c_1:\bitt}}{\Monunit{\treeNode{\dlvalTwo}{c_2:\bitt}{c_3:\bitt}}} \in \condmonad_\dlvaltreeOne(\lcset\cup\condmonad_{\treeNode{\dlvalTwo}{\emptytree}{\emptytree}}(\lcset))$. Note that this object, corresponding to Figure \ref{fig: flattening (before)}, is \textit{not} a lifted label context, as one of its leaves is itself a lifted label context. Because $\dlvalTwo$ does not occur in $\namesof{\dlvalOne=1}(\dlvaltreeOne)=\{\dlvalOne\}$, we can write $\Monflatten{\MongenericObjectOne}= \treeNode{\dlvalOne}{\treeNode{\dlvalTwo}{c_0:\bitt}{c_1:\bitt}}{\treeNode{\dlvalTwo}{c_2:\bitt}{c_3:\bitt}} \in \condmonad_{\dlvaltreeOne'}(\lcset)$, for $\dlvaltreeOne'=\treeNode{\dlvalOne}{\treeNode{\dlvalTwo}{\emptytree}{\emptytree}}{\treeNode{\dlvalTwo}{\emptytree}{\emptytree}}$. \textit{This} is an actual lifted label context, which corresponds to Figure \ref{fig: flattening (after)}.
\end{example}

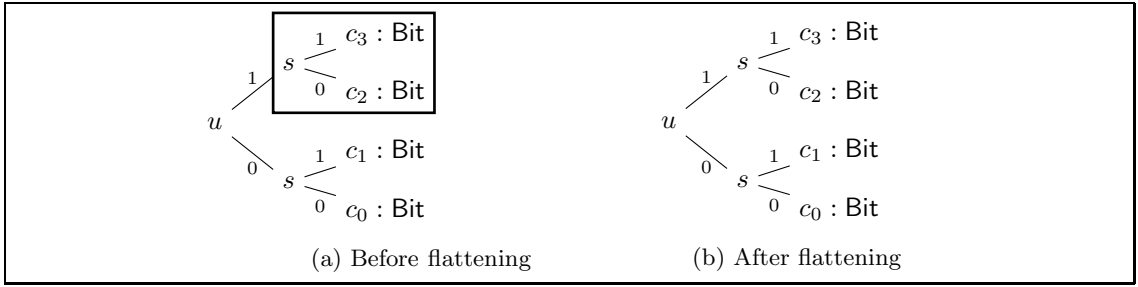
\begin{figure}[tb]\centering
	\fbox{\begin{minipage}{.98\textwidth}\centering
			\begin{subfigure}[c]{.4\textwidth}
				\begin{forest}
					[$\dlvalOne$,grow=east [$\dlvalTwo$,grow=east,edge label={node[midway,below,font=\scriptsize]{0}} [$c_0:\bitt$,edge label={node[midway,below,font=\scriptsize]{0}}] [$c_1:\bitt$,edge label={node[midway,above,font=\scriptsize]{1}}]]
					[$\dlvalTwo$,grow=east,edge label={node[midway,above,font=\scriptsize]{1}}, tikz={\node [draw, line width=1pt, inner sep=0,fit to=tree] {};} [$c_2:\bitt$,edge label={node[midway,below,font=\scriptsize]{0}}] [$c_3:\bitt$,edge label={node[midway,above,font=\scriptsize]{1}}]]]
				\end{forest}
				\subcaption{Before flattening}
				\label{fig: flattening (before)}
			\end{subfigure}
			\begin{subfigure}[c]{.26\textwidth}
				\begin{forest}
					[$\dlvalOne$,grow=east [$\dlvalTwo$,grow=east,edge label={node[midway,below,font=\scriptsize]{0}} [$c_0:\bitt$,edge label={node[midway,below,font=\scriptsize]{0}}] [$c_1:\bitt$,edge label={node[midway,above,font=\scriptsize]{1}}]]
					[$\dlvalTwo$,grow=east,edge label={node[midway,above,font=\scriptsize]{1}} [$c_2:\bitt$,edge label={node[midway,below,font=\scriptsize]{0}}] [$c_3:\bitt$,edge label={node[midway,above,font=\scriptsize]{1}}]]]
				\end{forest}
				\subcaption{After flattening}
				\label{fig: flattening (after)}
			\end{subfigure}
	\end{minipage}}
	\caption{An example of flattening}
	\label{fig: flattening}
\end{figure}

In conjunction, these two operations greatly simplify our treatment of dynamic lifting, as they allow us to describe in detail the desired shape of a lifted object. For example, in later sections we often write $\MonlcOne=\Monflatten{\Moncomp{\MonlcOne'}{\MonlcOne'',\MonlcTwo}{}{\{\assOne\}}}$ to say that the lifted label context $\MonlcOne$ is such that if we start from its root and follow the path described by $\assOne$, we find a sub-tree $\MonlcOne'',\MonlcTwo$ (the disjoint union is lifted in the natural way, as the branch-wise disjoint union of the leaves of $\MonlcOne''$ and $\MonlcTwo$) and that we are interested in only \textit{some} elements of the corresponding lifted label context, for instance those that occur in $\MonlcTwo$ and not in $\MonlcOne''$.

\paragraph*{A Formal System for Circuit Signatures}

Now that we have introduced lifting trees and lifted object as a means to reason about dynamic lifting, we are ready to introduce the notion of \emph{signature} of a circuit.
\begin{definition}[Circuit Signature]
	Given a circuit $\circuitOne$, a lifting tree $\dlvaltreeOne$, a label context $\lcOne$ and a lifted label context $\MonlcOne$, a \emph{circuit signature} ia a judgment of the form
	$
	\circjudgment{\circuitOne}{\dlvaltreeOne}{\lcOne}{\MonlcOne}.
	$
\end{definition}

A circuit signature is valid when it can be derived by the rules in Figure \ref{table: circ rules}. Note that $\MonlcOne\Monmerge{\assOne}\lcOne'$ is shorthand for $\Moncomp{\MonlcOne'}{\lcOne',\lcOne_{\assTwo}}{\assTwo}{\cassset_\dlvaltreeOne^\assOne}$ if $\MonlcOne=\Moncomp{\MonlcOne'}{\lcOne_\assTwo}{\assTwo}{\cassset_\dlvaltreeOne^\assOne}$, where $\cassset_\dlvaltreeOne^\assOne$ contains the paths in $\cassset_\dlvaltreeOne$ which extend $\assOne$.\longversion{ In other words, $\cassset_\dlvaltreeOne^\assOne$ contains all those paths in $\cassset_\dlvaltreeOne$ which are defined and equal to $\assOne$ on $\dom(\assOne)$.} On the other hand, for any $\MongenericObjectOne\in\condmonad_{\dlvaltreeOne}(\genericSetOne)$ and $\dlvaltreeTwo\in\dltreeset$ such that $\namesof{\assOne}(\dlvaltreeOne)\cap\namesof{}(\dlvaltreeTwo)=\emptyset$, we write $\MongenericObjectOne\Monsplit{\assOne}\dlvaltreeTwo$ as shorthand for $\Monflatten{\Moncomp{\MongenericObjectOne'}{\Moncomp{\dlvaltreeTwo}{\genericObjectOne_\assTwo}{\assThree}{}}{\assTwo}{\cassset_\dlvaltreeOne^\assOne}}$, if  $\MongenericObjectOne=\Moncomp{\MongenericObjectOne'}{\genericObjectOne_\assTwo}{\assTwo}{\cassset_\dlvaltreeOne^\assOne}$.

\begin{figure}[tb]
	\centering
	\fbox{\begin{minipage}{.98\textwidth}
		$$
		\inference[id]{}
		{\circjudgment{\cinput{\lcOne}}{\emptytree}{\lcOne}{\Monunit{\lcOne}}}
		\qquad
		\inference[lift]{
			\circjudgment{\circuitOne}{\dlvaltreeOne}{\lcOne}{\MonlcOne\Monmerge{\assOne}\labOne:\bitt}
			&
			\assOne\in\assset_{\dlvaltreeOne}
			&
			\dlvalOne\notin\namesof{\assOne}(\dlvaltreeOne)
		}
		{\circjudgment{\circuitOne;\assOne \cond \clift(\labOne)\Rightarrow\dlvalOne}{\dlvaltreeOne\Monsplit{\assOne}\treeNode{\dlvalOne}{\emptytree}{\emptytree}}{\lcOne}{\MonlcOne\Monsplit{\assOne}\treeNode{\dlvalOne}{\emptytree}{\emptytree}}}
		$$
		\rulespace
		$$
		\inference[gate]
		{\circjudgment{\circuitOne}{\dlvaltreeOne}{\lcOne}{\MonlcOne\Monmerge{\assOne}\lcOne'}
			&
			\assOne\in\assset_{\dlvaltreeOne}
			&
			\gateOne\in\gateset(\mtypeOne,\mtypeTwo)
			&
			\mjudgment{\lcOne'}{\vec\labOne}{\mtypeOne}
			&
			\mjudgment{\lcTwo}{\vec\labTwo}{\mtypeTwo}
			&
			\labelfresh(\vec\labTwo,\circuitOne)}
		%----
		{\circjudgment{\circuitOne;\assOne \cond \gateOne(\vec\labOne) \to \vec\labTwo}{\dlvaltreeOne}{\lcOne}{\MonlcOne\Monmerge{\assOne}\lcTwo}}
		$$
	\end{minipage}}
	\caption{The rules for CRL circuit signatures}
	\label{table: circ rules}
\end{figure}
	
%%%%%%%%%%%%%%%%%%%%%		
\section{\PQK}
%%%%%%%%%%%%%%%%%%%%%

We are finally ready to introduce \PQK, a programming language meant exactly to manipulate the kind of circuits that we presented in the previous section.
		
\subsection{Types and Terms}
The types and syntax of \PQK\ are given in Figure \ref{table: types and syntax}, where $\varOne$ and $\varTwo$ range over the set of regular variable names, $\dlvalOne_1,\dots,\dlvalOne_n$ over the set $\dlvalset$ of lifted variable names, $\dlvaltreeOne$ over the set $\dltreeset$ of lifting trees and Greek letters generally indicate lifted objects. More precisely, $\MontypeOne,\MontypeTwo\in\condmonad_{\dlvaltreeOne}(\typeset),\MonmtypeTwo\in\condmonad_{\dlvaltreeOne}(\mtypeset),\MontermOne\in\condmonad_\dlvaltreeOne(\termset)$ and $\MonmvalOne\in\condmonad_{\dlvaltreeOne}(\mvalset)$, each for  some $\dlvaltreeOne$.
Note that the M-values and M-types that we introduced in Section \ref{sec: generalized circuits} are now a proper subset of \PQK's values and types, respectively.
		
A value of the form $\boxedcirc{\vec\labOne}{\circuitOne}{\MonmvalOne}{\dlvaltreeOne}$ is what we called a \textit{boxed circuit} in Section \ref{sec: circuit-building languages}, that is, a datum representation of a circuit $\circuitOne$ that takes as input the labels in the label tuple $\vec\labOne$ and outputs one of the possible label tuples in $\MonmvalOne$ depending on the lifted variables in $\dlvaltreeOne$. Correspondingly, types of the form $\circt{\dlvaltreeOne}(\mtypeOne,\MonmtypeTwo)$ are called \emph{circuit types} and represent boxed circuits. Both boxed circuits and circuit types abstract over the lifted variables in $\namesof{}(\dlvaltreeOne)$ and thus enjoy a notion of $\alpha$-equivalence which we consider them to be identical under. The $\justboxt$ and $\apply{}$ constructs are those described in Section \ref{sec: circuit-building languages} and they respectively introduce and consume boxed circuits. Specifically, the programmer is never expected to write values of the form $\boxedcirc{\vec\labOne}{\circuitOne}{\MonmvalOne}{\dlvaltreeOne}$ by hand, but is rather  expected to introduce the desired circuit by supplying an appropriate circuit-building term $\termOne$ to the $\justboxt$ operator, obtaining a term of the form $\boxt{\mtypeOne}{}(\lift\termOne)$. The use of $\lift$ guarantees that $\termOne$ does not make use of any linear resources from the current environment, i.e. that it can be safely evaluated in a sandboxed environment to produce $\circuitOne$.
After a boxed circuit $\boxedcirc{\vec\labOne}{\circuitOne}{\MonmvalOne}{\dlvaltreeOne}$ is introduced, the programmer can potentially copy it and apply it to the underlying circuit $\circuitTwo$ via a term of the form $\apply{\dlvalOne_1,\dots,\dlvalOne_n}(\boxedcirc{\vec\labOne}{\circuitOne}{\MonmvalOne}{\dlvaltreeOne},\vec\labTwo)$. Such a term ``unboxes'' $\circuitOne$, finds the wires identified by $\vec\labTwo$ among the outputs of $\circuitTwo$ and appends $\circuitOne$ to them. In this process, any lifted variables in $\namesof{}(\dlvaltreeOne)$, which were abstracted in $\boxedcirc{\vec\labOne}{\circuitOne}{\MonmvalOne}{\dlvaltreeOne}$, have to be instantiated with concrete names. To this end, the programmer supplies the $n=|\namesof{}(\dlvaltreeOne)|$ lifted variables $\dlvalOne_1,\dots,\dlvalOne_n$, which are expected to be fresh.
		
\begin{figure}[tb]
	\centering
	\fbox{\begin{minipage}{.98\textwidth}\centering
		\begin{tabular}{l l l}
			Types & $\typeOne,\typeTwo$ & $::=  \unitt \mid \wtypeOne \mid \typeOne \multimap_\dlvaltreeOne \MontypeTwo \mid \bang \MontypeOne \mid \circt{\dlvaltreeOne}(\mtypeOne, \MonmtypeTwo) \mid \typeOne \otimes \typeTwo.$\\
			
			Parameter Types & $\ptypeOne,\ptypeTwo$ & $::= \unitt \mid \bang \MontypeOne \mid \circt{\dlvaltreeOne}(\mtypeOne, \MonmtypeTwo) \mid \ptypeOne \otimes \ptypeTwo.$\\
			
			M-types\quad & $\mtypeOne,\mtypeTwo$ & $::= \unitt \mid \wtypeOne \mid \mtypeOne \otimes \mtypeTwo.$\\
			
			&&\\
			
			Terms & $\termOne,\termTwo$ & $::= \valOne\valTwo \mid \letin{\varOne}{\termOne}{\MontermOne} \mid \letin{\tuple{\varOne,\varTwo}}{\valOne}{\termOne}$\\
			&& $\mid \force \valOne \mid \boxt{\mtypeOne}{\dlvalsubOne} \valOne \mid \apply{\dlvalOne_1,\dots,\dlvalOne_n}(\valOne, \valTwo) \mid \return \valOne.$\\
			
			Values & $\valOne,\valTwo$ & $::= \unitv \mid \varOne \mid \labOne \mid \lambda \varOne_\typeOne.\termOne \mid \lift \termOne \mid \boxedcirc{\vec{\labOne}}{\circuitOne}{\MonmvalOne}{\dlvaltreeOne} \mid \tuple{\valOne,\valTwo}.$\\
			
			M-values & $\vec{\labOne},\vec{\labTwo}$ & $::= \unitv \mid \labOne \mid \tuple{\vec{\labOne},\vec{\labTwo}}$.
		\end{tabular}					
	\end{minipage}}
	\caption{Types and terms of \PQK}
	\label{table: types and syntax}
\end{figure}

Notice that there exists a strong distinction between values and terms, the latter representing effectful computations that can introduce new lifted variables as a consequence of dynamic lifting. This choice does not detract from the expressiveness of the language, since first and foremost a value $\valOne$ can always be turned into an effectless computation $\return \valOne$. In the other direction, we have that terms such as $\termOne\termTwo$ can still be written in \PQK\ as $\letin{\varOne}{\termOne}{\{\letin{\varTwo}{\termTwo}{\{\varOne\varTwo\}}\}}$. The $\mathsf{let}$ construct is in fact a central construct in \PQK, as on top of serving as a sequencing operator it also doubles as a conditional statement. When we evaluate the term $\letin{\varOne}{\termOne}{\MontermOne}$, we first carry out the computation described by $\termOne$, which performs dynamic lifting according to some lifting tree $\dlvaltreeOne$ and consequently results in a lifted value $\MonvalOne\in\condmonad_{\dlvaltreeOne}(\valset)$. At this point, we are not limited to passing each and every possible value of $\MonvalOne$ to the \textit{same} continuation. Rather, we can define a different continuation for every possible outcome of the liftings in $\dlvaltreeOne$. To this effect, the programmer supplies a lifted term $\MontermOne\in\condmonad_{\dlvaltreeOne}(\termset)$, which matches $\MonvalOne$'s lifting tree and thus effectively provides such a roster of continuations. The following example is meant to help convey the role of $\mathsf{let}$ as a control flow operator.
		
\begin{example}
\label{ex: let}
	Imagine we wanted to measure qubit $\labOne$, dynamically lift its value into a variable $\dlvalOne$ and then apply $\gateid{H}$ to a second qubit $\labTwo$ \textit{only if} $u=1$. Suppose we had CRL definitions $\measlift=\cinput{\labOne:\qubitt};\gateid{Meas}(\labOne)\to\labOne';\clift(\labOne')\lto\dlvalOne$ and $H=\cinput{\labOne:\qubitt};\gateid{H}(\labOne)\to\labOne'$, corresponding to the circuit that measures and then lifts a qubit and the circuit that just applies the Hadamard gate to its input, respectively. We would write the following \PQK\ program:
	\begin{equation}
		\begin{aligned}
			&\letin{\_}{\apply{\dlvalOne}(\boxedcirc{\labOne}{\measlift}{\treeNode{\dlvalOne}{\unitv}{\unitv}}{\dlvalOne},\labOne)}{}\\
			&\quad\treeNode{\dlvalOne}{\return\labTwo}{\apply{}(\boxedcirc{\labOne}{H}{\Monunit{\labOne'}}{},\labTwo)},
		\end{aligned}
	\end{equation}
	which appears much more familiar under the following syntactic sugar:
	\begin{equation}
		\begin{aligned}
			&\letin{\_}{\apply{\dlvalOne}(\boxedcirc{\labOne}{\measlift}{\treeNode{\dlvalOne}{\unitv}{\unitv}}{\dlvalOne},\labOne)}{}\\
			&\quad\when{\dlvalOne=1}{\apply{}(\boxedcirc{\labOne}{H}{\Monunit{\labOne'}}{},\labTwo)}.
		\end{aligned}
	\end{equation}
\end{example}

In light of this example, we can see how the circuit shown in Figure \ref{fig: one-way} can be described in \PQK\ through a program such as the one shown in Figure \ref{fig: pqk one-way}. To conclude this section, recall that we mentioned earlier that whenever we apply a boxed circuit we need to instantiate its lifted variables with concrete names. This process is formalized through the lifted variable analog of substitution, which we call \emph{renaming}.

\begin{figure}[tb]
	\fbox{\begin{minipage}{.98\textwidth}
			\vspace{-1em}
			\begin{align*}
				\lambda q_{\qubitt}.\return\lambda a_{\qubitt}.
				&\letin{q}{\apply{}(\boxedcirc{\labOne}{H}{\Monunit{\labOne'}}{},q)}{\Monunit{\\
					&\letin{\_}{\apply{u}(\boxedcirc{\labOne}{\measlift}{\treeNode{\dlvalOne}{\unitv}{\unitv}}{\dlvalOne},q)}{\\
						&\treeNode{u}{\return a}{\apply{}(\boxedcirc{\labOne}{\mathit{Meas}}{\Monunit{\labOne'}}{},a)}}}}
			\end{align*}
	\end{minipage}}
	\caption{A \PQK\ program describing the circuit shown in Figure \protect\ref{fig: one-way}}
	\label{fig: pqk one-way}
\end{figure}

\begin{definition}[Renaming of Lifted Variables]
	Given a lifted variable-bearing object $\genericObjectOne$ and a permutation $\alpharenamingOne$ of $\dlvalset$, we call $\alpharenamingOne$ a \emph{renaming of lifted variables} and we define $\alpharenamed{\genericObjectOne}{\alpharenamingOne}$ as $\genericObjectOne$ in which every occurrence of a lifted variable $\dlvalOne$ is replaced by $\alpharenamingOne(\dlvalOne)$.\longversion{
		In our case:
		
		\vspace{1em}
		\begin{tabular}{l l l}
			Assignments & $\assOne$ & $\alpharenamed{\assOne}{\alpharenamingOne} = \assOne \circ \invert{\alpharenamingOne}.$\\
			&&\\
			Circuits & $\circuitOne$ & $\alpharenamed{\cinput{\lcOne}}{\alpharenamingOne} = \cinput{\lcOne},$\\
			&& $\alpharenamed{(\circuitOne';\assOne\cond\gateOne(\vec\labOne)\to\vec\labTwo)}{\alpharenamingOne} = \alpharenamed{\circuitOne'}{\alpharenamingOne};\alpharenamed{\assOne}{\alpharenamingOne}\cond\gateOne(\vec\labOne)\to\vec\labTwo,
			$\\
			&& $\alpharenamed{(\circuitOne';\assOne\cond\clift(\labOne)\lto\dlvalOne)}{\alpharenamingOne} = \alpharenamed{\circuitOne'}{\alpharenamingOne};\alpharenamed{\assOne}{\alpharenamingOne}\cond\clift(\labOne)\to\alpharenamingOne(\dlvalOne).$\\
			&&\\
			Lifted objects & $\MongenericObjectOne$ & $\alpharenamed{\treeLeaf{\genericObjectOne}}{\alpharenamingOne}=\treeLeaf{\genericObjectOne},$\\
			&& $\alpharenamed{(\treeNode{\dlvalOne}{\MongenericObjectOne_1}{\MongenericObjectOne_2})}{\alpharenamingOne} = \treeNode{\alpharenamingOne(\dlvalOne)}{\alpharenamed{\MongenericObjectOne_1}{\alpharenamingOne}}{\alpharenamed{\MongenericObjectOne_2}{\alpharenamingOne}}.$\\
			&&\\
		\end{tabular}}

		We denote by $\dlvalTwo_1/\dlvalOne_1,\dots,\dlvalTwo_n/\dlvalOne_n$ a permutation that exchanges $\dlvalTwo_1$ for $\dlvalOne_1,\dots\dlvalTwo_n$ for $\dlvalOne_n$.
\end{definition}

\subsection{\PQK's Typing Rules}
\label{sec: type system}

At its core, \PQK's type system is a linear type system which distinguishes between \emph{computations} (i.e. terms), which can be effectful, and \emph{values}. We therefore introduce two kinds of typing judgments. One for terms, which are given lifted types, and one for values, which have regular types.

\longshortversion
%LONG
{
	\begin{definition}[Computational Typing Judgment]
		When a typing context $\contextOne$ and a label context $\lcOne$ turn $\termOne$ into a term of lifted type $\MontypeOne\in\condmonad_{\dlvaltreeOne}(\typeset)$ depending on the lifted variables in $\dlvaltreeOne$, we write
		\begin{equation}
		\compjudgment{\contextOne}{\lcOne}{\dlvaltreeOne}{\termOne}{\MontypeOne}.
	\end{equation}
	\end{definition}
	\begin{definition}[Value Typing Judgments]
		When a typing context $\contextOne$ and a label context $\lcOne$ turn $\valOne$ into a value of type $\typeOne$, we write
		\begin{equation}
		\valjudgment{\contextOne}{\lcOne}{\valOne}{\typeOne}.
		\end{equation}
	\end{definition}
}
%SHORT
{
	\begin{definition}[Typing Judgments]
		When a typing context $\contextOne$ and a label context $\lcOne$ turn $\termOne$ into a term of lifted type $\MontypeOne$ depending on the lifted variables in $\dlvaltreeOne$, we write
		$
		\compjudgment{\contextOne}{\lcOne}{\dlvaltreeOne}{\termOne}{\MontypeOne}
		$.
		When $\contextOne$ and $\lcOne$ turn $\valOne$ into a value of type $\typeOne$, we write
		$
		\valjudgment{\contextOne}{\lcOne}{\valOne}{\typeOne}.
		$
	\end{definition}
}

When a typing context contains \textit{exclusively} variables with parameter types, we write it as $\pcontextOne$, but in general a typing context $\contextOne$ can contain both linear and parameter variables. Typing judgments are derived via the rules in Figure \ref{fig: type rules}, where $\dlvalTwo_1,\dots,\dlvalTwo_n$ range over $\dlvalset$ and the relational symbols $\Vdash_c$ and $\Vdash_v$ denote the intuitive lifting of the computation and value typing judgments to some tree $\dlvaltreeOne$\longshortversion
%LONG
{.
	\begin{definition}[Lifting of Typing Judgments]
		Given a lifting tree $\dlvaltreeOne$, if for all $\assOne\in\cassset_{\dlvaltreeOne}$ we have $\compjudgment{\contextOne_\assOne}{\lcOne_\assOne}{\dlvaltreeTwo_\assOne}{\termOne_\assOne}{\MontypeOne_\assOne}$, then we write \begin{equation}
			\Moncompjudgment{\Moncomp{\dlvaltreeOne}{\contextOne_\assOne}{\assOne}{}}{\Moncomp{\dlvaltreeOne}{\lcOne_\assOne}{\assOne}{}}{\Moncomp{\dlvaltreeOne}{\dlvaltreeTwo_\assOne}{\assOne}{}}{\Moncomp{\dlvaltreeOne}{\termOne_\assOne}{\assOne}{}}{\Moncomp{\dlvaltreeOne}{\MontypeOne_\assOne}{\assOne}{}}.
		\end{equation}
		If for all $\assOne\in\cassset_{\dlvaltreeOne}$ we have $\valjudgment{\contextOne_\assOne}{\lcOne_\assOne}{\valOne_\assOne}{\typeOne_\assOne}$, then we write
		\begin{equation}
			\Monvaljudgment{\Moncomp{\dlvaltreeOne}{\contextOne_\assOne}{\assOne}{}}{\Moncomp{\dlvaltreeOne}{\lcOne_\assOne}{\assOne}{}}{\dlvaltreeOne}{\Moncomp{\dlvaltreeOne}{\valOne_\assOne}{\assOne}{}}{\Moncomp{\dlvaltreeOne}{\typeOne_\assOne}{\assOne}{}}.
		\end{equation}
	\end{definition}
}
%SHORT
{ \cite{extended-ver}.}
\longshortversion{For convenience, within}{Within} such judgments, every lifted object (e.g. $\MonlcOne$ in $\Monvaljudgment{\emptycontext}{\MonlcOne}{\dlvaltreeOne}{\MonmvalOne}{\MonmtypeTwo}$ in the \textit{circ} rule) is assumed to be backed by $\dlvaltreeOne$, whereas every component which is \textit{not} a lifted object (e.g. $\pcontextOne$ or $\varOne$ in $\Moncompjudgment{\pcontextOne,\contextOne_2,\varOne:\MontypeOne}{\lcOne_2}{\Moncomp{\dlvaltreeOne}{\dlvaltreeTwo_\assOne}{\assOne}{}}{\MontermOne}{\MongenericObjectTwo}$ in the \textit{let} rule) is assumed to be constant across the branches of $\dlvaltreeOne$.
Note that we assume that $\dlvalset$ is totally ordered, so that whenever we write $\namesof{}(\dlvaltreeOne)=\{\dlvalOne_1,\dots,\dlvalOne_n\}$ (e.g. in the \emph{apply} rule) we have $\dlvalOne_1\leq_\dlvalset\dlvalOne_2\leq_\dlvalset\dots\leq_\dlvalset\dlvalOne_n$. An example of a type derivation that leverages the full expressiveness of \PQK's type system is the one for the program shown in Figure \ref{fig: pqk one-way}, which can be found in Appendix \ref{app: type derivations}. \longversion{
	
}To conclude this section, notice how just like M-values and M-types are a subset of the values and types of \PQK, the type system for M-values is in a one-to-one correspondence with a subset of \PQK's type system. \longversion{This correspondence is stated in Proposition \ref{prop: m-judgment and val-judgment isomorphism}.}

\longversion{
	
	\paragraph*{}\begin{proposition}\label{prop: m-judgment and val-judgment}
		\label{prop: m-judgment and val-judgment isomorphism}
		The following hold:\begin{enumerate}
			\item If $\mjudgment{\lcOne}{\vec\labOne}{\mtypeOne}$ then for every $\pcontextOne$ we have $\valjudgment{\pcontextOne}{\lcOne}{\vec\labOne}{\mtypeOne}.$
			\item If there exists $\pcontextOne$ such that $\valjudgment{\pcontextOne}{\lcOne}{\vec\labOne}{\mtypeOne}$, then $\mjudgment{\lcOne}{\vec\labOne}{\mtypeOne}$.
		\end{enumerate}
	\end{proposition}
	
	In addition to that, we also have that M-values are the \textit{only} closed values that can be given an M-type.
	
	\begin{proposition}\label{lem: m-type implies m-val}
		If $\valjudgment{\emptycontext}{\lcOne}{\valOne}{\mtypeOne}$ for some M-type $\mtypeOne$, then $\valOne$ is an M-value.
	\end{proposition}
	
}

\begin{figure}[tb]
	\centering
	\fbox{\begin{minipage}{.98\textwidth}
			
			$$
			\inference[unit]
			{}
			%---------------------------------------------
			{\valjudgment{\pcontextOne}{\emptylc}{\unitv}{\unitt}}
			\qquad
			\inference[var]
			{}
			%---------------------------------------------
			{\valjudgment{\pcontextOne,\varOne:\typeOne}{\emptylc}{\varOne}{\typeOne}}
			\qquad
			\inference[label]
			{}
			%---------------------------------------------
			{\valjudgment{\pcontextOne}{\labOne:\wtypeOne}{\labOne}{\wtypeOne}}
			$$
			\rulespace
			$$
			\inference[abs]
			{\compjudgment{\contextOne,\varOne:\typeOne}{\lcOne}{\dlvaltreeOne}{\termOne}{\MontypeTwo}}
			%---------------------------------------------
			{\valjudgment{\contextOne}{\lcOne}{\lambda \varOne_\typeOne.\termOne}{\typeOne \multimap_\dlvaltreeOne \MontypeTwo}}
			\qquad
			\inference[app]
			{\valjudgment{\pcontextOne,\contextOne_1}{\lcOne_1}{\valOne}{\typeOne \multimap_\dlvaltreeOne \MontypeTwo}
				&
				\valjudgment{\pcontextOne,\contextOne_2}{\lcOne_2}{\valTwo}{\typeOne}}
			%---------------------------------------------
			{\compjudgment{\pcontextOne,\contextOne_1,\contextOne_2}{\lcOne_1,\lcOne_2}{\dlvaltreeOne}{\valOne\valTwo}{\MontypeTwo}}
			$$
			\rulespace
			$$
			\inference[let]
			{\compjudgment{\pcontextOne,\contextOne_1}{\lcOne_1}{\dlvaltreeOne}{\termOne}{\MontypeOne}
				&
				\MontermOne\in\condmonad_{\dlvaltreeOne}(\termset)
				&
				\Moncompjudgment{\pcontextOne,\contextOne_2,\varOne:\MontypeOne}{\lcOne_2}{\Moncomp{\dlvaltreeOne}{\dlvaltreeTwo_\assOne}{\assOne}{}}{\MontermOne}{\MongenericObjectTwo}
			}
			%---------------------------------------------
			{\compjudgment{\pcontextOne,\contextOne_1,\contextOne_2}{\lcOne_1,\lcOne_2}{\Monflatten{\Moncomp{\dlvaltreeOne}{\dlvaltreeTwo_\assOne}{\assOne}{}}}{\letin{\varOne}{\termOne}{\MontermOne}}{\Monflatten{\MongenericObjectTwo}}}
			$$
			\rulespace
			$$
			\inference[tuple]
			{\valjudgment{\pcontextOne,\contextOne_1}{\lcOne_1}{\valOne}{\typeOne}
				&
				\valjudgment{\pcontextOne,\contextOne_2}{\lcOne_2}{\valTwo}{\typeTwo}}
			%---------------------------------------------
			{\valjudgment{\pcontextOne,\contextOne_1,\contextOne_2}{\lcOne_1,\lcOne_2}{\tuple{\valOne,\valTwo}}{\typeOne \otimes \typeTwo}}
			$$
			\rulespace
			$$
			\inference[dest]
			{\valjudgment{\pcontextOne,\contextOne_1}{\lcOne_1}{\valOne}{\typeOne \otimes \typeTwo}
				&
				\compjudgment{\pcontextOne,\contextOne_2,\varOne:\typeOne,\varTwo:\typeTwo}{\lcOne_2}{\dlvaltreeOne}{\termOne}{\MontypeOne}}
			%---------------------------------------------
			{\compjudgment{\pcontextOne,\contextOne_1,\contextOne_2}{\lcOne_1,\lcOne_2}{\dlvaltreeOne}{\letin{\tuple{\varOne,\varTwo}}{\valOne}{\termOne}}{\MontypeOne}}
			$$
			\rulespace
			$$
			\inference[lift]
			{\compjudgment{\pcontextOne}{\emptylc}{\emptytree}{\termOne}{\MontypeOne}}
			%---------------------------------------------
			{\valjudgment{\pcontextOne}{\emptylc}{\lift \termOne}{\bang \MontypeOne}}
			\qquad
			\inference[force]
			{\valjudgment{\contextOne}{\lcOne}{\valOne}{\bang \MontypeOne}}
			%---------------------------------------------
			{\compjudgment{\contextOne}{\lcOne}{\emptytree}{\force \valOne}{\MontypeOne}}
			\qquad
			\inference[box]
			{\valjudgment{\contextOne}{\lcOne}{\valOne}{\bang\Monunit{\mtypeOne \multimap_\dlvaltreeOne \MonmtypeTwo}}}
			%---------------------------------------------
			{\compjudgment{\contextOne}{\lcOne}{\emptytree}{\boxt{\mtypeOne}{\dlvalsubOne} \valOne}{\Monunit{\circt{\dlvaltreeOne}(\mtypeOne,\MonmtypeTwo)}}}
			$$
			\rulespace
			$$
			\inference[apply]
			{\valjudgment{\pcontextOne,\contextOne_1}{\lcOne_1}{\valOne}{\circt{\dlvaltreeOne}(\mtypeOne,\MonmtypeTwo)}
				&
				\valjudgment{\pcontextOne,\contextOne_2}{\lcOne_2}{\valTwo}{\mtypeOne}
				\\
				\namesof{}(\dlvaltreeOne)=\{\dlvalOne_1,\dots,\dlvalOne_n\}
				&
				\alpharenamingOne=\dlvalTwo_1/\dlvalOne_1,\dots,\dlvalTwo_n/\dlvalOne_n}
			%---------------------------------------------
			{\compjudgment{\pcontextOne,\contextOne_1,\contextOne_2}{\lcOne_1,\lcOne_2}{\alpharenamed{\dlvaltreeOne}{\alpharenamingOne}}{\apply{\dlvalTwo_1,\dots,\dlvalTwo_n}(\valOne,\valTwo)}{\alpharenamed{\MonmtypeTwo}{\alpharenamingOne}}}
			$$
			\rulespace
			$$
			\inference[circ]
			{\circjudgment{\circuitOne}{\dlvaltreeOne}{\lcOne}{\MonlcOne}
				&
				\valjudgment{\emptyset}{\lcOne}{\vec{\labOne}}{\mtypeOne}
				&
				\Monvaljudgment{\emptycontext}{\MonlcOne}{\dlvaltreeOne}{\MonmvalOne}{\MonmtypeTwo}
			}
			%---------------------------------------------
			{\valjudgment{\pcontextOne}{\emptylc}{\boxedcirc{\vec{\labOne}}{\circuitOne}{\MonmvalOne}{\dlvaltreeOne}}{\circt{\dlvaltreeOne}(\mtypeOne,\MonmtypeTwo)}}
			\qquad
			\inference[return]
			{\valjudgment{\contextOne}{\lcOne}{\valOne}{\typeOne}}
			%---------------------------------------------
			{\compjudgment{\contextOne}{\lcOne}{\emptytree}{\return \valOne}{\Monunit{\typeOne}}}
			$$
	\end{minipage}}
	\caption{The typing rules of \PQK}
	\label{fig: type rules}
\end{figure}

\subsection{A Big-Step Operational Semantics}
\label{sec: operational semantics}
The big-step operational semantics of the language is based on an evaluation relation $\bscfgl{\circuitOne}{\assOne}{\termOne}\bseval\bscfgr{\circuitTwo}{\MonvalOne}$, where $\MonvalOne$ is a lifted value. This means that the term $\termOne$ evaluates to one of the possible values in $\MonvalOne$, depending on the outcome of intermediate measurements, and updates the underlying circuit $C$ upon branch $\assOne$ as a side effect, obtaining an updated circuit $\circuitTwo$. We call $\bscfgl{\circuitOne}{\assOne}{\termOne}$ a \textit{left configuration} and $\bscfgr{\circuitTwo}{\MonvalOne}$ a \textit{right configuration}. Before we give the actual rules of the semantics, we must first give some definitions\longshortversion{, starting with that of substitution.}
{.

First and foremost, note that we are not as interested in the actual names of labels within a boxed circuit as much as we are in the structure that they convey. In fact, when we apply circuits to one another, it might be necessary to rename some of the labels occurring in the applicand in order to avoid naming conflicts, all while preserving its structure. For this reason, whenever two circuits share the same structure and only differ by their respective labels, we consider them to be equivalent.
}

\longversion{\begin{definition}[Substitution]
		Let $\termOne$ $(\valOne)$ be a term (value) and let $\valTwo$ be a value such that none of the variables that occur free in $\valTwo$ occur bound in $\termOne$ $(\valOne)$. We define \emph{the substitution of $\valTwo$ for $\varOne$ in $\termOne$ $(\valOne)$}, written $\termOne[\valTwo/\varOne]$ $(\valOne[\valTwo/\varOne])$, as
		\begin{equation}
			\begin{aligned}
				\varTwo[\valTwo/\varOne] &= \begin{cases}
					\valTwo & \textnormal{if } \varOne\equiv\varTwo,\\
					\varTwo	& \textnormal{otherwise},
				\end{cases}\\
				(\lambda\varTwo_\typeOne.\termOne)[\valTwo/\varOne] &= \begin{cases}
					\lambda\varTwo_\typeOne.\termOne & \textnormal{if } \varOne\equiv\varTwo,\\
					\lambda\varTwo_\typeOne.\termOne[\valTwo/\varOne] &\textnormal{otherwise},
				\end{cases}\\
				(\letin{\varTwo}{\termOne}{\MontermOne})[\valTwo/\varOne] &= \begin{cases}
					\letin{\varOne}{\termOne[\valTwo/\varOne]}{\MontermOne} & \textnormal{if } \varOne\equiv\varTwo,\\
					\letin{\varOne}{\termOne[\valTwo/\varOne]}{\MontermOne[\valTwo/\varOne]} &\textnormal{otherwise},
				\end{cases}\\
				(\letin{\tuple{\varTwo,\varThree}}{\valThree}{\termOne})[\valTwo/\varOne] &= \begin{cases}
					\letin{\tuple{\varTwo,\varThree}}{\valThree[\valTwo/\varOne]}{\termOne} &\textnormal{if } \varOne\equiv\varTwo \vee \varOne\equiv\varThree,\\
					\letin{\tuple{\varTwo,\varThree}}{\valThree[\valTwo/\varOne]}{\termOne[\valTwo/\varOne]} & \textnormal{otherwise},
				\end{cases}\\
				\unitv[\valTwo/\varOne] &= \unitv,\\
				(\lift\termOne)[\valTwo/\varOne] &= \lift\termOne[\valTwo/\varOne],\\
				\boxedcirc{\vec\labOne}{\circuitOne}{\MonmvalOne}{\dlvaltreeOne}[\valTwo/\varOne] &= \boxedcirc{\vec\labOne}{\circuitOne}{\MonmvalOne}{\dlvaltreeOne},\\
				\tuple{\valThree_1,\valThree_2}[\valTwo/\varOne] &= \tuple{\valThree_1[\valTwo/\varOne],\valThree_2[\valTwo/\varOne]},\\
				(\valThree_1\valThree_2)[\valTwo/\varOne] &= \valThree_1[\valTwo/\varOne]\valThree_2[\valTwo/\varOne],\\
				(\force\valThree)[\valTwo/\varOne] &= \force\valThree[\valTwo/\varOne],\\
				(\boxt{\mtypeOne}{\dlvalsubOne}\valThree)[\valTwo/\varOne] &= \boxt{\mtypeOne}{\dlvalsubOne}\valThree[\valTwo/\varOne],\\
				\apply{\dlvalOne_1,\dots,\dlvalOne_n}(\valThree_1,\valThree_2)[\valTwo/\varOne] &= \apply{\dlvalOne_1,\dots,\dlvalOne_n}(\valThree_1[\valTwo/\varOne],\valThree_2[\valTwo/\varOne]),\\
				(\return \valThree)[\valTwo/\varOne] &= \return\valThree[\valTwo/\varOne].
			\end{aligned}
		\end{equation}
		The substitution operation is lifted naturally to a lifted term $\MontermOne$ as follows:
		\begin{equation}
			\MontermOne[\valOne/\varOne](\assOne) = \MontermOne(\assOne)[\valOne/\varOne].
		\end{equation}
\end{definition}}

\longshortversion
%LONG
{
	Next, we need some definitions that are essential to work with circuits as a side-effect. Namely, in order to avoid conflicts during circuit application, we need a notion of \emph{renaming of labels}, which is similar to -- but nevertheless distinct from -- the concept of renaming of lifted variables.
	
	\begin{definition}[Renaming of Labels]
		Given a label-bearing object $\genericObjectOne$ and a permutation $\renamingOne$ of $\labelset$, we call $\renamingOne$ a \emph{renaming of labels} and define $\renamed{\genericObjectOne}{\renamingOne}$ as $\genericObjectOne$ in which every occurrence of a label $\labOne$ is replaced by $\renamingOne(\labOne)$. In our case:
		
		\vspace{1em}
		\begin{tabular}{l l l}
			Label tuples & $\vec\labOne$ & $\renamed{\unitv}{\renamingOne} = \unitv,$\\
			&& $\renamed{\labOne}{\renamingOne} = \renamingOne(\labOne),$\\
			&& $\renamed{\tuple{\vec{\labOne_1},\vec{\labOne_2}}}{\renamingOne} = \tuple{\renamed{\vec{\labOne_1}}{\renamingOne},\renamed{\vec{\labOne_2}}{\renamingOne}}.$\\
			&&\\
			Label contexts & $\lcOne$ & $\renamed{\lcOne}{\renamingOne} = \lcOne\circ\invert{\renamingOne}.$\\
			&&\\
			Circuits & $\circuitOne$ & $\renamed{\cinput{\lcOne}}{\renamingOne} = \cinput{\renamed{\lcOne}{\renamingOne}},$\\
			&& $\renamed{(\circuitOne';\assOne\cond\gateOne(\vec\labOne)\to\vec\labTwo)}{\renamingOne} = \renamed{\circuitOne'}{\renamingOne};\assOne\cond\gateOne(\renamed{\vec\labOne}{\renamingOne}) \to \renamed{\vec\labTwo}{\renamingOne},$\\
			&& $\renamed{(\circuitOne';\assOne\cond\clift(\labOne)\lto\dlvalOne)}{\renamingOne} = \renamed{\circuitOne'}{\renamingOne};\assOne\cond\clift(\renamingOne(\labOne))\lto\dlvalOne.$\\
			&&\\
			Lifted objects & $\MongenericObjectOne$ & $\renamed{\treeLeaf{x}}{\renamingOne} = \treeLeaf{\renamed{x}{\renamingOne}},$\\
			&& $\renamed{\treeNode{\dlvalOne}{\MongenericObjectOne_1}{\MongenericObjectOne_2}}{\renamingOne} = \treeNode{\dlvalOne}{\renamed{\MongenericObjectOne_1}{\renamingOne}}{\renamed{\MongenericObjectOne_2}{\renamingOne}}.$\\
		\end{tabular}
	\end{definition}
	
	This allows us to define the notion of \emph{equivalent boxed circuits}. Intuitively, two boxed circuits are equivalent when they share the same lifting and gate application structure. That is, when they only differ by a renaming of labels.
	
	\begin{definition}[Equivalent Boxed Circuits]
		We say that two boxed circuits $\boxedcirc{\vec\labOne}{\circuitOne}{\MonmvalOne}{\dlvaltreeOne},\boxedcirc{\vec{\labOne'}}{\circuitOne'}{\MonmvalOne'}{\dlvaltreeOne}$ are \emph{equivalent}, and we write $\boxedcirc{\vec\labOne}{\circuitOne}{\MonmvalOne}{\dlvaltreeOne} \cong \boxedcirc{\vec{\labOne'}}{\circuitOne'}{\MonmvalOne'}{\dlvaltreeOne}$, when there exists a renaming of labels $\renamingOne$ such that $\vec{\labOne'} = \renamed{\vec\labOne}{\renamingOne},\circuitOne'=\renamed{\circuitOne}{\renamingOne}$ and $\MonmvalOne' = \renamed{\MonmvalOne}{\renamingOne}$.
	\end{definition}
}
%SHORT
{
	\begin{definition}[Equivalent Boxed Circuits]
		We say that two boxed circuits $\boxedcirc{\vec\labOne}{\circuitOne}{\MonmvalOne}{\dlvaltreeOne}$ and $\boxedcirc{\vec{\labOne'}}{\circuitOne'}{\MonmvalOne'}{\dlvaltreeOne}$ are \emph{equivalent}, and we write $\boxedcirc{\vec\labOne}{\circuitOne}{\MonmvalOne}{\dlvaltreeOne} \cong \boxedcirc{\vec{\labOne'}}{\circuitOne'}{\MonmvalOne'}{\dlvaltreeOne}$, when they only differ by a renaming of labels.
	\end{definition}
}

\longversion{
	We give the following results about renamings of labels and their relationship with \PQK's type system. These results are crucial to the type soundness of the language, which we prove in Section \ref{sec: soundness}.
	
	\begin{lemma}\label{lem: apply-1a}
		For any renaming of labels $\renamingOne$, if $\circjudgment{\circuitOne}{\dlvaltreeOne}{\lcOne}{\MonlcOne}$, then $\circjudgment{\renamed{\circuitOne}{\renamingOne}}{\dlvaltreeOne}{\renamed{\lcOne}{\renamingOne}}{\renamed{\MonlcOne}{\renamingOne}}$.
	\end{lemma}
	
	\begin{lemma}\label{lem: apply-1b}
		For any renaming of labels $\renamingOne$, if $\valjudgment{\pcontextOne}{\lcOne}{\vec\labOne}{\mtypeOne}$, then $\valjudgment{\pcontextOne}{\renamed{\lcOne}{\renamingOne}}{\renamed{\vec\labOne}{\renamingOne}}{\mtypeOne}$.
	\end{lemma}
	
	\begin{proposition}\label{lem: apply-1}
		If $\valjudgment{\emptycontext}{\emptycontext}{\boxedcirc{\vec\labOne}{\circuitOne}{\MonmvalOne}{\dlvaltreeOne}}{\circt{\dlvaltreeOne}(\mtypeOne,\MonmtypeTwo)}$ and $\boxedcirc{\vec\labOne}{\circuitOne}{\MonmvalOne}{\dlvaltreeOne} \cong \boxedcirc{\vec{\labOne'}}{\circuitOne'}{\MonmvalOne'}{\dlvaltreeOne}$, then $\valjudgment{\emptycontext}{\emptycontext}{\boxedcirc{\vec{\labOne'}}{\circuitOne'}{\MonmvalOne'}{\dlvaltreeOne}}{\circt{\dlvaltreeOne}(\mtypeOne,\MonmtypeTwo)}$.
	\end{proposition}
	\begin{proof}
		By expansion of the definition of $\boxedcirc{\vec\labOne}{\circuitOne}{\MonmvalOne}{\dlvaltreeOne} \cong \boxedcirc{\vec{\labOne'}}{\circuitOne'}{\MonmvalOne'}{\dlvaltreeOne}$ and lemmata \ref{lem: apply-1a} and \ref{lem: apply-1b}.
	\end{proof}
}

\longshortversion{Lastly}{Next}, we define the two operations that actually implement the semantics of $\apply{}$. The circuit insertion function $\cconcat_\assOne$ is just a simple concatenation function defined on CRL circuits and all the heavy lifting is actually performed by the $\append$ function. This function is responsible for the actual unboxing of a boxed circuit, the renaming of its labels (to match the outputs of the underlying circuit and avoid label conflicts), the instantiation of the abstracted lifted variable names within it and its insertion in the underlying circuit.
		
\begin{definition}[Insertion of Circuits]
	Suppose $\circuitOne$ and $\circuitTwo$ are two circuits. We define \emph{the insertion of $\circuitTwo$ in $\circuitOne$ on branch $\assOne$}, written $\circuitOne\cconcat_\assOne \circuitTwo$ as:
	\begin{equation}
		\begin{aligned}
			\circuitOne \cconcat_\assOne \cinput{\lcOne} &= \circuitOne,\\
			\circuitOne \cconcat_\assOne (\circuitTwo';\assTwo\cond\gateOne(\lcOne) \to \lcTwo) &= (\circuitOne \cconcat_\assOne \circuitTwo'); \assOne\assmerge\assTwo\cond \gateOne(\lcOne) \to \lcTwo,\\
			\circuitOne \cconcat_\assOne (\circuitTwo';\assTwo\cond\clift(\labOne)\lto\dlvalOne) &= (\circuitOne \cconcat_\assOne \circuitTwo'); \assOne\assmerge\assTwo\cond \lift(\labOne)\lto\dlvalOne\longshortversion{.}{,}
		\end{aligned}
	\end{equation}
	\shortversion{where $\assOne\assmerge\assTwo$ denotes the union of two assignments with disjoint domains.}
\end{definition}

\begin{definition}[$\append$]
	Suppose $\circjudgment{\circuitOne}{\dlvaltreeOne}{\lcOne}{\MonlcOne}$ is a circuit and $\boxedcirc{\vec\labOne}{\circuitTwo}{\MonmvalOne}{\dlvaltreeTwo}$ is a boxed circuit with $\namesof{}(\dlvaltreeTwo)=\{\dlvalOne_1,\dots,\dlvalOne_n\}$. Suppose $\assOne\in\cassset_{\dlvaltreeOne}$ and let $\vec\labTwo$ be a label tuple whose labels all occur in $\MonlcOne(\assOne)$. Finally, let $\dlvalTwo_1,\dots,\dlvalTwo_n$ be a sequence of distinct lifted variable names which do not occur in $\namesof{\assOne}(\dlvaltreeOne)$. We define $\append(\circuitOne,\assOne,\vec\labTwo,\boxedcirc{\vec\labOne}{\circuitTwo}{\MonmvalOne}{\dlvaltreeTwo},\dlvalTwo_1,\dots,\dlvalTwo_n)$
	as the function that
	\begin{enumerate}
		\item Finds $\boxedcirc{\vec\labTwo}{\circuitTwo'}{\MonmvalOne'}{\dlvaltreeOne} \cong \boxedcirc{\vec\labOne}{\circuitTwo}{\MonmvalOne}{\dlvaltreeOne}$ such that all the labels occurring in $\circuitTwo'$, but not in $\vec\labTwo$, are fresh in $\circuitOne$,
		\item Computes $\circuitTwo''=\alpharenamed{\circuitTwo'}{\dlvalTwo_1/\dlvalOne_1,\dots,\dlvalTwo_n/\dlvalOne_n}$ and $\MonmvalOne''=\alpharenamed{\MonmvalOne'}{\dlvalTwo_1/\dlvalOne_1,\dots,\dlvalTwo_n/\dlvalOne_n}$,
		\item Returns $(\circuitOne\cconcat_\assOne \circuitTwo'',\MonmvalOne'')$.
	\end{enumerate}
\end{definition}

The rules of the operational semantics can be found in Figure \ref{table: bs semantic rules}, where $\freshlabels(\mtypeOne)$ produces $(\lcOne,\vec\labOne)$ such that $\mjudgment{\lcOne}{\vec\labOne}{\mtypeOne}$.
We also consider a notion of \textit{divergence} of a configuration. Intuitively, a configuration $\bscfgl{\circuitOne}{\assOne}{\termOne}$ \emph{diverges}, and we write $\bscfgl{\circuitOne}{\assOne}{\termOne}\bsdiverge$, when its evaluation does not terminate.
\longshortversion{More formally, divergence can be derived by the rules in Figure \ref{table: bs divergence rules}.}{The rules for divergence, which we omit for brevity, closely resemble the corresponding rules of the big-step semantics, with the fundamental difference that they are given a coinductive reading. As a (non-exhaustive) example, consider the following rules for the divergence of the application, force and let constructs, respectively:
	
\begin{equation}
	\inference[app]
	{\bscfgl{\circuitOne}{\assOne}{\termOne[\valOne/\varOne]} \bsdiverge}
	%---------------------------------------------
	{\bscfgl{\circuitOne}{\assOne}{(\lambda\varOne_\typeOne.\termOne)\valOne} \bsdiverge}
	\quad
	\inference[force]
	{\bscfgl{\circuitOne}{\assOne}{\termOne} \bsdiverge}
	%---------------------------------------------
	{\bscfgl{\circuitOne}{\assOne}{\force(\lift\termOne)} \bsdiverge}
	\quad
	\inference[let-now]
	{\bscfgl{\circuitOne}{\assOne}{\termOne}\bsdiverge
	}
	%---------------------------------------------
	{\bscfgl{\circuitOne}{\assOne}{\letin{\varOne}{\termOne}{\MontermOne}} \bsdiverge}
\end{equation}}
\begin{figure}[tb]
	\centering
	\fbox{\begin{minipage}{.98\textwidth}
			
			$$
			\inference[app]
			{\bscfgl{\circuitOne}{\assOne}{\termOne[\valOne/\varOne]} \bseval \bscfgr{\circuitTwo}{\MonvalOne}}
			%---------------------------------------------
			{\bscfgl{\circuitOne}{\assOne}{(\lambda\varOne_\typeOne.\termOne)\valOne} \bseval \bscfgr{\circuitTwo}{\MonvalOne}}
			\qquad
			\inference[dest]
			{\bscfgl{\circuitOne}{\assOne}{\termOne[\valOne/\varOne,\valTwo/\varTwo]} \bseval \bscfgr{\circuitTwo}{\MonvalOne}}
			%---------------------------------------------
			{\bscfgl{\circuitOne}{\assOne}{\letin{\tuple{\varOne,\varTwo}}{\tuple{\valOne,\valTwo}}{\termOne}} \bseval \bscfgr{\circuitTwo}{\MonvalOne}}
			$$
			\rulespace
			$$
			\inference[let]
			{\bscfgl{\circuitOne}{\assOne}{\termOne}\bseval\bscfgr{\circuitOne_1}{\MonvalOne}
				&
				\MonvalOne\in\condmonad_{\dlvaltreeOne}(\valset)
				&
				\MontermOne\in\condmonad_{\dlvaltreeOne}(\termset)
				\\
				\cassset_{\dlvaltreeOne} = \{\assOne_1,\dots,\assOne_n\}
				&
				\bscfgl{\circuitOne_i}{\assOne\cup\assOne_i}{\MontermOne(\assOne_i)[\MonvalOne(\assOne_i)/\varOne]} \bseval \bscfgr{\circuitOne_{i+1}}{\MonvalTwo_{\assOne_i}} \for i=1,\dots, n
			}
			%---------------------------------------------
			{\bscfgl{\circuitOne}{\assOne}{\letin{\varOne}{\termOne}{\MontermOne}} \bseval \bscfgr{\circuitOne_{n+1}}{\Monflatten{\Moncomp{\dlvaltreeOne}{\MonvalTwo_\assOne}{\assOne}{}}}}
			$$
			\rulespace
			$$
			\inference[force]
			{\bscfgl{\circuitOne}{\assOne}{\termOne} \bseval \bscfgr{\circuitTwo}{\MonvalOne}}
			%---------------------------------------------
			{\bscfgl{\circuitOne}{\assOne}{\force(\lift\termOne)} \bseval \bscfgr{\circuitTwo}{\MonvalOne}}
			\qquad
			\inference[apply]
			{(\circuitOne',\MonmvalOne') = \append(\circuitOne,\assOne,\vec\labTwo,\boxedcirc{\vec\labOne}{\circuitTwo}{\MonmvalOne}{\dlvaltreeOne},\dlvalTwo_1,\dots,\dlvalTwo_n)}
			%---------------------------------------------
			{\bscfgl{\circuitOne}{\assOne}{\apply{\dlvalTwo_1,\dots,\dlvalTwo_n}(\boxedcirc{\vec\labOne}{\circuitTwo}{\MonmvalOne}{\dlvaltreeOne},\vec\labTwo)} \bseval \bscfgr{\circuitOne'}{\MonmvalOne'}}
			$$
			\rulespace
			$$
			\inference[box]
			{
				(\lcOne,\vec\labOne)=\freshlabels(\mtypeOne)
				&
				\bscfgl{\cinput{\lcOne}}{\emptyassignment}{\letin{\varOne}{\termOne}{ \Monunit{\varOne\vec\labOne}}} \bseval \bscfgr{\circuitTwo}{\MonmvalOne}
				&
				\MonmvalOne\in\condmonad_{\dlvaltreeOne}(\mvalset)}
			%---------------------------------------------
			{\bscfgl{\circuitOne}{\assOne}{\boxt{\mtypeOne}{}(\lift\termOne)} \bseval \bscfgr{\circuitOne}{\Monunit{\boxedcirc{\vec\labOne}{\circuitTwo}{\MonmvalOne}{\dlvaltreeOne}}}}
			$$
			\rulespace
			$$
			\inference[return]
			{}
			%---------------------------------------------
			{\bscfgl{\circuitOne}{\assOne}{\return \valOne}\bseval\bscfgr{\circuitOne}{\Monunit{\valOne}}}
			$$
	\end{minipage}}
	\caption{The big-step operational semantics of \PQK}
	\label{table: bs semantic rules}
\end{figure}
\longversion{
	\begin{figure}[tb]
		\centering
		\fbox{\begin{minipage}{.98\textwidth}
				$$
				\inference[app]
				{\bscfgl{\circuitOne}{\assOne}{\termOne[\valOne/\varOne]} \bsdiverge}
				%---------------------------------------------
				{\bscfgl{\circuitOne}{\assOne}{(\lambda\varOne_\typeOne.\termOne)\valOne} \bsdiverge}
				\qquad
				\inference[dest]
				{\bscfgl{\circuitOne}{\assOne}{\termOne[\valOne/\varOne,\valTwo/\varTwo]} \bsdiverge}
				%---------------------------------------------
				{\bscfgl{\circuitOne}{\assOne}{\letin{\tuple{\varOne,\varTwo}}{\tuple{\valOne,\valTwo}}{\termOne}} \bsdiverge}
				$$
				\rulespace
				$$
				\inference[force]
				{\bscfgl{\circuitOne}{\assOne}{\termOne} \bsdiverge}
				%---------------------------------------------
				{\bscfgl{\circuitOne}{\assOne}{\force(\lift\termOne)} \bsdiverge}
				\qquad
				\inference[let-now]
				{\bscfgl{\circuitOne}{\assOne}{\termOne}\bsdiverge
				}
				%---------------------------------------------
				{\bscfgl{\circuitOne}{\assOne}{\letin{\varOne}{\termOne}{\MontermOne}} \bsdiverge}
				$$
				\rulespace
				$$
				\inference[let-then]
				{\bscfgl{\circuitOne}{\assOne}{\termOne}\bseval\bscfgr{\circuitOne_1}{\MonvalOne}
					&
					\MonvalOne\in\condmonad_{\dlvaltreeOne}(\valset)
					&
					\MontermOne\in\condmonad_{\dlvaltreeOne}(\termset)
					\\
					\cassset_{\dlvaltreeOne} = \{\assOne_1,\dots,\assOne_n\}
					&
					\bscfgl{\circuitOne_i}{\assOne\assmerge\assOne_i}{\MontermOne(\assOne_i)[\MonvalOne(\assOne_i)/\varOne]} \bseval \bscfgr{\circuitOne_{i+1}}{\MonvalTwo_{\assOne_i}} \for i=1,\dots,j-1
					\\
					\bscfgl{\circuitOne_j}{\assOne\assmerge\assOne_j}{\MontermOne(\assOne_j)[\MonvalOne(\assOne_j)/\varOne]} \bsdiverge
				}
				%---------------------------------------------
				{\bscfgl{\circuitOne}{\assOne}{\letin{\varOne}{\termOne}{\MontermOne}} \bsdiverge}
				$$
				\rulespace
				$$
				\inference[box]
				{
					(\lcOne,\vec\labOne)=\freshlabels(\mtypeOne)
					&
					\bscfgl{\cinput{\lcOne}}{\emptyassignment}{\letin{\varOne}{\termOne}{\Monunit{\varOne\vec\labOne}}} \bsdiverge}
				%---------------------------------------------
				{\bscfgl{\circuitOne}{\assOne}{\boxt{\mtypeOne}{\dlvalsubOne}(\lift\termOne)} \bsdiverge}
				$$
		\end{minipage}}
		\caption{The big-step divergence rules of \PQK}
		\label{table: bs divergence rules}
	\end{figure}
}

%%%%%%%%%%%%%%%%%%%%%%%%%%%%%%%%%
\section{Type Soundness}
\label{sec: soundness}
%%%%%%%%%%%%%%%%%%%%%%%%%%%%%%%%%

\longversion{\subsection{General Computations}}

In general, the well-typedness of a \PQK\ program strongly depends on the underlying circuit. Specifically, a term $\termOne$ is well-typed when all the free labels occurring in it can be found with the appropriate type in the outputs of the underlying circuit, on the branch that $\termOne$ is manipulating. For this reason, we give the following notions of well-typedness on configurations.

\longshortversion{
	%LONG
	\begin{definition}[Well-Typed Left Configuration]
	We say that a left configuration $\bscfgl{\circuitOne}{\assOne}{\termOne}$ is \emph{well-typed with input $\lcOne$, past lifting tree $\dlvaltreeOne$, future lifting tree $\dlvaltreeTwo$, lifted type $\MontypeOne$ and outputs $\MonlcOne$}, and we write $\lbscfgjudgment{\lcOne}{\dlvaltreeOne}{\dlvaltreeTwo}{\bscfgl{\circuitOne}{\assOne}{\termOne}}{\MontypeOne}{\MonlcOne}$, when $\assOne\in\cassset_\dlvaltreeOne,\namesof{\assOne}(\dlvaltreeOne)\cap\namesof{}(\dlvaltreeTwo)=\emptyset,\circjudgment{\circuitOne}{\dlvaltreeOne}{\lcOne}{\MonlcOne\Monmerge{\assOne}\lcOne'}$ and $\compjudgment{\emptycontext}{\lcOne'}{\dlvaltreeTwo}{\termOne}{\MontypeOne}$.
	\end{definition}

	\begin{definition}[Well-Typed Right Configuration]
	We say that a right configuration $\bscfgr{\circuitOne}{\MonvalOne}$ is \emph{well-typed in the $\assOne$ branch with input $\lcOne$, overall lifting tree $\dlvaltreeOne$, lifted type $\MontypeOne$ and outputs $\MonlcOne$}, and we write $\rbscfgjudgment{\lcOne}{\dlvaltreeOne}{\dlvaltreeTwo}{\assOne}{\bscfgr{\circuitOne}{\MonvalOne}}{\MontypeOne}{\MonlcOne}$, when $\dlvaltreeOne=\Monflatten{\Moncomp{\dlvaltreeOne'}{\dlvaltreeTwo}{}{\{\assOne\}}},\MonlcOne=\Monflatten{\Moncomp{\MonlcOne'}{\MonlcOne''}{}{\{\assOne\}}},\circjudgment{\circuitOne}{\dlvaltreeOne}{\lcOne}{\Monflatten{\Moncomp{\MonlcOne'}{\MonlcOne'',\MonlcTwo}{}{\{\assOne\}}}}$ and $\Monvaljudgment{\emptycontext}{\MonlcTwo}{\dlvaltreeTwo}{\MonvalOne}{\MontypeOne}$.
	\end{definition}
}{
	\begin{definition}[Well-Typed Configuration]
		We say that
		\begin{itemize}
			\item a left configuration $\bscfgl{\circuitOne}{\assOne}{\termOne}$ is \emph{well-typed with input $\lcOne$, past lifting tree $\dlvaltreeOne$, future lifting tree $\dlvaltreeTwo$, lifted type $\MontypeOne$ and outputs $\MonlcOne$}, and we write $\lbscfgjudgment{\lcOne}{\dlvaltreeOne}{\dlvaltreeTwo}{\bscfgl{\circuitOne}{\assOne}{\termOne}}{\MontypeOne}{\MonlcOne}$, when $\assOne\in\cassset_\dlvaltreeOne,\namesof{\assOne}(\dlvaltreeOne)\cap\namesof{}(\dlvaltreeTwo)=\emptyset,\circjudgment{\circuitOne}{\dlvaltreeOne}{\lcOne}{\MonlcOne\Monmerge{\assOne}\lcOne'}$ and $\compjudgment{\emptycontext}{\lcOne'}{\dlvaltreeTwo}{\termOne}{\MontypeOne}$,
			
			\item a right configuration $\bscfgr{\circuitOne}{\MonvalOne}$ is \emph{well-typed in the $\assOne$ branch with input $\lcOne$, overall lifting tree $\dlvaltreeOne$, lifted type $\MontypeOne$ and outputs $\MonlcOne$}, and we write $\rbscfgjudgment{\lcOne}{\dlvaltreeOne}{\dlvaltreeTwo}{\assOne}{\bscfgr{\circuitOne}{\MonvalOne}}{\MontypeOne}{\MonlcOne}$, when $\dlvaltreeOne=\Monflatten{\Moncomp{\dlvaltreeOne'}{\dlvaltreeTwo}{}{\{\assOne\}}},\MonlcOne=\Monflatten{\Moncomp{\MonlcOne'}{\MonlcOne''}{}{\{\assOne\}}},\circjudgment{\circuitOne}{\dlvaltreeOne}{\lcOne}{\Monflatten{\Moncomp{\MonlcOne'}{\MonlcOne'',\MonlcTwo}{}{\{\assOne\}}}}$ and $\Monvaljudgment{\emptycontext}{\MonlcTwo}{\dlvaltreeTwo}{\MonvalOne}{\MontypeOne}$.
		\end{itemize}
	\end{definition}
}

\longversion{
	
	\subsubsection{Subject Reduction}
	
	To prove the subject reduction result, we avail ourselves of the following lemmata.
	
	\begin{lemma}\label{lem: apply-2.1}
		For any renaming of lifted variables $\alpharenamingOne$, if $\circjudgment{\circuitOne}{\dlvaltreeOne}{\lcOne}{\MonlcOne}$, then $\circjudgment{\alpharenamed{\circuitOne}{\alpharenamingOne}}{\alpharenamed{\dlvaltreeOne}{\alpharenamingOne}}{\lcOne}{\alpharenamed{\MonlcOne}{\alpharenamingOne}}$.
	\end{lemma}
	
	\begin{lemma}\label{lem: apply-2.2}
		For any renaming of lifted variables $\alpharenamingOne$, if $\Monvaljudgment{\Moncontext}{\MonlcOne}{\dlvaltreeOne}{\MonvalOne}{\MontypeOne}$, where $\Moncontext\in\condmonad_\dlvaltreeOne(\contextset)$, then $\Monvaljudgment{\alpharenamed{\Moncontext}{\alpharenamingOne}}{\alpharenamed{\MonlcOne}{\alpharenamingOne}}{\alpharenamed{\dlvaltreeOne}{\alpharenamingOne}}{\alpharenamed{\MonvalOne}{\alpharenamingOne}}{\alpharenamed{\MontypeOne}{\alpharenamingOne}}$.
	\end{lemma}
	
	\begin{lemma}[Insertion Signature]\label{lem: apply-3}
		If $\circjudgment{\circuitOne}{\dlvaltreeOne}{\lcOne}{\Moncomp{\MonlcOne}{\lcOne',\lcOne''}{}{\{\assOne\}}}$ for some $\assOne\in\cassset_\dlvaltreeOne$ and $\circjudgment{\circuitTwo}{\dlvaltreeTwo}{\lcOne'}{\MonlcTwo}$ for some $\circuitTwo$ such that the labels that occur in $\circuitTwo$, but not in $\lcOne'$, are fresh in $\circuitOne$ and $\namesof{\assOne}(\dlvaltreeOne)\cap\namesof{}(\dlvaltreeTwo)=\emptyset$, then $\circjudgment{\circuitOne\cconcat_\assOne\circuitTwo}{\Monflatten{\Moncomp{\dlvaltreeOne}{\dlvaltreeTwo}{}{\{\assOne\}}}}{\lcOne}{\Monflatten{\Moncomp{\MonlcOne}{\MonlcTwo,\lcOne''}{}{\{\assOne\}}}}$.
	\end{lemma}
	\begin{proof}
		By induction on $\circjudgment{\circuitTwo}{\dlvaltreeTwo}{\lcOne'}{\MonlcTwo}$.
	\end{proof}
	
	\begin{lemma}\label{lem: comma extraction}
		Suppose $\MongenericObjectOne\in\condmonad_{\dlvaltreeOne}(\genericSetOne),\MongenericObjectTwo\in\condmonad_{\dlvaltreeTwo}(\genericSetOne),\assOne\in\cassset_\dlvaltreeOne$ and $\assTwo\in\assset_{\dlvaltreeTwo}$. We have $\Monflatten{\Moncomp{\MongenericObjectOne}{\Moncomp{\MongenericObjectTwo}{\genericObjectOne_\assThree}{\assThree}{\cassset_{\dlvaltreeTwo}^\assTwo}}{}{\{\assOne\}}} = \Moncomp{\Monflatten{\Moncomp{\MongenericObjectOne}{\MongenericObjectTwo}{}{\{\assOne\}}}}{\genericObjectOne_\assThree}{\assThree}{\cassset_{\Monflatten{\Moncomp{\dlvaltreeOne}{\dlvaltreeTwo}{}{\{\assOne\}}}}^{\assOne\assmerge\assTwo}}$.
	\end{lemma}
	
	\begin{lemma}\label{lem: meta-lifted-judgment}
		If $\Monvaljudgment{\Moncontext_\assOne}{\MonlcOne_\assOne}{\dlvaltreeTwo_\assOne}{\MonvalOne_\assOne}{\MontypeOne_\assOne}$ for all $\assOne\in\cassset_{\dlvaltreeOne}$, then $\Monvaljudgment{\Monflatten{\Moncomp{\dlvaltreeOne}{\Moncontext_\assOne}{\assOne}{}}}{\Monflatten{\Moncomp{\dlvaltreeOne}{\MonlcOne_\assOne}{\assOne}{}}}{\Monflatten{\Moncomp{\dlvaltreeOne}{\dlvaltreeTwo_\assOne}{\assOne}{}}}{\Monflatten{\Moncomp{\dlvaltreeOne}{\MonvalOne_\assOne}{\assOne}{}}}{\Monflatten{\Moncomp{\dlvaltreeOne}{\MontypeOne_\assOne}{\assOne}{}}}$.
	\end{lemma}
	
	\begin{lemma}[Substitution]\label{lem: substitution}
		If $\valjudgment{\pcontextOne,\contextOne'}{\lcOne'}{\valOne}{\typeOne}$ and $\Pi$ is a type derivation, then
		\begin{enumerate}
			\item If the conclusion of $\Pi$ is $\compjudgment{\pcontextOne,\contextOne,\varOne:\typeOne}{\lcOne}{\dlvaltreeOne}{\termOne}{\MontypeOne}$, then $\compjudgment{\pcontextOne,\contextOne,\contextOne'}{\lcOne,\lcOne'}{\dlvaltreeOne}{\termOne[\valOne/\varOne]}{\MontypeOne}$,
			\item If the conclusion of $\Pi$ is $\valjudgment{\pcontextOne,\contextOne,\varOne:\typeOne}{\lcOne}{\valTwo}{\typeTwo}$, then $\valjudgment{\pcontextOne,\contextOne,\contextOne'}{\lcOne,\lcOne'}{\valTwo[\valOne/\varOne]}{\typeTwo}$,
		\end{enumerate}
	\end{lemma}
	\begin{proof}
		We prove the claim separately for the cases in which $\valOne$ has parameter type and linear type. In both cases, we proceed by induction on the size of $\Pi$ and case analysis on its last rule.
	\end{proof}
}

\longversion{
	\begin{theorem}[Subject Reduction]\label{thm: SR}
		If $\lbscfgjudgment{\lcOne}{\dlvaltreeOne}{\dlvaltreeTwo}{\bscfgl{\circuitOne}{\assOne}{\termOne}}{\MontypeOne}{\MonlcOne}$ and $\exists\circuitTwo,\MonvalOne.\bscfgl{\circuitOne}{\assOne}{\termOne}\bseval\bscfgr{\circuitTwo}{\MonvalOne}$, then $\rbscfgjudgment{\lcOne}{\Monflatten{\Moncomp{\dlvaltreeOne}{\dlvaltreeTwo}{}{\{\assOne\}}}}{\dlvaltreeTwo}{\assOne}{\bscfgr{\circuitTwo}{\MonvalOne}}{\MontypeOne}{\MonlcOne\Monsplit{\assOne}\dlvaltreeTwo}$.
	\end{theorem}
	\begin{proof}
		By induction on $\bscfgl{\circuitOne}{\assOne}{\termOne}\bseval\bscfgr{\circuitTwo}{\MonvalOne}$ and case analysis on the last rule used in its derivation. The \textit{app, dest, force} cases are straightforward, so we focus on the interesting cases of \textit{return, box, apply} and \textit{let}.
				
		\subparagraph*{Case of \textit{return}}
		In this case we have $\lbscfgjudgment{\lcOne}{\dlvaltreeOne}{\dlvaltreeTwo}{\bscfgl{\circuitOne}{\assOne}{\return\valOne}}{\MontypeOne}{\MonlcOne}$ and $\bscfgl{\circuitOne}{\assOne}{\return\valOne}\bseval\bscfgr{\circuitTwo}{\MonvalOne}$. By the definition of well-typedness we know that $\assOne\in\cassset_\dlvaltreeOne,\namesof{\assOne}(\dlvaltreeOne)\cap\namesof{}(\dlvaltreeTwo)=\emptyset,\circjudgment{\circuitOne}{\dlvaltreeOne}{\lcOne}{\MonlcOne\Monmerge{\assOne}\lcOne'}$ and $\compjudgment{\emptycontext}{\lcOne'}{\dlvaltreeTwo}{\return\valOne}{\MontypeOne}$. By inversion of the applicable rules we have
		\begin{equation}
			\inference[return]
			{}
			{\bscfgl{\circuitOne}{\assOne}{\return \valOne}\bseval\bscfgr{\circuitOne}{\Monunit{\valOne}}}
			\qquad
			\inference[return]
			{\valjudgment{\emptycontext}{\lcOne'}{\valOne}{\typeOne}}
			{\compjudgment{\emptycontext}{\lcOne'}{\emptytree}{\return \valOne}{\Monunit{\typeOne}}}
		\end{equation}
		Thus $\circuitTwo\equiv\circuitOne,\MonvalOne\equiv\Monunit{\valOne},\dlvaltreeTwo\equiv\emptytree$ and $\MontypeOne\equiv\Monunit{\typeOne}$. We immediately have $\MonlcOne\Monmerge{\assOne}\lcOne'=\Moncomp{\MonlcOne'}{\lcOne',\lcOne''}{}{\{\assOne\}} = \Monflatten{\Moncomp{\MonlcOne'}{\Monunit{\lcOne'},\Monunit{\lcOne''}}{}{\{\assOne\}}}$ and $\Monvaljudgment{\Monunit{\emptycontext}}{\Monunit{\lcOne'}}{\emptytree}{\Monunit{\valOne}}{\Monunit{\typeOne}}$. Furthermore, $\dlvaltreeOne\Monsplit{\assOne}\emptytree=\dlvaltreeOne$ and $\MonlcOne\Monsplit{\assOne}\emptytree=\MonlcOne$, so we conclude $\rbscfgjudgment{\lcOne}{\dlvaltreeOne}{}{\assOne}{\bscfgr{\circuitOne}{\Monunit{\valOne}}}{\Monunit{\typeOne}}{\MonlcOne}$.
		
		\subparagraph*{Case of \textit{box}}
		In this case we have $\lbscfgjudgment{\lcOne}{\dlvaltreeOne}{\dlvaltreeTwo}{\bscfgl{\circuitOne}{\assOne}{\boxt{\mtypeOne}{}\valOne}}{\MontypeOne}{\MonlcOne}$ and $\bscfgl{\circuitOne}{\assOne}{\boxt{\mtypeOne}{}\valOne}\bseval\bscfgr{\circuitTwo}{\MonvalOne}$. By the definition of well-typedness we know that $\assOne\in\cassset_\dlvaltreeOne,\namesof{\assOne}(\dlvaltreeOne)\cap\namesof{}(\dlvaltreeTwo)=\emptyset,\circjudgment{\circuitOne}{\dlvaltreeOne}{\lcOne}{\MonlcOne\Monmerge{\assOne}\lcOne'}$ and $\compjudgment{\emptycontext}{\lcOne'}{\dlvaltreeTwo}{\boxt{\mtypeOne}{}\valOne}{\MontypeOne}$. By inversion of the applicable rules we have
		\begin{equation}
			\inference[box]
		{
			(\lcTwo,\vec\labOne)=\freshlabels(\mtypeOne)
			&
			\MonmvalOne\in\condmonad_{\dlvaltreeTwo''}(\mvalset)
			\\
			\bscfgl{\cinput{\lcTwo}}{\emptyassignment}{\letin{\varOne}{\termOne}{ \Monunit{\varOne\vec\labOne}}} \bseval \bscfgr{\circuitTwo'}{\MonmvalOne}}
		%---------------------------------------------
		{\bscfgl{\circuitOne}{\assOne}{\boxt{\mtypeOne}{}(\lift\termOne)} \bseval \bscfgr{\circuitOne}{\Monunit{\boxedcirc{\vec\labOne}{\circuitTwo'}{\MonmvalOne}{\dlvaltreeTwo''}}}}
		\qquad
		\inference[box]
		{
			\inference[lift]
			{\compjudgment{\emptycontext}{\emptycontext}{\emptytree}{\termOne}{\Monunit{\mtypeOne \multimap_{\dlvaltreeTwo'} \MonmtypeTwo}}}
			{\valjudgment{\emptycontext}{\emptycontext}{\lift\termOne}{\bang\Monunit{\mtypeOne \multimap_{\dlvaltreeTwo'} \MonmtypeTwo}}}
		}
		%---------------------------------------------
		{\compjudgment{\emptycontext}{\emptycontext}{\emptytree}{\boxt{\mtypeOne}{\dlvalsubOne} (\lift\termOne)}{\Monunit{\circt{\dlvaltreeTwo'}(\mtypeOne,\MonmtypeTwo)}}}
		\end{equation}
		Thus $\valOne\equiv\lift\termOne,\circuitTwo\equiv\circuitOne,\MonvalOne\equiv\Monunit{\boxedcirc{\vec\labOne}{\circuitTwo'}{\MonmvalOne}{\dlvaltreeTwo''}},\lcOne'\equiv\emptycontext,\dlvaltreeTwo\equiv\emptytree$ and $\MontypeOne\equiv\Monunit{\circt{\dlvaltreeTwo'}(\mtypeOne,\MonmtypeTwo)}$. Because $\freshlabels(\mtypeOne)$ produces $(\lcTwo,\vec\labOne)$ such that $\mjudgment{\lcTwo}{\vec\labOne}{\mtypeOne}$, by Proposition \ref{prop: m-judgment and val-judgment} we have $\valjudgment{\emptycontext}{\lcTwo}{\vec\labOne}{\mtypeOne}$ and we can derive the following:
		\begin{equation}
			\inference[let]
		{
			\compjudgment{\emptycontext}{\emptycontext}{\emptytree}{\termOne}{\Monunit{\mtypeOne \multimap_{\dlvaltreeTwo'} \MonmtypeTwo}}
			&
			\inference[app]
			{
				\inference[var]{}
				{\valjudgment{\varOne:\mtypeOne\multimap_{\dlvaltreeTwo'}\MonmtypeTwo}{\emptycontext}{\varOne}{\mtypeOne\multimap_{\dlvaltreeTwo'}\MonmtypeTwo}}
				&
				\valjudgment{\emptycontext}{\lcTwo}{\vec\labOne}{\mtypeOne}			
			}
			{\compjudgment{\varOne:\mtypeOne\multimap_{\dlvaltreeTwo'}\MonmtypeTwo}{\lcTwo}{\dlvaltreeTwo'}{\varOne\vec\labOne}{\MonmtypeTwo}}
			\\
			\Monunit{\varOne\vec\labOne}\in\condmonad_\emptytree(\termset)
		}
		{\compjudgment{\emptycontext}{\lcTwo}{\dlvaltreeTwo'}{\letin{\varOne}{\termOne}{\Monunit{\varOne\vec\labOne}}}{\MonmtypeTwo}}
		\end{equation}
		Because we know $\circjudgment{\cinput{\lcTwo}}{\emptytree}{\lcTwo}{\Monunit{\lcTwo}}$, we can say $\lbscfgjudgment{\lcTwo}{\emptytree}{\dlvaltreeTwo'}{\bscfgl{\cinput{\lcTwo}}{\emptyassignment}{\letin{\varOne}{\termOne}{\Monunit{\varOne\vec\labOne}}}}{\MonmtypeTwo}{\Monunit{\emptycontext}}$. By inductive hypothesis we therefore get $\rbscfgjudgment{\lcTwo}{\dlvaltreeTwo'}{}{\emptyassignment}{\bscfgr{\circuitTwo'}{\MonmvalOne}}{\MonmtypeTwo}{\Moncomp{\dlvaltreeTwo'}{\emptycontext}{\assTwo}{}}$, which entails $\circjudgment{\circuitTwo'}{\dlvaltreeTwo'}{\lcTwo}{\MonlcTwo}$ and $\Monvaljudgment{\emptycontext}{\MonlcTwo}{\dlvaltreeTwo'}{\MonmvalOne}{\MonmtypeTwo}$. This implies $\dlvaltreeTwo'\equiv\dlvaltreeTwo''$ and, together with Proposition \ref{lem: m-type implies m-val}, $\MonmvalOne\in\condmonad_{\dlvaltreeTwo'}(\mvalset)$. We can now derive
		\begin{equation}
			\inference[circ]
		{
			\circjudgment{\circuitTwo'}{\dlvaltreeTwo'}{\lcTwo}{\MonlcTwo}
			&
			\valjudgment{\emptycontext}{\lcTwo}{\vec\labOne}{\mtypeOne}
			&
			\Monvaljudgment{\emptycontext}{\MonlcTwo}{\dlvaltreeTwo'}{\MonmvalOne}{\MonmtypeTwo}
		}
		{
			\valjudgment{\emptycontext}{\emptycontext}{\boxedcirc{\vec\labOne}{\circuitTwo'}{\MonmvalOne}{\dlvaltreeTwo'}}{\circt{\dlvaltreeTwo'}(\mtypeOne,\MonmtypeTwo)}
		}
		\end{equation}
		which allows us to conclude $\rbscfgjudgment{\lcOne}{\dlvaltreeOne}{}{\assOne}{\bscfgr{\circuitOne}{\Monunit{\boxedcirc{\vec\labOne}{\circuitTwo'}{\MonmvalOne}{\dlvaltreeTwo'}}}}{\Monunit{\circt{\dlvaltreeTwo'}(\mtypeOne,\MonmtypeTwo)}}{\MonlcOne}$.
		
		\subparagraph*{Case of \textit{apply}}
		In this case we have $\lbscfgjudgment{\lcOne}{\dlvaltreeOne}{\dlvaltreeTwo}{\bscfgl{\circuitOne}{\assOne}{\apply{\dlvalTwo_1,\dots,\dlvalTwo_n}(\valOne,\valTwo)}}{\MontypeOne}{\MonlcOne}$ as well as $\bscfgl{\circuitOne}{\assOne}{\apply{\dlvalTwo_1,\dots,\dlvalTwo_n}(\valOne,\valTwo)}\bseval\bscfgr{\circuitTwo}{\MonvalOne}$. By the definition of well-typedness we know $\assOne\in\cassset_\dlvaltreeOne,\namesof{\assOne}(\dlvaltreeOne)\cap\namesof{}(\dlvaltreeTwo)=\emptyset,\circjudgment{\circuitOne}{\dlvaltreeOne}{\lcOne}{\MonlcOne\Monmerge{\assOne}\lcOne'}$ and $\compjudgment{\emptycontext}{\lcOne'}{\dlvaltreeTwo}{\apply{\dlvalTwo_1,\dots,\dlvalTwo_n}(\valOne,\valTwo)}{\MontypeOne}$. By inversion of the applicable rules we have
		\begin{equation}
			\inference[apply]
		{
			(\circuitTwo,\MonvalOne) = \append(\circuitOne,\assOne,\vec\labTwo,\boxedcirc{\vec\labOne}{\circuitTwo'}{\MonmvalOne}{\dlvaltreeTwo'},\dlvalTwo_1,\dots,\dlvalTwo_n)
		}
		{
			\bscfgl{\circuitOne}{\assOne}{\apply{\dlvalTwo_1,\dots,\dlvalTwo_n}(\boxedcirc{\vec\labOne}{\circuitTwo'}{\MonmvalOne}{\dlvaltreeTwo'},\vec\labTwo)} \bseval \bscfgr{\circuitTwo}{\MonvalOne}
		}
		\end{equation}
		\begin{equation}
			\inference[apply]
		{
			\valjudgment{\emptycontext}{\emptycontext}{\boxedcirc{\vec\labOne}{\circuitTwo'}{\MonmvalOne}{\dlvaltreeTwo'}}{\circt{\dlvaltreeTwo'}(\mtypeOne,\MonmtypeTwo)}
			&
			\valjudgment{\emptycontext}{\lcOne'}{\vec\labTwo}{\mtypeOne}
			\\
			\namesof{}(\dlvaltreeTwo')=\{\dlvalOne_1,\dots,\dlvalOne_n\}
			&
			\alpharenamingOne=\dlvalTwo_1/\dlvalOne_1,\dots,\dlvalTwo_n/\dlvalOne_n
		}
		{
			\compjudgment{\emptycontext}{\lcOne'}{\alpharenamed{\dlvaltreeTwo'}{\alpharenamingOne}}{\apply{\dlvalTwo_1,\dots,\dlvalTwo_n}(\boxedcirc{\vec\labOne}{\circuitTwo'}{\MonmvalOne}{\dlvaltreeTwo'},\vec\labTwo)}{\alpharenamed{\MonmtypeTwo}{\alpharenamingOne}}
		}
		\end{equation}
		
		Thus $\valOne\equiv\boxedcirc{\vec\labOne}{\circuitTwo'}{\MonmvalOne}{\dlvaltreeTwo'},\valTwo\equiv\vec\labTwo,\dlvaltreeTwo\equiv\alpharenamed{\dlvaltreeTwo'}{\alpharenamingOne}$ and $\MontypeOne\equiv\alpharenamed{\MonmtypeTwo}{\alpharenamingOne}$. We know that the invocation of $\append(\circuitOne,\assOne,\vec\labTwo,\boxedcirc{\vec\labOne}{\circuitTwo'}{\MonmvalOne}{\dlvaltreeTwo'},\dlvalTwo_1,\dots,\dlvalTwo_n)$ follows the following procedure:
		\begin{enumerate}
			
			\item Finds $\boxedcirc{\vec\labTwo}{\circuitTwo''}{\MonmvalOne'}{\dlvaltreeTwo'} \cong \boxedcirc{\vec\labOne}{\circuitTwo'}{\MonmvalOne}{\dlvaltreeTwo'}$ such that all labels occurring in $\circuitTwo''$, but not in $\vec\labTwo$, are fresh in $\circuitOne$. By Proposition \ref{lem: apply-1} we know that $\valjudgment{\emptycontext}{\emptycontext}{\boxedcirc{\vec\labOne}{\circuitTwo'}{\MonmvalOne}{\dlvaltreeTwo'}}{\circt{\dlvaltreeTwo'}(\mtypeOne,\MonmtypeTwo)}$ implies $\valjudgment{\emptycontext}{\emptycontext}{\boxedcirc{\vec\labTwo}{\circuitTwo''}{\MonmvalOne'}{\dlvaltreeTwo'}}{\circt{\dlvaltreeTwo'}(\mtypeOne,\MonmtypeTwo)}$. By inversion we thus get
			\begin{equation}
				\inference[circ]
			{
				\circjudgment{\circuitTwo''}{\dlvaltreeTwo'}{\lcOne'}{\MonlcTwo}
				&
				\valjudgment{\emptycontext}{\lcOne'}{\vec\labTwo}{\mtypeOne}
				&
				\Monvaljudgment{\emptycontext}{\MonlcTwo}{\dlvaltreeTwo'}{\MonmvalOne'}{\MonmtypeTwo}
			}
			{
				\valjudgment{\emptycontext}{\emptycontext}{\boxedcirc{\vec\labTwo}{\circuitTwo''}{\MonmvalOne'}{\dlvaltreeTwo'}}{\circt{\dlvaltreeTwo'}(\mtypeOne,\MonmtypeTwo)}
			}
			\end{equation}
			
			\item Computes $\alpharenamed{\circuitTwo''}{\alpharenamingOne}$ and $\alpharenamed{\MonmvalOne''}{\alpharenamingOne}=\MonvalOne$. By Lemma \ref{lem: apply-2.1} we know $\circjudgment{\alpharenamed{\circuitTwo''}{\alpharenamingOne}}{\alpharenamed{\dlvaltreeTwo'}{\alpharenamingOne}}{\lcOne'}{\alpharenamed{\MonlcTwo}{\alpharenamingOne}}$ and by Lemma \ref{lem: apply-2.2} we get $\Monvaljudgment{\emptycontext}{\alpharenamed{\MonlcTwo}{\alpharenamingOne}}{\alpharenamed{\dlvaltreeTwo'}{\alpharenamingOne}}{\MonvalOne}{\alpharenamed{\MonmtypeTwo}{\alpharenamingOne}}$.
			
			\item Computes $\circuitOne\cconcat_\assOne \alpharenamed{\circuitTwo''}{\alpharenamingOne}=\circuitTwo$ and returns $(\circuitTwo,\MonvalOne)$.
			By Lemma \ref{lem: apply-3} we know $\circjudgment{\circuitTwo}{\Monflatten{\Moncomp{\dlvaltreeOne}{\alpharenamed{\dlvaltreeTwo'}{\alpharenamingOne}}{}{\{\assOne\}}}}{\lcOne}{\Monflatten{\Moncomp{\MonlcOne'}{\alpharenamed{\MonlcTwo}{\alpharenamingOne},\lcOne''}{}{\{\assOne\}}}}$, where $\MonlcOne=\Monflatten{\Moncomp{\MonlcOne'}{\lcOne''}{}{\{\assOne\}}}$.
		\end{enumerate}
	
		Since $\Monflatten{\Moncomp{\MonlcOne'}{\Moncomp{\alpharenamed{\dlvaltreeTwo'}{\alpharenamingOne}}{\lcOne''}{\assTwo}{}}{}{\{\assOne\}}} = \MonlcOne\Monsplit{\assOne}\alpharenamed{\dlvaltreeTwo'}{\alpharenamingOne}$, we conclude $\rbscfgjudgment{\lcOne}{\Monflatten{\Moncomp{\dlvaltreeOne}{\alpharenamed{\dlvaltreeTwo'}{\alpharenamingOne}}{}{\{\assOne\}}}}{}{\assOne}{\bscfgr{\circuitTwo}{\MonmvalOne}}{\alpharenamed{\MonmtypeTwo}{\alpharenamingOne}}{\MonlcOne\Monsplit{\assOne}\alpharenamed{\dlvaltreeTwo'}{\alpharenamingOne}}$.
		
		\subparagraph*{Case of \textit{let}}
		In this case we have $\lbscfgjudgment{\lcOne}{\dlvaltreeOne}{\dlvaltreeTwo}{\bscfgl{\circuitOne}{\assOne}{\letin{\varOne}{\termOne}{\MontermOne}}}{\MontypeOne}{\MonlcOne}$ and $\bscfgl{\circuitOne}{\assOne}{\letin{\varOne}{\termOne}{\MontermOne}}\bseval\bscfgr{\circuitTwo}{\MonvalOne}$. By the definition of well-typedness we know that $\assOne\in\cassset_\dlvaltreeOne,\namesof{\assOne}(\dlvaltreeOne)\cap\namesof{}(\dlvaltreeTwo)=\emptyset,\circjudgment{\circuitOne}{\dlvaltreeOne}{\lcOne}{\MonlcOne\Monmerge{\assOne}\lcOne'}$ and $\compjudgment{\emptycontext}{\lcOne'}{\dlvaltreeTwo}{\letin{\varOne}{\termOne}{\MontermOne}}{\MontypeOne}$. By inversion of the applicable rules we have
		
		\begin{equation}
			\inference[let]
		{\bscfgl{\circuitOne}{\assOne}{\termOne}\bseval\bscfgr{\circuitOne_1}{\MonvalOne_{\lindex}}
			&
			\MonvalOne_{\lindex}\in\condmonad_{\dlvaltreeTwo_{\lindex}}(\valset)
			&
			\MontermOne\in\condmonad_{\dlvaltreeTwo_{\lindex}}(\termset)
			\\
			\cassset_{\dlvaltreeTwo_{\lindex}} = \{\assOne_1,\dots,\assOne_n\}
			&
			\bscfgl{\circuitOne_i}{\assOne\cup\assOne_i}{\MontermOne(\assOne_i)[\MonvalOne_{\lindex}(\assOne_i)/\varOne]} \bseval \bscfgr{\circuitOne_{i+1}}{\MonvalTwo_i} \for i =1,\dots,n
		}
		%---------------------------------------------
		{\bscfgl{\circuitOne}{\assOne}{\letin{\varOne}{\termOne}{\MontermOne}} \bseval \bscfgr{\circuitOne_{n+1}}{\Monflatten{\Moncomp{\dlvaltreeTwo_{\lindex}}{\MonvalTwo_i}{\assOne_i}{}}}}
		\end{equation}				
		\begin{equation}
			\inference[let]
		{\compjudgment{\emptycontext}{\lcOne'_{\lindex}}{\dlvaltreeTwo_{\lindex}}{\termOne}{\MontypeTwo}
			&
			\MontermOne\in\condmonad_{\dlvaltreeTwo_{\lindex}}(\termset)
			&
			\Moncompjudgment{\varOne:\MontypeTwo}{\lcOne'_{\rindex}}{\Moncomp{\dlvaltreeTwo_{\lindex}}{\dlvaltreeTwo_{\assOne_i}}{\assOne_i}{}}{\MontermOne}{\MongenericObjectTwo}
		}
		%---------------------------------------------
		{\compjudgment{\emptycontext}{\lcOne'_{\lindex},\lcOne'_{\rindex}}{\Monflatten{\Moncomp{\dlvaltreeTwo_{\lindex}}{\dlvaltreeTwo_{\assOne_i}}{\assOne_i}{}}}{\letin{\varOne}{\termOne}{\MontermOne}}{\Monflatten{\MongenericObjectTwo}}}
		\end{equation}
		
		Note that, for the sake of brevity, when we know $\cassset_{\dlvaltreeTwo_{\lindex}} = \{\assOne_1,\dots,\assOne_n\}$ we write (with an abuse of notation) $\Moncomp{\dlvaltreeTwo_{\lindex}}{\MonvalTwo_i}{\assOne_i}{}$ as a sort of ``pattern matching'' to denote $\Moncomp{\dlvaltreeTwo_{\lindex}}{\MonvalTwo_\assTwo}{\assTwo}{\cassset_{\dlvaltreeTwo_{\lindex}}}$, where $\MonvalTwo_\assTwo=\MonvalTwo_i$ when $\assTwo=\assOne_i$. We have $\circuitTwo\equiv\circuitOne_{n+1},\MonvalOne\equiv\Monflatten{\Moncomp{\dlvaltreeTwo_{\lindex}}{\MonvalTwo_i}{\assOne_i}{}},\lcOne'\equiv(\lcOne'_{\lindex},\lcOne'_{\rindex}),\dlvaltreeTwo\equiv\Monflatten{\Moncomp{\dlvaltreeTwo_{\lindex}}{\dlvaltreeTwo_{\assOne_i}}{\assOne_i}{}}$ and $\MontypeOne\equiv\Monflatten{\MongenericObjectTwo}$. Suppose $\MonlcOne=\Moncomp{\MonlcOne'}{\lcOne''}{}{\{\assOne\}}$ and thus $\MonlcOne\Monmerge{\assOne}(\lcOne'_{\lindex},\lcOne'_{\rindex})=\Moncomp{\MonlcOne'}{\lcOne'_{\lindex},\lcOne'_{\rindex},\lcOne''}{}{\{\assOne\}}$. By the definition of well-typed left configuration we have $\lbscfgjudgment{\lcOne}{\dlvaltreeOne}{\dlvaltreeTwo_{\lindex}}{\bscfgl{\circuitOne}{\assOne}{\termOne}}{\MontypeTwo}{\Moncomp{\MonlcOne'}{\lcOne'_{\rindex},\lcOne''}{}{\{\assOne\}}}$.
		By inductive hypothesis, we therefore get $\rbscfgjudgment{\lcOne}{\Monflatten{\Moncomp{\dlvaltreeOne}{\dlvaltreeTwo_{\lindex}}{}{\{\assOne\}}}}{}{\assOne}{\bscfgr{\circuitOne_1}{\MonvalOne_{\lindex}}}{\MontypeTwo}{\Monflatten{\Moncomp{\MonlcOne'}{\Monconst{(\lcOne'_{\rindex},\lcOne'')}{\dlvaltreeTwo_{\lindex}}}{}{\{\assOne\}}}}$, where $\Monconst{(\lcOne'_{\rindex},\lcOne'')}{\dlvaltreeTwo_{\lindex}}$ is shorthand for $\Moncomp{\dlvaltreeTwo_{\lindex}}{\lcOne'_{\rindex},\lcOne''}{\assTwo}{}$. This entails $\circjudgment{\circuitOne_1}{\Monflatten{\Moncomp{\dlvaltreeOne}{\dlvaltreeTwo_{\lindex}}{}{\{\assOne\}}}}{\lcOne}{\Monflatten{\Moncomp{\MonlcOne'}{\MonlcTwo,\lcOne'_{\rindex},\lcOne''}{}{\{\assOne\}}}}$ for some $\MonlcTwo\in\condmonad_{\dlvaltreeTwo_{\lindex}}(\lcset)$ such that $\Monvaljudgment{\emptycontext}{\MonlcTwo}{\dlvaltreeTwo_{\lindex}}{\MonvalOne_{\lindex}}{\MontypeTwo}$. Knowing this and the last premise of the \textit{let} typing rule above, by expanding the lifted judgments and applying Lemma \ref{lem: substitution} $n$ times, we get
		
		\begin{equation}\label{eq: lifted-sub}
			\forall\assOne_i\in\cassset_{\dlvaltreeTwo_{\lindex}}.\compjudgment{\emptycontext}{\MonlcTwo(\assOne_i),\lcOne'_{\rindex}}{\dlvaltreeTwo_{\assOne_i}}{\MontermOne(\assOne_i)[\MonvalOne_{\lindex}(\assOne_i)/\varOne]}{\MongenericObjectTwo(\assOne_i)}.
		\end{equation}
		
		%Branch a1	
		Consider now the configuration $\bscfgl{\circuitOne_1}{\assOne\assmerge\assOne_1}{\MontermOne(\assOne_1)[\MonvalOne_{\lindex}(\assOne_1)/\varOne]}$. Because by hypothesis $\assOne\in\cassset_{\dlvaltreeOne}$ and $\assOne_1\in\cassset_{\dlvaltreeTwo_{\lindex}}$, we know $\assOne\assmerge\assOne_1\in\cassset_{\Monflatten{\Moncomp{\dlvaltreeOne}{\dlvaltreeTwo_{\lindex}}{}{\{\assOne\}}}}$. Furthermore, we have $\namesof{\assOne\cup\assOne_1}(\Monflatten{\Moncomp{\dlvaltreeOne}{\dlvaltreeTwo_{\lindex}}{}{\{\assOne\}}}) \cap \namesof{}(\dlvaltreeTwo_{\assOne_1}) = (\namesof{\assOne}(\dlvaltreeOne)\cup\namesof{\assOne_1}(\dlvaltreeTwo_{\lindex}))\cap\namesof{}(\dlvaltreeTwo_{\assOne_1}) = \emptyset$, since $\namesof{\assOne}(\dlvaltreeOne)\cap\namesof{}(\dlvaltreeTwo_{\assOne_1}) \subseteq \namesof{\assOne}(\dlvaltreeOne)\cap\namesof{}(\dlvaltreeTwo)=\emptyset$ by hypothesis and $\namesof{\assOne_1}(\dlvaltreeTwo_{\lindex})\cap\namesof{}(\dlvaltreeTwo_{\assOne_1})=\emptyset$ since the flattening of $\Moncomp{\dlvaltreeTwo_{\lindex}}{\dlvaltreeTwo_{\assOne_i}}{\assOne_i}{}$ is well-defined in the \textit{let} typing rule above.
		Lastly, consider again $\circjudgment{\circuitOne_1}{\Monflatten{\Moncomp{\dlvaltreeOne}{\dlvaltreeTwo_{\lindex}}{}{\{\assOne\}}}}{\lcOne}{\Monflatten{\Moncomp{\MonlcOne'}{\MonlcTwo,\lcOne'_{\rindex},\lcOne''}{}{\{\assOne\}}}}$ and rewrite $\Monflatten{\Moncomp{\MonlcOne'}{\MonlcTwo,\lcOne'_{\rindex},\lcOne''}{}{\{\assOne\}}}$ first as $\Monflatten{\Moncomp{\MonlcOne'}{\Moncomp{\dlvaltreeTwo_{\lindex}}{\MonlcTwo(\assOne_j),\lcOne'_{\rindex},\lcOne''}{\assOne_j}{\{\assOne_j\mid j=1\dots n\}}}{}{\{\assOne\}}}$ and then as $\Moncomp{\Monflatten{\Moncomp{\MonlcOne'}{\dlvaltreeTwo_{\lindex}}{}{\{\assOne\}}}}{\MonlcTwo(\assOne_j),\lcOne'_{\rindex},\lcOne''}{\assOne\assmerge\assOne_j}{\{\assOne\assmerge\assOne_j\mid j=1\dots n\}}$, by Lemma \ref{lem: comma extraction}.
		This, together with Equation \ref{eq: lifted-sub}, lets us say
		
		\begin{equation}
			\begin{aligned}
				\lbscfgjudgment{\lcOne}{\Monflatten{\Moncomp{\dlvaltreeOne}{\dlvaltreeTwo_{\lindex}}{}{\{\assOne\}}}}{\dlvaltreeTwo_{\assOne_1}}{&\bscfgl{\circuitOne_1}{\assOne\assmerge\assOne_1}{\MontermOne(\assOne_1)[\MonvalOne_{\lindex}(\assOne_1)/\varOne]}}{\MongenericObjectTwo(\assOne_1)}{\\&\Moncomp{\Moncomp{\Monflatten{\Moncomp{\MonlcOne'}{\dlvaltreeTwo_{\lindex}}{}{\{\assOne\}}}}{\lcOne''}{}{\{\assOne\assmerge\assOne_1\}}}{\MonlcTwo(\assOne_j),\lcOne'_{\rindex},\lcOne''}{\assOne\assmerge\assOne_j}{\{\assOne\assmerge\assOne_j\mid j=2\dots n\}}}.
			\end{aligned}
		\end{equation}
	
		For brevity, let $\dlvaltreeOne_{\lindex} = \Monflatten{\Moncomp{\dlvaltreeOne}{\dlvaltreeTwo_{\lindex}}{}{\{\assOne\}}}$ and $\MonlcOne_{\lindex} = \Monflatten{\Moncomp{\MonlcOne'}{\dlvaltreeTwo_{\lindex}}{}{\{\assOne\}}}$ in the rest of this proof case. By inductive hypothesis, we get
		
		\begin{equation}
			\rbscfgjudgment{\lcOne}{\Monflatten{\Moncomp{\dlvaltreeOne_{\lindex}}{\dlvaltreeTwo_{\assOne_1}}{}{\{\assOne\assmerge\assOne_1\}}}}{}{\assOne\assmerge\assOne_1}{\bscfgr{\circuitOne_2}{\MonvalTwo_1}}{\MongenericObjectTwo(\assOne_1)}{\Monflatten{\Moncomp{\Moncomp{\MonlcOne_{\lindex}}{\Monconst{\lcOne''}{\dlvaltreeTwo_{\assOne_1}}}{}{\{\assOne\assmerge\assOne_1\}}}{\MonlcTwo(\assOne_j),\lcOne'_{\rindex},\lcOne''}{\assOne\assmerge\assOne_j}{\{\assOne\assmerge\assOne_j\mid j=2\dots n\}}}}.
		\end{equation}
	
		By the definition of well-typed right configuration, this entails
		
		\begin{equation}
			\circjudgment{\circuitOne_2}{\Monflatten{\Moncomp{\dlvaltreeOne_{\lindex}}{\dlvaltreeTwo_{\assOne_1}}{}{\{\assOne\assmerge\assOne_1\}}}}{\lcOne}{\Monflatten{\Moncomp{\Moncomp{\MonlcOne_{\lindex}}{\MonlcTwo'_1,\lcOne''}{}{\{\assOne\assmerge\assOne_1\}}}{\MonlcTwo(\assOne_j),\lcOne'_{\rindex},\lcOne''}{\assOne\assmerge\assOne_j}{\{\assOne\assmerge\assOne_j\mid j=2\dots n\}}}}
		\end{equation}
		for some $\MonlcTwo'_1\in\condmonad_{\dlvaltreeTwo_{\assOne_1}}(\lcset)$ such that $\Monvaljudgment{\emptycontext}{\MonlcTwo'_1}{\dlvaltreeTwo_{\assOne_1}}{\MonvalTwo_1}{\MongenericObjectTwo(\assOne_1)}$.
		Since the various label contexts $(\MonlcTwo(\assOne_j),\lcOne'_{\rindex},\lcOne'')$, where $j\neq 1$, are not tree-like structures and are composed on branches disjoint from $\assOne_1$, we can factor them out of the outermost flattening, thus rewriting the output of $\circuitOne_2$ as 
		\begin{equation}
			\Moncomp{\Monflatten{\Moncomp{\MonlcOne_{\lindex}}{\MonlcTwo'_1,\lcOne''}{}{\{\assOne\assmerge\assOne_1\}}}}{\MonlcTwo(\assOne_j),\lcOne'_{\rindex},\lcOne''}{\assOne\assmerge\assOne_j}{\{\assOne\assmerge\assOne_j\mid j=2\dots n\}}.
		\end{equation}
		
		%Branch a2
		Moving on, consider the configuration $\bscfgl{\circuitOne_2}{\assOne\assmerge\assOne_2}{\MontermOne(\assOne_2)[\MonvalOne_{\lindex}(\assOne_2)/\varOne]}$. We know that $\assOne_1$ and $\assOne_2$ are distinct, so because by hypothesis $\assOne\in\cassset_{\dlvaltreeOne}$ and $\assOne_2\in\cassset_{\dlvaltreeTwo_{\lindex}}$, we know $\assOne\assmerge\assOne_2\in\cassset_{\Monflatten{\Moncomp{\dlvaltreeOne_{\lindex}}{\dlvaltreeTwo_{\assOne_1}}{}{\{\assOne\assmerge\assOne_1\}}}}$. Furthermore, we have $\namesof{\assOne\cup\assOne_2}(\Monflatten{\Moncomp{\dlvaltreeOne_{\lindex}}{\dlvaltreeTwo_{\assOne_1}}{}{\{\assOne\assmerge\assOne_1\}}}) \cap \namesof{}(\dlvaltreeTwo_{\assOne_2}) = \namesof{\assOne\cup\assOne_2}(\Monflatten{\Moncomp{\Monflatten{\Moncomp{\dlvaltreeOne}{\dlvaltreeTwo_{\lindex}}{}{\{\assOne\}}}}{\dlvaltreeTwo_{\assOne_1}}{}{\{\assOne\assmerge\assOne_1\}}}) \cap \namesof{}(\dlvaltreeTwo_{\assOne_2}) = (\namesof{\assOne}(\dlvaltreeOne)\cup\namesof{\assOne_2}(\dlvaltreeTwo_{\lindex}))\cap\namesof{}(\dlvaltreeTwo_{\assOne_2}) = \emptyset$.
		This, together with $\circuitOne_2$'s signature and Equation \ref{eq: lifted-sub}, lets us say
 \begin{equation}
		 	\begin{aligned}
		 		\lbscfgjudgment{\lcOne}{\Monflatten{\Moncomp{\dlvaltreeOne_{\lindex}}{\dlvaltreeTwo_{\assOne_1}}{}{\{\assOne\assmerge\assOne_1\}}}}{\dlvaltreeTwo_{\assOne_2}}{&\bscfgl{\circuitOne_2}{\assOne\assmerge\assOne_2}{\MontermOne(\assOne_2)[\MonvalOne_{\lindex}(\assOne_2)/\varOne]}}{\MongenericObjectTwo(\assOne_2)}{\\&\Moncomp{\Moncomp{\Monflatten{\Moncomp{\MonlcOne_{\lindex}}{\MonlcTwo'_1,\lcOne''}{}{\{\assOne\assmerge\assOne_1\}}}}{\lcOne''}{}{\{\assOne\assmerge\assOne_2\}}}{\MonlcTwo(\assOne_j),\lcOne'_{\rindex},\lcOne''}{\assOne\assmerge\assOne_j}{\{\assOne\assmerge\assOne_j\mid j=3\dots n\}}}
		 	\end{aligned}
 \end{equation}
		and by inductive hypothesis we get
		\begin{equation}
			\begin{aligned}
				\rbscfgjudgment{\lcOne}{\Monflatten{\Moncomp{\Monflatten{\Moncomp{\dlvaltreeOne_{\lindex}}{\dlvaltreeTwo_{\assOne_1}}{}{\{\assOne\assmerge\assOne_1\}}}}{\dlvaltreeTwo_{\assOne_2}}{}{\{\assOne\assmerge\assOne_2\}}}}{}{\dlvaltreeTwo_{\assOne_2}}{&\bscfgr{\circuitOne_3}{\MonvalTwo_2}}{\MongenericObjectTwo(\assOne_2)}{\\&\Moncomp{\Monflatten{\Moncomp{\Monflatten{\Moncomp{\MonlcOne_{\lindex}}{\MonlcTwo'_1,\lcOne''}{}{\{\assOne\assmerge\assOne_1\}}}}{\Monconst{\lcOne''}{\dlvaltreeTwo_{\assOne_2}}}{}{\{\assOne\assmerge\assOne_2\}}}}{\MonlcTwo(\assOne_j),\lcOne'_{\rindex},\lcOne''}{\assOne\assmerge\assOne_j}{\{\assOne\assmerge\assOne_j\mid j=3\dots n\}}},
			\end{aligned}
		\end{equation}
		which entails that $\circuitOne_3$ has signature
		\begin{equation}
			\begin{aligned}
				\circjudgment{\circuitOne_3&}{\Monflatten{\Moncomp{\Monflatten{\Moncomp{\dlvaltreeOne_{\lindex}}{\dlvaltreeTwo_{\assOne_1}}{}{\{\assOne\assmerge\assOne_1\}}}}{\dlvaltreeTwo_{\assOne_2}}{}{\{\assOne\assmerge\assOne_2\}}}}{\\&\lcOne}{\Moncomp{\Monflatten{\Moncomp{\Monflatten{\Moncomp{\MonlcOne_{\lindex}}{\MonlcTwo'_1,\lcOne''}{}{\{\assOne\assmerge\assOne_1\}}}}{\MonlcTwo'_2,\lcOne''}{}{\{\assOne\assmerge\assOne_2\}}}}{\MonlcTwo(\assOne_j),\lcOne'_{\rindex},\lcOne''}{\assOne\assmerge\assOne_j}{\{\assOne\assmerge\assOne_j\mid j=3\dots n\}}}
			\end{aligned}
		\end{equation}
		for some $\MonlcTwo'_2\in\condmonad_{\dlvaltreeTwo_{\assOne_2}}(\lcset)$ such that $\Monvaljudgment{\emptycontext}{\MonlcTwo'_2}{\dlvaltreeTwo_{\assOne_2}}{\MonvalTwo_2}{\MongenericObjectTwo(\assOne_2)}$.
		
		Now, when we consider the configuration $\bscfgl{\circuitOne_3}{\assOne\assmerge\assOne_3}{\MontermOne(\assOne_3)[\MonvalOne_{\lindex}(\assOne_3)/\varOne]}$ we can apply to it the same reasoning that we applied to $\bscfgl{\circuitOne_2}{\assOne\assmerge\assOne_2}{\MontermOne(\assOne_2)[\MonvalOne_{\lindex}(\assOne_2)/\varOne]}$ after we considered $\bscfgl{\circuitOne_1}{\assOne\assmerge\assOne_1}{\MontermOne(\assOne_1)[\MonvalOne_{\lindex}(\assOne_1)/\varOne]}$. In fact, we can apply this kind of reasoning pattern iteratively to the remaining $\bscfgl{\circuitOne_i}{\assOne\assmerge\assOne_i}{\MontermOne(\assOne_i)[\MonvalOne_{\lindex}(\assOne_i)/\varOne]}$, for $i=3\dots n$. After $n-2$ additional iterations, we end up with
		\begin{equation}
			\circjudgment
		{\circuitOne_{n+1}}
		{\Monflatten{\Moncomp{\Monflatten{\cdots\Monflatten{\Moncomp{\dlvaltreeOne_{\lindex}}{\dlvaltreeTwo_{\assOne_1}}{}{\{\assOne\assmerge\assOne_1\}}}\cdots}}{\dlvaltreeTwo_{\assOne_n}}{}{\{\assOne\assmerge\assOne_n\}}}}
		{\lcOne}
		{\Monflatten{\Moncomp{\Monflatten{\cdots\Monflatten{\Moncomp{\MonlcOne_{\lindex}}{\MonlcTwo'_1,\lcOne''}{}{\{\assOne\assmerge\assOne_1\}}}\cdots}}{\MonlcTwo'_n,\lcOne''}{}{\{\assOne\assmerge\assOne_n\}}}}.
		\end{equation}
	
		Since all the assignments of the form $\assOne\assmerge\assOne_j$ are by definition distinct, we can rewrite $\circuitOne_{n+1}$'s signature by replacing the $n$ flattening operations with just one final flattening. If we also expand $\dlvaltreeOne_{\lindex}$ and $\MonlcOne_{\lindex}$ back, we obtain the following signature:
		\begin{equation}
			\circjudgment
		{\circuitOne_{n+1}}
		{\Monflatten{\Moncomp{\Monflatten{\Moncomp{\dlvaltreeOne}{\dlvaltreeTwo_{\lindex}}{}{\{\assOne\}}}}{\dlvaltreeTwo_{\assOne_i}}{\assOne\assmerge\assOne_i}{\{\assOne\assmerge\assOne_i \mid i=1\dots n\}}}}
		{\lcOne}
		{\Monflatten{\Moncomp{\Monflatten{\Moncomp{\MonlcOne'}{\dlvaltreeTwo_{\lindex}}{}{\{\assOne\}}}}{\MonlcTwo'_i,\lcOne''}{\assOne\assmerge\assOne_i}{\{\assOne\assmerge\assOne_i \mid i=1\dots n\}}}}.
		\end{equation}
		
		Next, by Lemma \ref{lem: comma extraction}, we further rewrite the signature above as
		\begin{equation}
			\circjudgment
		{\circuitOne_{n+1}}
		{\Monflatten{\Monflatten{\Moncomp{\dlvaltreeOne}{\Moncomp{\dlvaltreeTwo_{\lindex}}{\dlvaltreeTwo_{\assOne_i}}{\assOne_i}{\cassset_{\dlvaltreeTwo_{\lindex}}}}{}{\{\assOne\}}}}}
		{\lcOne}
		{\Monflatten{\Monflatten{\Moncomp{\MonlcOne'}{\Moncomp{\dlvaltreeTwo_{\lindex}}{\MonlcTwo'_i,\lcOne''}{\assOne_i}{\cassset_{\dlvaltreeTwo_{\lindex}}}}{}{\{\assOne\}}}}}.
		\end{equation}
		
		Furthermore, because $\dlvaltreeOne$ and $\MonlcOne'$ only have tree-like leaves where we explicitly attached them, we can write
		\begin{equation}
			\circjudgment
		{\circuitOne_{n+1}}
		{\Monflatten{\Moncomp{\dlvaltreeOne}{\Monflatten{\Moncomp{\dlvaltreeTwo_{\lindex}}{\dlvaltreeTwo_{\assOne_i}}{\assOne_i}{\cassset_{\dlvaltreeTwo_{\lindex}}}}}{}{\{\assOne\}}}}
		{\lcOne}
		{\Monflatten{\Moncomp{\MonlcOne'}{\Monflatten{\Moncomp{\dlvaltreeTwo_{\lindex}}{\MonlcTwo'_i,\lcOne''}{\assOne_i}{\cassset_{\dlvaltreeTwo_{\lindex}}}}}{}{\{\assOne\}}}}.
		\end{equation}
		
		Lastly, because $\dlvaltreeTwo\equiv\Monflatten{\Moncomp{\dlvaltreeTwo_{\lindex}}{\dlvaltreeTwo_{\assOne_i}}{\assOne_i}{\cassset_{\dlvaltreeTwo_{\lindex}}}}$, we can write
		\begin{equation}
			\circjudgment
		{\circuitOne_{n+1}}
		{\Monflatten{\Moncomp{\dlvaltreeOne}{\dlvaltreeTwo}{}{\{\assOne\}}}}
		{\lcOne}
		{\Monflatten{\Moncomp{\MonlcOne'}{\Monflatten{\Moncomp{\dlvaltreeTwo_{\lindex}}{\MonlcTwo'_i}{\assOne_i}{\cassset_{\dlvaltreeTwo_{\lindex}}}},\lcOne''^\dlvaltreeTwo}{}{\{\assOne\}}}},
		\end{equation}
		
		where $\Monvaljudgment{\emptycontext}{\MonlcTwo'_i}{\dlvaltreeTwo_{\assOne_i}}{\MonvalTwo_i}{\MongenericObjectTwo(\assOne_i)}$ for $i=1\dots n$. By Lemma \ref{lem: meta-lifted-judgment}, this allows us to say $\Monvaljudgment{\emptycontext}{\Monflatten{\Moncomp{\dlvaltreeTwo_{\lindex}}{\MonlcTwo'_i}{\assOne_i}{\cassset_{\dlvaltreeTwo_{\lindex}}}}}{\dlvaltreeTwo}{\Monflatten{\Moncomp{\dlvaltreeTwo_{\lindex}}{\MonvalTwo_i}{\assOne_i}{\cassset_{\dlvaltreeTwo_{\lindex}}}}}{\Monflatten{\MongenericObjectTwo}}$, by which we conclude
		\begin{equation}
			\rbscfgjudgment{\lcOne}{\Monflatten{\Moncomp{\dlvaltreeOne}{\dlvaltreeTwo}{}{\{\assOne\}}}}{}{\assOne}{\bscfgr{\circuitOne_{n+1}}{\Monflatten{\Moncomp{\dlvaltreeTwo_{\lindex}}{\MonvalTwo_i}{\assOne_i}{}}}}{\Monflatten{\MongenericObjectTwo}}{\Monflatten{\Moncomp{\MonlcOne'}{\lcOne''^\dlvaltreeTwo}{}{\{\assOne\}}}}.
		\end{equation}
	\end{proof}
}
		
\longversion{
	
	\subsection{Progress}
	
	To prove the progress result, we avail ourselves of the following lemmata.
	
	\begin{lemma}\label{lem: varset after alpha-renaming}
		For any renaming of lifted labels $\alpharenamingOne$ we have $\namesof{\alpharenamed{\assOne}{\alpharenamingOne}}(\alpharenamed{\dlvaltreeOne}{\alpharenamingOne})=\alpharenamingOne[\namesof{\assOne}(\dlvaltreeOne)]$.
	\end{lemma}
	
	\begin{lemma}\label{lem: same labels tuple and context}
		If $\valjudgment{\pcontextOne}{\lcOne}{\vec\labOne}{\mtypeOne}$, then $\dom(\lcOne)$ contains all and only the labels that occur in $\vec\labOne$.
	\end{lemma}
	
	\begin{lemma}\label{lem: same type implies renaming possible}
		If $\valjudgment{\pcontextOne}{\lcOne}{\vec\labOne}{\mtypeOne}$ and $\valjudgment{\pcontextOne}{\lcTwo}{\vec\labTwo}{\mtypeOne}$, then there exists a renaming of labels $\renamingOne$ such that $\renamed{\vec\labOne}{\renamingOne}=\vec\labTwo$.
	\end{lemma}
	
	\begin{theorem}[Progress]\label{thm: Progress}
		If $\lbscfgjudgment{\lcOne}{\dlvaltreeOne}{\dlvaltreeTwo}{\bscfgl{\circuitOne}{\assOne}{\termOne}}{\MontypeOne}{\MonlcOne}$, then either $\exists\circuitTwo,\MonvalOne.\bscfgl{\circuitOne}{\assOne}{\termOne}\bseval\bscfgr{\circuitTwo}{\MonvalOne}$ or $\bscfgl{\circuitOne}{\assOne}{\termOne}\bsdiverge$.
	\end{theorem}
	\begin{proof}
		We prove the equivalent claim that if $\lbscfgjudgment{\lcOne}{\dlvaltreeOne}{\dlvaltreeTwo}{\bscfgl{\circuitOne}{\assOne}{\termOne}}{\MontypeOne}{\MonlcOne}$ and $\nexists\circuitTwo,\MonvalOne.\bscfgl{\circuitOne}{\assOne}{\termOne}\bseval\bscfgr{\circuitTwo}{\MonvalOne}$, then $\bscfgl{\circuitOne}{\assOne}{\termOne}\bsdiverge$. We proceed by coinduction and case analysis on $\termOne$. The cases in which $\termOne$ is an application, a destructuring or an use of return or force are fairly straightforward, so we focus on the cases in which $\termOne$ is a use of $\justboxt,\apply{}$, or a let statement.
		
		\subparagraph*{Case $\termOne\equiv\boxt{\mtypeOne}{}\valOne.$}
		In this case we have $\lbscfgjudgment{\lcOne}{\dlvaltreeOne}{\dlvaltreeTwo}{\bscfgl{\circuitOne}{\assOne}{\boxt{\mtypeOne}{}\valOne}}{\MontypeOne}{\MonlcOne}$ and $\nexists\circuitTwo,\MonvalOne.\bscfgl{\circuitOne}{\assOne}{\boxt{\mtypeOne}{}\valOne}\bseval\bscfgr{\circuitTwo}{\MonvalOne}$. By the definition of well-typedness we know that $\assOne\in\cassset_\dlvaltreeOne,\namesof{\assOne}(\dlvaltreeOne)\cap\namesof{}(\dlvaltreeTwo)=\emptyset,\circjudgment{\circuitOne}{\dlvaltreeOne}{\lcOne}{\MonlcOne\Monmerge{\assOne}\lcOne'}$ and $\compjudgment{\emptycontext}{\lcOne'}{\dlvaltreeTwo}{\boxt{\mtypeOne}{}\valOne}{\MontypeOne}$. By inversion of the \textit{box} and \textit{lift} rules we have
		\begin{equation}
			\inference[box]
		{
			\inference[lift]
			{\compjudgment{\emptycontext}{\emptycontext}{\emptytree}{\termTwo}{\Monunit{\mtypeOne \multimap_{\dlvaltreeTwo'} \MonmtypeTwo}}}
			{\valjudgment{\emptycontext}{\emptycontext}{\lift\termTwo}{\bang\Monunit{\mtypeOne \multimap_{\dlvaltreeTwo'} \MonmtypeTwo}}}
		}
		%---------------------------------------------
		{\compjudgment{\emptycontext}{\emptycontext}{\emptytree}{\boxt{\mtypeOne}{\dlvalsubOne} (\lift\termTwo)}{\Monunit{\circt{\dlvaltreeTwo'}(\mtypeOne,\MonmtypeTwo)}}}
		\end{equation}
		where $\valOne\equiv\lift\termTwo$ for some $\termTwo$ because $\lift\termTwo$ is the only kind of value that can be given type $\bang\Monunit{\mtypeOne \multimap_{\dlvaltreeTwo'} \MonmtypeTwo}$ in an empty typing context. We also have $\lcOne'\equiv\emptycontext,\dlvaltreeTwo\equiv\emptytree$ and $\MontypeOne\equiv\Monunit{\circt{\dlvaltreeTwo'}(\mtypeOne,\MonmtypeTwo)}$. Let $(\lcTwo,\vec\labOne)=\freshlabels(\mtypeOne)$, which implies $\mjudgment{\lcTwo}{\vec\labOne}{\mtypeOne}$ and therefore $\valjudgment{\emptycontext}{\lcTwo}{\vec\labOne}{\mtypeOne}$ by Proposition \ref{prop: m-judgment and val-judgment isomorphism}. We need to prove that $\bscfgl{\cinput{\lcTwo}}{\emptyassignment}{\letin{\varOne}{\termTwo}{\Monunit{\varOne\vec\labOne}}}$ is well-typed and cannot be evaluated. We know $\circjudgment{\cinput{\lcTwo}}{\emptytree}{\lcTwo}{\Monunit{\lcTwo}}$ and we derive the following:
		\begin{equation}
			\inference[let]
		{
			\compjudgment{\emptycontext}{\emptycontext}{\emptytree}{\termTwo}{\Monunit{\mtypeOne \multimap_{\dlvaltreeTwo'} \MonmtypeTwo}}
			&
			\inference[app]
			{
				\inference[var]{}
				{\valjudgment{\varOne:\mtypeOne\multimap_{\dlvaltreeTwo'}\MonmtypeTwo}{\emptycontext}{\varOne}{\mtypeOne\multimap_{\dlvaltreeTwo'}\MonmtypeTwo}}
				&
				\valjudgment{\emptycontext}{\lcTwo}{\vec\labOne}{\mtypeOne}			
			}
			{\compjudgment{\varOne:\mtypeOne\multimap_{\dlvaltreeTwo'}\MonmtypeTwo}{\lcTwo}{\dlvaltreeTwo'}{\varOne\vec\labOne}{\MonmtypeTwo}}
			\\
			\Monunit{\varOne\vec\labOne}\in\condmonad_\emptytree(\termset)
		}
		{\compjudgment{\emptycontext}{\lcTwo}{\dlvaltreeTwo'}{\letin{\varOne}{\termTwo}{\Monunit{\varOne\vec\labOne}}}{\MonmtypeTwo}}
		\end{equation}
		
		by which we can say $\lbscfgjudgment{\lcTwo}{\emptytree}{\dlvaltreeTwo'}{\bscfgl{\cinput{\lcTwo}}{\emptyassignment}{\letin{\varOne}{\termTwo}{\Monunit{\varOne\vec\labOne}}}}{\MonmtypeTwo}{\Monunit{\emptycontext}}$. Now, either there exist $\circuitTwo',\MonvalTwo$ such that $\bscfgl{\cinput{\lcTwo}}{\emptyassignment}{\letin{\varOne}{\termTwo}{\Monunit{\varOne\vec\labOne}}}\bseval\bscfgr{\circuitTwo'}{\MonvalTwo}$ or not. In the latter case, we immediately conclude the claim by coinduction, using the \textit{box} divergence rule. In the former case, by Theorem \ref{thm: SR} we have $\rbscfgjudgment{\lcTwo}{\dlvaltreeTwo'}{}{\emptyassignment}{\bscfgr{\circuitTwo'}{\MonvalTwo}}{\MonmtypeTwo}{\Moncomp{\dlvaltreeTwo'}{\emptycontext}{\assTwo}{}}$. This entails, among other things, that $\MonvalTwo$ is given M-type $\MonmtypeTwo$. By Proposition \ref{lem: m-type implies m-val}, this implies $\MonvalTwo\in\condmonad_{\dlvaltreeTwo'}(\mvalset)$, which allows us to derive $\bscfgl{\circuitOne}{\assOne}{\boxt{\mtypeOne}{}(\lift \termTwo)}\bseval\bscfgr{\circuitOne}{\boxedcirc{\vec\labOne}{\circuitTwo'}{\MonvalTwo}{\dlvaltreeTwo'}}$, contradicting the hypothesis.
		
		\subparagraph*{Case $\termOne\equiv\apply{\dlvalTwo_1,\dots,\dlvalTwo_n}(\valOne,\valTwo).$}
		In this case we have $\lbscfgjudgment{\lcOne}{\dlvaltreeOne}{\dlvaltreeTwo}{\bscfgl{\circuitOne}{\assOne}{\apply{\dlvalTwo_1,\dots,\dlvalTwo_n}(\valOne,\valTwo)}}{\MontypeOne}{\MonlcOne}$ and $\nexists\circuitTwo,\MonvalOne.\bscfgl{\circuitOne}{\assOne}{\apply{\dlvalTwo_1,\dots,\dlvalTwo_n}(\valOne,\valTwo)}\bseval\bscfgr{\circuitTwo}{\MonvalOne}$. By the definition of well-typedness we know that $\assOne\in\cassset_\dlvaltreeOne,\namesof{\assOne}(\dlvaltreeOne)\cap\namesof{}(\dlvaltreeTwo)=\emptyset,\circjudgment{\circuitOne}{\dlvaltreeOne}{\lcOne}{\MonlcOne\Monmerge{\assOne}\lcOne'}$ and $\compjudgment{\emptycontext}{\lcOne'}{\dlvaltreeTwo}{\apply{\dlvalTwo_1,\dots,\dlvalTwo_n}(\valOne,\valTwo)}{\MontypeOne}$. By inversion of the \textit{apply} and \textit{circ} rules we have
		\begin{equation}
			\inference[apply]
		{
			\inference[circ]
			{
				\circjudgment{\circuitTwo'}{\dlvaltreeTwo'}{\lcTwo}{\MonlcTwo}
				&
				\valjudgment{\emptycontext}{\lcTwo}{\vec\labOne}{\mtypeOne}
				&
				\Monvaljudgment{\emptycontext}{\MonlcTwo}{\dlvaltreeTwo'}{\MonmvalOne}{\MonmtypeTwo}
			}
			{
				\valjudgment{\emptycontext}{\emptycontext}{\boxedcirc{\vec\labOne}{\circuitTwo'}{\MonmvalOne}{\dlvaltreeTwo'}}{\circt{\dlvaltreeTwo'}(\mtypeOne,\MonmtypeTwo)}
			}
			&
			\valjudgment{\emptycontext}{\lcOne'}{\vec\labTwo}{\mtypeOne}
			\\
			\namesof{}(\dlvaltreeTwo')=\{\dlvalOne_1,\dots,\dlvalOne_n\}
			&
			\alpharenamingOne=\dlvalTwo_1/\dlvalOne_1,\dots,\dlvalTwo_n/\dlvalOne_n
		}
		{
			\compjudgment{\emptycontext}{\lcOne'}{\alpharenamed{\dlvaltreeTwo'}{\alpharenamingOne}}{\apply{\dlvalTwo_1,\dots,\dlvalTwo_n}(\boxedcirc{\vec\labOne}{\circuitTwo'}{\MonmvalOne}{\dlvaltreeTwo'},\vec\labTwo)}{\alpharenamed{\MonmtypeTwo}{\alpharenamingOne}}
		}
		\end{equation}
		where $\valOne\equiv\boxedcirc{\vec\labOne}{\circuitTwo'}{\MonmvalOne}{\dlvaltreeTwo'}$ and $\valTwo\equiv\vec\labTwo$ for some $\vec\labOne,\circuitTwo',\MonmvalOne,\dlvaltreeTwo'$ and $\vec\labTwo$, because $\boxedcirc{\vec\labOne}{\circuitTwo'}{\MonmvalOne}{\dlvaltreeTwo'}$ and $\vec\labTwo$ are the only kinds of value that can respectively be given types $\circt{\dlvaltreeTwo'}(\mtypeOne,\MonmtypeTwo)$ and $\mtypeOne$ in an empty typing context. We also have $\dlvaltreeTwo\equiv\alpharenamed{\dlvaltreeTwo'}{\alpharenamingOne}$ and $\MontypeOne\equiv\alpharenamed{\MonmtypeTwo}{\alpharenamingOne}$. We prove that the claim is vacuously true by showing that $\bscfgl{\circuitOne}{\assOne}{\apply{\dlvalTwo_1,\dots,\dlvalTwo_n}(\boxedcirc{\vec\labOne}{\circuitTwo'}{\MonmvalOne}{\dlvaltreeTwo'},\vec\labTwo)}$ always evaluates to some configuration. To do this, it suffices to show that $\append(\circuitOne,\assOne,\vec\labTwo,\boxedcirc{\vec\labOne}{\circuitTwo'}{\MonmvalOne}{\dlvaltreeTwo'},\dlvalTwo_1,\dots,\dlvalTwo_n)$ is always defined under the hypothesis. We start with the preconditions of $\append$. So far we already know $\circjudgment{\circuitOne}{\dlvaltreeOne}{\lcOne}{\MonlcOne\Monmerge{\assOne}\lcOne'},\namesof{}(\dlvaltreeTwo')=\{\dlvalOne_1,\dots,\dlvalOne_n\},\assOne\in\cassset_\dlvaltreeOne$ and, with the help of Lemma \ref{lem: varset after alpha-renaming}, $\namesof{\assOne}(\dlvaltreeOne)\cap\namesof{}(\dlvaltreeTwo) = \namesof{\assOne}(\dlvaltreeOne)\cap\namesof{}(\alpharenamed{\dlvaltreeTwo'}{\alpharenamingOne}) = \namesof{\assOne}(\dlvaltreeOne)\cap\alpharenamingOne[\{\dlvalOne_1,\dots,\dlvalOne_n\}] = \namesof{\assOne}(\dlvaltreeOne)\cap\{\dlvalTwo_1,\dots,\dlvalTwo_n\} = \emptyset$. We also know that the labels in $\vec\labTwo$ all occur in $(\MonlcOne\Monmerge{\assOne}\lcOne')(\assOne)$, since by definition $(\MonlcOne\Monmerge{\assOne}\lcOne')(\assOne)=\lcOne',\lcOne''$ for some $\lcOne''$ and by Lemma \ref{lem: same labels tuple and context} the judgment $\valjudgment{\emptycontext}{\lcOne'}{\vec\labTwo}{\mtypeOne}$ implies that all the labels in $\vec\labTwo$ occur in $\lcOne'$. Lastly, we know that $\dlvalTwo_1,\dots,\dlvalTwo_n$ are distinct, since otherwise $\alpharenamingOne$ would not be a valid permutation and the variable renamings in the \textit{apply} rule would be undefined. The preconditions are met, so we proceed with $\append$'s steps:
		\begin{enumerate}
			\item $\append$ must find $\boxedcirc{\vec\labTwo}{\circuitTwo''}{\MonmvalOne'}{\dlvaltreeTwo'} \cong \boxedcirc{\vec\labOne}{\circuitTwo'}{\MonmvalOne}{\dlvaltreeTwo'}$ such that the labels that occur in $\circuitTwo''$, but not in $\vec\labTwo$, are fresh in $\circuitOne$. Because $\vec\labOne$ and $\vec\labTwo$ share the same type $\mtypeOne$, by Lemma \ref{lem: same type implies renaming possible} we know that there exists a label renaming $\renamingOne$ such that $\renamed{\vec\labOne}{\renamingOne}=\vec\labTwo$. Next, let $\labsubOne_{\vec\labOne}$ be the set of labels that occur in $\vec\labOne$, $\labsubOne_{\circuitTwo'}$ the set of labels that occur in $\circuitTwo'$, $\labsubOne_{\vec\labTwo}$ the set of labels that occur in $\vec\labTwo$ and $\labsubOne_{\circuitOne}$ the set of labels that occur in $\circuitOne$. We have $\labsubOne_{\vec\labTwo}\subseteq\labsubOne_{\circuitOne}$ and $\labsubOne_{\vec\labOne}\subseteq\labsubOne_{\circuitTwo'}$. Because $\circuitOne$ and $\circuitTwo'$ are finite syntactic entities, it is always possible to find a set of labels $\labsubOne_f$ of cardinality $|\labsubOne_f| = |\labsubOne_{\circuitTwo'}\setminus\labsubOne_{\vec\labOne}\,|$ such that $\labsubOne_f\cap(\labsubOne_{\circuitOne}\cup\labsubOne_{\circuitTwo'})=\emptyset$. Let $f$ be any bijection between $\labsubOne_{\circuitTwo'}\setminus\labsubOne_{\vec\labOne}$ and $\labsubOne_f$, and let $\renamingOne'$ be any permutation of $\labsubOne_{\circuitTwo'}\cup\labsubOne_{\vec\labTwo}\cup\labsubOne_f$ satisfying the following: 
			\begin{equation}
				\renamingOne'(\labOne)=\renamingOne(\labOne) \textnormal{ for all } \labOne\in\labsubOne_{\vec\labOne}, \quad \renamingOne'(\labOne)=f(\labOne) \textnormal{ for all } \labOne\in\labsubOne_{\circuitTwo'}\setminus\labsubOne_{\vec\labOne}.
			\end{equation}
			
			We can then extend $\renamingOne'$ from $\labsubOne_{\circuitTwo'}\cup\labsubOne_{\vec\labTwo}\cup\labsubOne_f$ to $\labelset$ by ``padding'' it with the identity function. We obtain a renaming of labels $\renamingOne''$ such that $\renamed{\vec\labOne}{\renamingOne''} = \vec\labTwo$. If we further set $\circuitTwo''=\renamed{\circuitTwo'}{\renamingOne''}$ and $\MonmvalOne'=\renamed{\MonmvalOne}{\renamingOne''}$ we get $\boxedcirc{\vec\labTwo}{\circuitTwo''}{\MonmvalOne'}{\dlvaltreeTwo'} \cong \boxedcirc{\vec\labOne}{\circuitTwo'}{\MonmvalOne}{\dlvaltreeTwo'}$, where $\circuitTwo''$ has the required freshness properties.
			
			\item $\append$ must compute $\alpharenamed{\circuitTwo''}{\alpharenamingOne}$ and $\alpharenamed{\MonmvalOne'}{\alpharenamingOne}$. Because we know that $\alpharenamingOne$ is a valid renaming of lifted variables, and thus a valid permutation, this step cannot fail.
			
			\item $\append$ must return $(\circuitOne\cconcat_\assOne\alpharenamed{\circuitTwo''}{\alpharenamingOne}, \alpharenamed{\MonmvalOne'}{\alpharenamingOne})$. This step could fail if there existed an assignment $\assTwo$ occurring in $\alpharenamed{\circuitTwo''}{\alpharenamingOne}$ such that $\dom(\assOne)\cap\dom(\assTwo)\neq\emptyset$, in which case $\assOne\assmerge\assTwo$, and therefore $\circuitOne\cconcat_\assOne\alpharenamed{\circuitTwo''}{\alpharenamingOne}$, would be undefined. However, by Lemma \ref{lem: apply-1a} we know $\circjudgment{\circuitTwo''}{\dlvaltreeTwo'}{\renamed{\lcTwo}{\renamingOne}}{\renamed{\MonlcTwo}{\renamingOne}}$ and by Lemma \ref{lem: apply-2.1} we know $\circjudgment{\alpharenamed{\circuitTwo''}{\alpharenamingOne}}{\alpharenamed{\dlvaltreeTwo'}{\alpharenamingOne}}{\renamed{\lcTwo}{\renamingOne}}{\alpharenamed{(\renamed{\MonlcTwo}{\renamingOne})}{\alpharenamingOne}}$. Therefore, for every assignment $\assTwo$ occurring in $\alpharenamed{\circuitTwo''}{\alpharenamingOne}$ we have  $\dom(\assTwo)\subseteq\namesof{}(\alpharenamed{\dlvaltreeTwo'}{\alpharenamingOne})$, and because $\assOne\in\cassset_\dlvaltreeOne$ we also have $\dom(\assOne)=\namesof{\assOne}(\dlvaltreeOne)$. By hypothesis we know $\namesof{}(\alpharenamed{\dlvaltreeTwo'}{\alpharenamingOne})\cap\namesof{\assOne}(\dlvaltreeOne) = \emptyset$ and therefore $\dom(\assTwo)\cap\dom(\assOne)=\emptyset$, by which we conclude that this step always succeeds.
		\end{enumerate}
		This allows us to derive $\bscfgl{\circuitOne}{\assOne}{\apply{\dlvalTwo_1,\dots,\dlvalTwo_n}(\boxedcirc{\vec\labOne}{\circuitTwo'}{\MonmvalOne}{\dlvaltreeTwo'},\vec\labTwo)}\bseval\bscfgr{\circuitOne\cconcat_\assOne\alpharenamed{\circuitTwo''}{\alpharenamingOne}}{\alpharenamed{\MonmvalOne'}{\alpharenamingOne}}$ by the \textit{apply} rule, contradicting the hypothesis.
		
		\subparagraph*{Case $\termOne\equiv \letin{\varOne}{\termTwo}{\MontermOne}.$}
		In this case we have $\lbscfgjudgment{\lcOne}{\dlvaltreeOne}{\dlvaltreeTwo}{\bscfgl{\circuitOne}{\assOne}{\letin{\varOne}{\termTwo}{\MontermOne}}}{\MontypeOne}{\MonlcOne}$ and $\nexists\circuitTwo,\MonvalOne.\bscfgl{\circuitOne}{\assOne}{\letin{\varOne}{\termTwo}{\MontermOne}}\bseval\bscfgr{\circuitTwo}{\MonvalOne}$. By the definition of well-typedness we know that $\assOne\in\cassset_\dlvaltreeOne,\namesof{\assOne}(\dlvaltreeOne)\cap\namesof{}(\dlvaltreeTwo)=\emptyset,\circjudgment{\circuitOne}{\dlvaltreeOne}{\lcOne}{\MonlcOne\Monmerge{\assOne}\lcOne'}$ and $\compjudgment{\emptycontext}{\lcOne'}{\dlvaltreeTwo}{\letin{\varOne}{\termTwo}{\MontermOne}}{\MontypeOne}$. By inversion of the \textit{let} rule we have
		\begin{equation}
			\inference[let]
		{\compjudgment{\emptycontext}{\lcOne'_{\lindex}}{\dlvaltreeTwo_{\lindex}}{\termTwo}{\MontypeTwo}
			&
			\MontermOne\in\condmonad_{\dlvaltreeTwo_{\lindex}}(\termset)
			&
			\Moncompjudgment{\varOne:\MontypeTwo}{\lcOne'_{\rindex}}{\Moncomp{\dlvaltreeTwo_{\lindex}}{\dlvaltreeTwo_{\assOne_i}}{\assOne_i}{}}{\MontermOne}{\MongenericObjectTwo}
		}
		%---------------------------------------------
		{\compjudgment{\emptycontext}{\lcOne'_{\lindex},\lcOne'_{\rindex}}{\Monflatten{\Moncomp{\dlvaltreeTwo_{\lindex}}{\dlvaltreeTwo_{\assOne_i}}{\assOne_i}{}}}{\letin{\varOne}{\termTwo}{\MontermOne}}{\Monflatten{\MongenericObjectTwo}}}
		\end{equation}
	
		Note that, for the sake of brevity, we assume $\cassset_{\dlvaltreeTwo_{\lindex}} = \{\assOne_1,\dots,\assOne_n\}$ and write (with an abuse of notation) $\Moncomp{\dlvaltreeTwo_{\lindex}}{\dlvaltreeTwo_{\assOne_i}}{\assOne_i}{}$ as a sort of ``pattern matching'' to denote $\Moncomp{\dlvaltreeTwo_{\lindex}}{\dlvaltreeTwo_{\assTwo}}{\assTwo}{\cassset_{\dlvaltreeTwo_\lindex}}$, where $\dlvaltreeTwo_\assTwo=\dlvaltreeTwo_{\assOne_i}$ when $\assTwo=\assOne_i$. This notation will be particularly useful in less trivial compositions in the rest of this proof case. We have $\lcOne'\equiv(\lcOne'_{\lindex},\lcOne'_{\rindex}),\dlvaltreeTwo\equiv\Monflatten{\Moncomp{\dlvaltreeTwo_{\lindex}}{\dlvaltreeTwo_{\assOne_i}}{\assOne_i}{}}$ and $\MontypeOne\equiv\Monflatten{\MongenericObjectTwo}$. Suppose $\MonlcOne=\Moncomp{\MonlcOne'}{\lcOne''}{}{\{\assOne\}}$ and therefore $\MonlcOne\Monmerge{\assOne}(\lcOne'_{\lindex},\lcOne'_{\rindex})=\Moncomp{\MonlcOne'}{\lcOne'_{\lindex},\lcOne'_{\rindex},\lcOne''}{}{\{\assOne\}}$. By the definition of well-typed left configuration we have $\lbscfgjudgment{\lcOne}{\dlvaltreeOne}{\dlvaltreeTwo_{\lindex}}{\bscfgl{\circuitOne}{\assOne}{\termTwo}}{\MontypeTwo}{\Moncomp{\MonlcOne'}{\lcOne'_{\rindex},\lcOne''}{}{\{\assOne\}}}$.
		At this point, either there exist $\circuitOne_1,\MonvalOne_{\lindex}$ such that $\bscfgl{\circuitOne}{\assOne}{\termTwo}\bseval\bscfgr{\circuitOne_1}{\MonvalOne_{\lindex}}$ or not. In the latter case, we immediately conclude the claim by coinduction through the \textit{let-now} rule. In the former case, by Theorem \ref{thm: SR} we get $\rbscfgjudgment{\lcOne}{\Monflatten{\Moncomp{\dlvaltreeOne}{\dlvaltreeTwo_{\lindex}}{}{\{\assOne\}}}}{}{\assOne}{\bscfgr{\circuitOne_1}{\MonvalOne_{\lindex}}}{\MontypeTwo}{\Monflatten{\Moncomp{\MonlcOne'}{\Monconst{(\lcOne'_{\rindex},\lcOne'')}{\dlvaltreeTwo_{\lindex}}}{}{\{\assOne\}}}}$, where $\Monconst{(\lcOne'_{\rindex},\lcOne'')}{\dlvaltreeTwo_{\lindex}}$ is shorthand for $\Moncomp{\dlvaltreeTwo_{\lindex}}{\lcOne'_{\rindex},\lcOne''}{\assTwo}{}$. This entails $\circjudgment{\circuitOne_1}{\Monflatten{\Moncomp{\dlvaltreeOne}{\dlvaltreeTwo_{\lindex}}{}{\{\assOne\}}}}{\lcOne}{\Monflatten{\Moncomp{\MonlcOne'}{\MonlcTwo,\lcOne'_{\rindex},\lcOne''}{}{\{\assOne\}}}}$ for some $\MonlcTwo\in\condmonad_{\dlvaltreeTwo_{\lindex}}(\lcset)$ such that $\Monvaljudgment{\emptycontext}{\MonlcTwo}{\dlvaltreeTwo_{\lindex}}{\MonvalOne_{\lindex}}{\MontypeTwo}$. Knowing this and the last premise of the \textit{let} typing rule above, by expanding the lifted judgments and applying Lemma \ref{lem: substitution} $n$ times, we get
		
		\begin{equation}\label{eq: lifted-sub'}
			\forall\assOne_i\in\cassset_{\dlvaltreeTwo_{\lindex}}.\compjudgment{\emptycontext}{\MonlcTwo(\assOne_i),\lcOne'_{\rindex}}{\dlvaltreeTwo_{\assOne_i}}{\MontermOne(\assOne_i)[\MonvalOne_{\lindex}(\assOne_i)/\varOne]}{\MongenericObjectTwo(\assOne_i)}.
		\end{equation}
		
		%Branch a1	
		Now consider the configuration $\bscfgl{\circuitOne_1}{\assOne\assmerge\assOne_1}{\MontermOne(\assOne_1)[\MonvalOne_{\lindex}(\assOne_1)/\varOne]}$. Because by hypothesis $\assOne\in\cassset_{\dlvaltreeOne}$ and $\assOne_1\in\cassset_{\dlvaltreeTwo_{\lindex}}$, we know $\assOne\assmerge\assOne_1\in\cassset_{\Monflatten{\Moncomp{\dlvaltreeOne}{\dlvaltreeTwo_{\lindex}}{}{\{\assOne\}}}}$. Furthermore, we have $\namesof{\assOne\cup\assOne_1}(\Monflatten{\Moncomp{\dlvaltreeOne}{\dlvaltreeTwo_{\lindex}}{}{\{\assOne\}}}) \cap \namesof{}(\dlvaltreeTwo_{\assOne_1}) = (\namesof{\assOne}(\dlvaltreeOne)\cup\namesof{\assOne_1}(\dlvaltreeTwo_{\lindex}))\cap\namesof{}(\dlvaltreeTwo_{\assOne_1}) = \emptyset$, since $\namesof{\assOne}(\dlvaltreeOne)\cap\namesof{}(\dlvaltreeTwo_{\assOne_1}) \subseteq \namesof{\assOne}(\dlvaltreeOne)\cap\namesof{}(\dlvaltreeTwo)=\emptyset$ by hypothesis and $\namesof{\assOne_1}(\dlvaltreeTwo_{\lindex})\cap\namesof{}(\dlvaltreeTwo_{\assOne_1})=\emptyset$ since the flattening of $\Moncomp{\dlvaltreeTwo_{\lindex}}{\dlvaltreeTwo_{\assOne_i}}{\assOne_i}{}$ is well-defined in the \textit{let} typing rule above.
		Lastly, we consider again $\circjudgment{\circuitOne_1}{\Monflatten{\Moncomp{\dlvaltreeOne}{\dlvaltreeTwo_{\lindex}}{}{\{\assOne\}}}}{\lcOne}{\Monflatten{\Moncomp{\MonlcOne'}{\MonlcTwo,\lcOne'_{\rindex},\lcOne''}{}{\{\assOne\}}}}$ and we rewrite $\Monflatten{\Moncomp{\MonlcOne'}{\MonlcTwo,\lcOne'_{\rindex},\lcOne''}{}{\{\assOne\}}}$ first as $\Monflatten{\Moncomp{\MonlcOne'}{\Moncomp{\dlvaltreeTwo_{\lindex}}{\MonlcTwo(\assOne_j),\lcOne'_{\rindex},\lcOne''}{\assOne_j}{\{\assOne_j\mid j=1\dots n\}}}{}{\{\assOne\}}}$ and then by Lemma \ref{lem: comma extraction} as $\Moncomp{\Monflatten{\Moncomp{\MonlcOne'}{\dlvaltreeTwo_{\lindex}}{}{\{\assOne\}}}}{\MonlcTwo(\assOne_j),\lcOne'_{\rindex},\lcOne''}{\assOne\assmerge\assOne_j}{\{\assOne\assmerge\assOne_j\mid j=1\dots n\}}$. This, together with Equation \ref{eq: lifted-sub'}, lets us say
		\begin{equation}
			\begin{aligned}
		 \lbscfgjudgment{\lcOne}{\Monflatten{\Moncomp{\dlvaltreeOne}{\dlvaltreeTwo_{\lindex}}{}{\{\assOne\}}}}{\dlvaltreeTwo_{\assOne_1}}{&\bscfgl{\circuitOne_1}{\assOne\assmerge\assOne_1}{\MontermOne(\assOne_1)[\MonvalOne_{\lindex}(\assOne_1)/\varOne]}}{\MongenericObjectTwo(\assOne_1)}{\\&\Moncomp{\Moncomp{\Monflatten{\Moncomp{\MonlcOne'}{\dlvaltreeTwo_{\lindex}}{}{\{\assOne\}}}}{\lcOne''}{}{\{\assOne\assmerge\assOne_1\}}}{\MonlcTwo(\assOne_j),\lcOne'_{\rindex},\lcOne''}{\assOne\assmerge\assOne_j}{\{\assOne\assmerge\assOne_j\mid j=2\dots n\}}}.
			\end{aligned}
		\end{equation}
		For brevity, let $\dlvaltreeOne_{\lindex} = \Monflatten{\Moncomp{\dlvaltreeOne}{\dlvaltreeTwo_{\lindex}}{}{\{\assOne\}}}$ and $\MonlcOne_{\lindex} = \Monflatten{\Moncomp{\MonlcOne'}{\dlvaltreeTwo_{\lindex}}{}{\{\assOne\}}}$ in the rest of this proof case. At this point, either there exist $\circuitOne_2,\MonvalTwo_1$ such that $\bscfgl{\circuitOne_1}{\assOne\assmerge\assOne_1}{\MontermOne(\assOne_1)[\MonvalOne_{\lindex}(\assOne_1)/\varOne]}\bseval\bscfgr{\circuitOne_2}{\MonvalTwo_1}$ or not. In the latter case, we immediately conclude the claim by coinduction through the \textit{let-then} rule, with $j=1$. In the former case, by Theorem \ref{thm: SR} we get
		\begin{equation}
			\rbscfgjudgment{\lcOne}{\Monflatten{\Moncomp{\dlvaltreeOne_{\lindex}}{\dlvaltreeTwo_{\assOne_1}}{}{\{\assOne\assmerge\assOne_1\}}}}{}{\assOne\assmerge\assOne_1}{\bscfgr{\circuitOne_2}{\MonvalTwo_1}}{\MongenericObjectTwo(\assOne_1)}{\Monflatten{\Moncomp{\Moncomp{\MonlcOne_{\lindex}}{\Monconst{\lcOne''}{\dlvaltreeTwo_{\assOne_1}}}{}{\{\assOne\assmerge\assOne_1\}}}{\MonlcTwo(\assOne_j),\lcOne'_{\rindex},\lcOne''}{\assOne\assmerge\assOne_j}{\{\assOne\assmerge\assOne_j\mid j=2\dots n\}}}}.
		\end{equation}
	
		By the definition of well-typed right configuration, this entails
		\begin{equation}
			\circjudgment{\circuitOne_2}{\Monflatten{\Moncomp{\dlvaltreeOne_{\lindex}}{\dlvaltreeTwo_{\assOne_1}}{}{\{\assOne\assmerge\assOne_1\}}}}{\lcOne}{\Monflatten{\Moncomp{\Moncomp{\MonlcOne_{\lindex}}{\MonlcTwo'_1,\lcOne''}{}{\{\assOne\assmerge\assOne_1\}}}{\MonlcTwo(\assOne_j),\lcOne'_{\rindex},\lcOne''}{\assOne\assmerge\assOne_j}{\{\assOne\assmerge\assOne_j\mid j=2\dots n\}}}}
		\end{equation}
		for some $\MonlcTwo'_1\in\condmonad_{\dlvaltreeTwo_{\assOne_1}}(\lcset)$ such that $\Monvaljudgment{\emptycontext}{\MonlcTwo'_1}{\dlvaltreeTwo_{\assOne_1}}{\MonvalTwo_1}{\MongenericObjectTwo(\assOne_1)}$.
		Since the various label contexts $(\MonlcTwo(\assOne_j),\lcOne'_{\rindex},\lcOne'')$, where $j\neq 1$, are not tree-like structures and are composed on branches disjoint from $\assOne_1$, we can factor them out of the outermost flattening, thus rewriting the output of $\circuitOne_2$ as 
		\begin{equation}
			\Moncomp{\Monflatten{\Moncomp{\MonlcOne_{\lindex}}{\MonlcTwo'_1,\lcOne''}{}{\{\assOne\assmerge\assOne_1\}}}}{\MonlcTwo(\assOne_j),\lcOne'_{\rindex},\lcOne''}{\assOne\assmerge\assOne_j}{\{\assOne\assmerge\assOne_j\mid j=2\dots n\}}.
		\end{equation}
		
		%Branch a2
		Moving on, consider the configuration $\bscfgl{\circuitOne_2}{\assOne\assmerge\assOne_2}{\MontermOne(\assOne_2)[\MonvalOne_{\lindex}(\assOne_2)/\varOne]}$. We know that $\assOne_1$ and $\assOne_2$ are distinct, so because by hypothesis $\assOne\in\cassset_{\dlvaltreeOne}$ and $\assOne_2\in\cassset_{\dlvaltreeTwo_{\lindex}}$, we know $\assOne\assmerge\assOne_2\in\cassset_{\Monflatten{\Moncomp{\dlvaltreeOne_{\lindex}}{\dlvaltreeTwo_{\assOne_1}}{}{\{\assOne\assmerge\assOne_1\}}}}$. Furthermore, we have $\namesof{\assOne\cup\assOne_2}(\Monflatten{\Moncomp{\dlvaltreeOne_{\lindex}}{\dlvaltreeTwo_{\assOne_1}}{}{\{\assOne\assmerge\assOne_1\}}}) \cap \namesof{}(\dlvaltreeTwo_{\assOne_2}) = \namesof{\assOne\cup\assOne_2}(\Monflatten{\Moncomp{\Monflatten{\Moncomp{\dlvaltreeOne}{\dlvaltreeTwo_{\lindex}}{}{\{\assOne\}}}}{\dlvaltreeTwo_{\assOne_1}}{}{\{\assOne\assmerge\assOne_1\}}}) \cap \namesof{}(\dlvaltreeTwo_{\assOne_2}) = (\namesof{\assOne}(\dlvaltreeOne)\cup\namesof{\assOne_2}(\dlvaltreeTwo_{\lindex}))\cap\namesof{}(\dlvaltreeTwo_{\assOne_2}) = \emptyset$.
		This, together with $\circuitOne_2$'s signature and Equation \ref{eq: lifted-sub'}, lets us say
		\begin{equation}
			\begin{aligned}
				\lbscfgjudgment{\lcOne}{\Monflatten{\Moncomp{\dlvaltreeOne_{\lindex}}{\dlvaltreeTwo_{\assOne_1}}{}{\{\assOne\assmerge\assOne_1\}}}}{\dlvaltreeTwo_{\assOne_2}}{&\bscfgl{\circuitOne_2}{\assOne\assmerge\assOne_2}{\MontermOne(\assOne_2)[\MonvalOne_{\lindex}(\assOne_2)/\varOne]}}{\MongenericObjectTwo(\assOne_2)}{\\&\Moncomp{\Moncomp{\Monflatten{\Moncomp{\MonlcOne_{\lindex}}{\MonlcTwo'_1,\lcOne''}{}{\{\assOne\assmerge\assOne_1\}}}}{\lcOne''}{}{\{\assOne\assmerge\assOne_2\}}}{\MonlcTwo(\assOne_j),\lcOne'_{\rindex},\lcOne''}{\assOne\assmerge\assOne_j}{\{\assOne\assmerge\assOne_j\mid j=3\dots n\}}}
			\end{aligned}
		\end{equation}
		once again, either there exist $\circuitOne_3,\MonvalTwo_2$ such that $\bscfgl{\circuitOne_2}{\assOne\assmerge\assOne_2}{\MontermOne(\assOne_2)[\MonvalOne_{\lindex}(\assOne_2)/\varOne]}\bseval\bscfgr{\circuitOne_3}{\MonvalTwo_2}$ or not. In the latter case, we immediately conclude the claim by coinduction through the \textit{let-then} rule, with $j=2$. In the former case, by Theorem \ref{thm: SR} we get
		\begin{equation}
			\begin{aligned}
				\rbscfgjudgment{\lcOne}{\Monflatten{\Moncomp{\Monflatten{\Moncomp{\dlvaltreeOne_{\lindex}}{\dlvaltreeTwo_{\assOne_1}}{}{\{\assOne\assmerge\assOne_1\}}}}{\dlvaltreeTwo_{\assOne_2}}{}{\{\assOne\assmerge\assOne_2\}}}}{}{\dlvaltreeTwo_{\assOne_2}}{&\bscfgr{\circuitOne_3}{\MonvalTwo_2}}{\MongenericObjectTwo(\assOne_2)}{\\&\Moncomp{\Monflatten{\Moncomp{\Monflatten{\Moncomp{\MonlcOne_{\lindex}}{\MonlcTwo'_1,\lcOne''}{}{\{\assOne\assmerge\assOne_1\}}}}{\Monconst{\lcOne''}{\dlvaltreeTwo_{\assOne_2}}}{}{\{\assOne\assmerge\assOne_2\}}}}{\MonlcTwo(\assOne_j),\lcOne'_{\rindex},\lcOne''}{\assOne\assmerge\assOne_j}{\{\assOne\assmerge\assOne_j\mid j=3\dots n\}}},
			\end{aligned}
		\end{equation}
		which entails that $\circuitOne_3$ has signature
		\begin{equation}
			\begin{aligned}
				\circjudgment{\circuitOne_3&}{\Monflatten{\Moncomp{\Monflatten{\Moncomp{\dlvaltreeOne_{\lindex}}{\dlvaltreeTwo_{\assOne_1}}{}{\{\assOne\assmerge\assOne_1\}}}}{\dlvaltreeTwo_{\assOne_2}}{}{\{\assOne\assmerge\assOne_2\}}}}{\\&\lcOne}{\Moncomp{\Monflatten{\Moncomp{\Monflatten{\Moncomp{\MonlcOne_{\lindex}}{\MonlcTwo'_1,\lcOne''}{}{\{\assOne\assmerge\assOne_1\}}}}{\MonlcTwo'_2,\lcOne''}{}{\{\assOne\assmerge\assOne_2\}}}}{\MonlcTwo(\assOne_j),\lcOne'_{\rindex},\lcOne''}{\assOne\assmerge\assOne_j}{\{\assOne\assmerge\assOne_j\mid j=3\dots n\}}}
			\end{aligned}
		\end{equation}
		for some $\MonlcTwo'_2\in\condmonad_{\dlvaltreeTwo_{\assOne_2}}(\lcset)$ such that $\Monvaljudgment{\emptycontext}{\MonlcTwo'_2}{\dlvaltreeTwo_{\assOne_2}}{\MonvalTwo_2}{\MongenericObjectTwo(\assOne_2)}$.
		
		Now, when we consider the configuration $\bscfgl{\circuitOne_3}{\assOne\assmerge\assOne_3}{\MontermOne(\assOne_3)[\MonvalOne_{\lindex}(\assOne_3)/\varOne]}$ we can apply to it the same reasoning that we applied to $\bscfgl{\circuitOne_2}{\assOne\assmerge\assOne_2}{\MontermOne(\assOne_2)[\MonvalOne_{\lindex}(\assOne_2)/\varOne]}$ after we considered $\bscfgl{\circuitOne_1}{\assOne\assmerge\assOne_1}{\MontermOne(\assOne_1)[\MonvalOne_{\lindex}(\assOne_1)/\varOne]}$. In fact, we can apply this kind of reasoning pattern iteratively to the remaining $\bscfgl{\circuitOne_i}{\assOne\assmerge\assOne_i}{\MontermOne(\assOne_i)[\MonvalOne_{\lindex}(\assOne_i)/\varOne]}$, for $i=3\dots n$. At each step $i$, we know that $\bscfgl{\circuitOne_i}{\assOne\assmerge\assOne_i}{\MontermOne(\assOne_i)[\MonvalOne_{\lindex}(\assOne_i)/\varOne]}$ is well-typed and we distinguish between the case in which there exist $\circuitOne_{i+1},\MonvalTwo_i$ such that $\bscfgl{\circuitOne_i}{\assOne\assmerge\assOne_i}{\MontermOne(\assOne_i)[\MonvalOne_{\lindex}(\assOne_i)/\varOne]}\bseval\bscfgr{\circuitOne_{i+1}}{\MonvalTwo_i}$ and the case in which such $\circuitOne_{i+1},\MonvalTwo_i$ do not exist. In the latter case, we immediately conclude by coinduction through the \textit{let-then} rule, with $j=i$, while in the former case we proceed by Theorem \ref{thm: SR} and obtain that $\bscfgl{\circuitOne_{i+1}}{\assOne\assmerge\assOne_{i+1}}{\MontermOne(\assOne_{i+1})[\MonvalOne_{\lindex}(\assOne_{i+1})/\varOne]}$ is well-typed too. If after $n-2$ iterations we have not yet concluded, we end up in a situation in which we have $\bscfgl{\circuitOne_i}{\assOne\assmerge\assOne_i}{\MontermOne(\assOne_i)[\MonvalOne_{\lindex}(\assOne_i)/\varOne]}\bseval\bscfgr{\circuitOne_{i+1}}{\MonvalTwo_i}$ for $i=1\dots n$, where $\MonvalTwo_i\in\condmonad_{\dlvaltreeTwo_{\assOne_i}}(\valset)$. Because $\Monflatten{\Moncomp{\dlvaltreeTwo_{\lindex}}{\dlvaltreeTwo_{\assOne_i}}{\assOne_i}{}}$ is defined in the \textit{let} typing rule, we know that $\namesof{\assOne_i}(\dlvaltreeTwo_{\lindex})\cap\namesof{}(\dlvaltreeTwo_{\assOne_i})=\emptyset$ for $i=1,\dots,n$ and therefore $\Monflatten{\Moncomp{\dlvaltreeTwo_{\lindex}}{\MonvalTwo_i}{\assOne_i}{}}$ is also defined. This allows us to derive $\bscfgl{\circuitOne}{\assOne}{\letin{\varOne}{\termTwo}{\MontermOne}}\bseval\bscfgr{\circuitOne_{n+1}}{\Monflatten{\Moncomp{\dlvaltreeTwo_{\lindex}}{\MonvalTwo_i}{\assOne_i}{}}}$, contradicting the hypothesis.
	\end{proof}
}

\longversion{\subsection{Closed Computations}}
\longshortversion{Although the previous results hold for all computations starting from arbitrary well-typed configurations,}{That being said,} we are mainly interested in \textit{closed} computations, in which the evaluation of a term builds the underlying circuit entirely from scratch. That is, we are interested in computations that start from configurations of the form $\bscfgl{\cinput{\emptylc}}{\emptyassignment}{\termOne}$, for some $\termOne$.
Whenever $\bscfgl{\cinput{\emptylc}}{\emptyassignment}{\termOne}\bseval\bscfgr{\circuitOne}{\MonvalOne}$ for some $\circuitOne,\MonvalOne$, we simply write $\termOne\bseval\bscfgr{\circuitOne}{\MonvalOne}$ and whenever $\bscfgl{\cinput{\emptylc}}{\emptyassignment}{\termOne}\bsdiverge$ we write $\termOne\bsdiverge$.
In the same spirit, we say that $\termOne$ is a \emph{well-typed term with lifted type $\MontypeOne$ depending on $\dlvaltreeOne$}, and we write $\simpjudgment{\termOne}{\dlvaltreeOne}{\MontypeOne}$, whenever $\lbscfgjudgment{\emptylc}{\emptytree}{\dlvaltreeOne}{\bscfgl{\cinput{\emptylc}}{\emptyassignment}{\termOne}}{\MontypeOne}{\Monunit{\emptylc}}$, while we say that $\bscfgr{\circuitOne}{\MonvalOne}$ is a \emph{well-typed closed configuration with lifted type $\MontypeOne$ depending on $\dlvaltreeOne$}, and we write $\simpjudgment{\bscfgr{\circuitOne}{\MonvalOne}}{\dlvaltreeOne}{\MontypeOne}$, whenever $\rbscfgjudgment{\emptylc}{\dlvaltreeOne}{\dlvaltreeOne}{\emptyassignment}{\bscfgr{\circuitOne}{\MonvalOne}}{\MontypeOne}{\dlvaltreeOne[\emptylc]_\assTwo}$.
\longshortversion{The previous subject reduction and progress results are therefore declined in the following versions for closed computations.}{We are now ready to give the relevant type safety results for \PQK.}

\longshortversion{
	\begin{theorem}[Subject Reduction for Closed Computations]
		If $\simpjudgment{\termOne}{\dlvaltreeOne}{\MontypeOne}$ and $\exists\circuitOne,\MonvalOne.\termOne \simpbseval \bscfgr{\circuitOne}{\MonvalOne}$, then $\simpjudgment{\bscfgr{\circuitOne}{\MonvalOne}}{\dlvaltreeOne}{\MontypeOne}$.
	\end{theorem}
	\begin{proof}
		The claim follows naturally from Theorem \ref{thm: SR}.
	\end{proof}
	
	\begin{theorem}[Progress for Closed Computations]
		If $\simpjudgment{\termOne}{\dlvaltreeOne}{\MontypeOne}$, then either $\exists\circuitOne,\MonvalOne.\termOne \simpbseval \bscfgr{\circuitOne}{\MonvalOne}$ or $\termOne\simpbsdiverge$.
	\end{theorem}
	\begin{proof}
		The claim follows naturally from Theorem \ref{thm: Progress}.
	\end{proof}
}
{
	\begin{theorem}[Subject Reduction]
		\label{thm: closed SR}
		If $\simpjudgment{\termOne}{\dlvaltreeOne}{\MontypeOne}$ and $\exists\circuitOne,\MonvalOne.\termOne \simpbseval \bscfgr{\circuitOne}{\MonvalOne}$, then $\simpjudgment{\bscfgr{\circuitOne}{\MonvalOne}}{\dlvaltreeOne}{\MontypeOne}$.
	\end{theorem}
	\begin{proof}
		We prove the more general claim that if $\lbscfgjudgment{\lcOne}{\dlvaltreeTwo}{\dlvaltreeOne}{\bscfgl{\circuitTwo}{\assOne}{\termOne}}{\MontypeOne}{\MonlcOne}$ and $\exists\circuitOne,\MonvalOne.\bscfgl{\circuitTwo}{\assOne}{\termOne}\bseval\bscfgr{\circuitOne}{\MonvalOne}$, then $\rbscfgjudgment{\lcOne}{\Monflatten{\Moncomp{\dlvaltreeTwo}{\dlvaltreeOne}{}{\{\assOne\}}}}{\dlvaltreeTwo}{\assOne}{\bscfgr{\circuitOne}{\MonvalOne}}{\MontypeOne}{\MonlcOne\Monsplit{\assOne}\dlvaltreeOne}$, by induction on $\bscfgl{\circuitTwo}{\assOne}{\termOne}\bseval\bscfgr{\circuitOne}{\MonvalOne}$. The subject reduction claim is then obtained by choosing $\circuitTwo=\cinput{\emptylc},\assOne=\emptyassignment,\lcOne=\emptylc,\dlvaltreeTwo=\emptytree,\MonlcOne=\Monunit{\emptylc}$.
	\end{proof}
	
	\begin{theorem}[Progress]
		\label{thm: closed progress}
		If $\simpjudgment{\termOne}{\dlvaltreeOne}{\MontypeOne}$, then either $\exists\circuitOne,\MonvalOne.\termOne \simpbseval \bscfgr{\circuitOne}{\MonvalOne}$ or $\termOne\simpbsdiverge$.
	\end{theorem}
	\begin{proof}
		We consider the equivalent claim that if $\simpjudgment{\termOne}{\dlvaltreeOne}{\MontypeOne}$ and $\nexists\circuitOne,\MonvalOne.\termOne \simpbseval \bscfgr{\circuitOne}{\MonvalOne}$, then $\termOne\simpbsdiverge$. We prove the more general result that if $\lbscfgjudgment{\lcOne}{\dlvaltreeTwo}{\dlvaltreeOne}{\bscfgl{\circuitTwo}{\assOne}{\termOne}}{\MontypeOne}{\MonlcOne}$ and $\nexists\circuitOne,\MonvalOne.\bscfgl{\circuitTwo}{\assOne}{\termOne}\bseval\bscfgr{\circuitOne}{\MonvalOne}$, then $\bscfgl{\circuitTwo}{\assOne}{\termOne}\bsdiverge$, by coinduction and case analysis on $\termOne$. The progress claim is then obtained by choosing $\circuitTwo=\cinput{\emptylc},\assOne=\emptyassignment,\lcOne=\emptylc,\dlvaltreeTwo=\emptytree,\MonlcOne=\Monunit{\emptylc}$.
	\end{proof}
	
	The complete proofs of theorems \ref{thm: closed SR} and \ref{thm: closed progress} can be found in the extended version of this paper \cite{extended-ver}.
}

%%%%%%%%%%%%%%%%%%%%%%%%%%		
\section{Conclusion}
\label{sect:related}
%%%%%%%%%%%%%%%%%%%%%%%%%%

\paragraph*{Our Contributions}
This paper introduces a new paradigmatic language for dynamic lifting belonging to the \PQ\ family of languages. The language, called \PQK, can be seen as an extension of \PQM\ which allows for the Boolean information flowing within bit wires to be lifted from the circuit level to the program level. In order to make the circuit construction process as flexible and as general as possible, a powerful type and effect system based on lifting trees has been introduced. This allows for the typing of programs which produce highly non-uniform circuits. The main technical results we obtained are type soundness, in the sense of subject reduction and progress theorems.

\paragraph*{Future Work}
In this paper we focused on the operational aspects, leaving an 
investigation about a possible denotational account of \PQK\ as future work. 
A related problem is that of understanding the precise nature of the lifting 
operation that we use pervasively in the paper. In particular, while it would be tempting to interpret $\condmonad_{\dlvaltreeOne}$ as a graded monad with unit 
$\munit:\genericSetOne\to\condmonad_\emptytree(\genericSetOne)$ defined as 
$\munit\genericObjectOne=\treeLeaf{\genericObjectOne}$ and multiplication 
$\Monflatten{\cdot}: 
\condmonad_\dlvaltreeOne(\condmonad_\dlvaltreeTwo(\genericSetOne))\to\condmonad_{\Moncomp{\dlvaltreeOne}{\dlvaltreeTwo}{\assOne}{}}(\genericSetOne)$, we generally work with objects that do not belong to $\condmonad_\dlvaltreeOne(\condmonad_\dlvaltreeTwo(\genericSetOne))$ for some fixed $\dlvaltreeOne,\dlvaltreeTwo$, so we find that this interpretation is not entirely appropriate, at least if we assume grades to be elements of \emph{one fixed monoid}.

Another problem that we leave
open concerns the notion of generalized circuits used in this paper. On the one hand, it can certainly be 
said that it can model quantum circuits in their full generality, including 
those measurement patterns found in measurement-based quantum computing~\cite{measurement-calculus}. On the other hand, while it is 
clear that such circuits make computational sense (after all, any such computation can be simulated by a classically controlled quantum Turing 
machines~\cite{ccqc}), it is not clear what kind of correspondence exists between them and quantum circuits in their usual form~\cite{nielsen-chuang}, which is the one 
most often considered in the literature.

\paragraph*{Related Work}
As already mentioned in the Introduction, various paradigmatic  		$\lambda$-calculi modeling the \Quipper\ programming language have been introduced in the literature~\cite{proto-quipper-d,proto-quipper-l,proto-quipper-m,proto-quipper-s}. In this work, we took inspiration from \PQM\ \cite{proto-quipper-m}, which however cannot handle dynamic lifting. The only member of the \PQ \ family that can handle dynamic lifting, in a uniform and more restricted form than ours, is \PQL\ \cite{proto-quipper-l}, which has been introduced very recently  		and independently of this work. Noticeably, the class of circuits targeted by \PQL\ (which the authors have named \textit{quantum channels}) indeed includes non-uniform circuits like the one in  Figure~\ref{fig: one-way}, but the type system of the language rejects the programs that would build them. Most of the aforementioned contributions differ from ours in that they focus mainly on the denotational semantics of the language rather than its
operational semantics.

The type and effect system paradigm is well known from the literature \cite{monadic-affine,parametric-effect-monads,effect-systems-revisited,types-and-effects} and has been used in various contexts as a way to reflect information on the effects produced by the underlying program in its type. In our case, the relevant effect is a choice effect, which is mirrored both in the operational semantics and in the type system. As already stated, the problem of giving a proper monadic status to the considered choice effect remains open, although this work shows that \emph{operationally speaking} everything works smoothly.

Finally, it is worth mentioning that the circuit description paradigm is not the only approach to designing quantum programming languages and calculi. For instance, \QCL\ \cite{qcl}, \QML\ \cite{qml} and Selinger and Valiron's quantum $\lambda$-calculus \cite{quantum-lambda-calculus} are some examples of quantum programming languages whose instructions are designed to be executed individually and directly on quantum hardware, without any direct reference to quantum circuits.

\bibliography{bibliography}

\newpage

\appendix

%%%%%%%%%%%%%%%%%%%%%%%%%%%%%%%%
\section{Type Derivations}
\label{app: type derivations}
%%%%%%%%%%%%%%%%%%%%%%%%%%%%%%%%

\begin{figure}[ht]
	\fbox{\begin{minipage}{.98\textwidth}
		$$
		\inference[abs]
		{
		\inference[return]
		{
		\inference[abs]
		{
		\inference[let]
		{
			\inference[apply]
			{
				\valjudgment{\emptycontext}{\emptycontext}{\gateid{H}}{\circt{\emptytree}(\qubitt,\Monunit{\qubitt})}
				&
				\inference[var]{}
				{
					\valjudgment{q:\qubitt}{\emptycontext}{q}{\qubitt}
				}
			}{
				\compjudgment{q:\qubitt}{\emptycontext}{\emptytree}{\apply{}(\gateid{H},q)}{\Monunit{\qubitt}}
			}
			&
			\Pi
			\\
			\Monunit{\letin{\_}{\apply{u}(\gateid{ML},q)}{\treeNode{u}{\return a}{\apply{}(\gateid{Meas},a)}}} \in \condmonad_\emptytree(\termset)
			}
		{
			\begin{aligned}
				\compjudgment{a:\qubitt,q:\qubitt}{\emptycontext}{\treeNode{\dlvalOne}{\emptytree}{\emptytree}}{
					&\letin{q}{\apply{}(\gateid{H},q)}{\Monunit{\\
						&\letin{\_}{\apply{u}(\gateid{ML},q)}{\\
							&\treeNode{u}{\return a}{\apply{}(\gateid{Meas},a)}}}}}
				{\treeNode{u}{\qubitt}{\bitt}}
			\end{aligned}}
		}
		{
			\begin{aligned}
				\valjudgment{q:\qubitt}{\emptycontext}{
					\lambda a_{\qubitt}.
					&\letin{q}{\apply{}(\gateid{H},q)}{\Monunit{\\
					&\letin{\_}{\apply{u}(\gateid{ML},q)}{\\
							&\treeNode{u}{\return a}{\apply{}(\gateid{Meas},a)}}}}}
				{\qubitt\multimap\treeNode{u}{\qubitt}{\bitt}}
		\end{aligned}}
		}
		{
		\begin{aligned}
			\compjudgment{q:\qubitt}{\emptycontext}{\emptytree}{
				\return\lambda a_{\qubitt}.
				&\letin{q}{\apply{}(\gateid{H},q)}{\Monunit{\\
				&\letin{\_}{\apply{u}(\gateid{ML},q)}{\\
				&\treeNode{u}{\return a}{\apply{}(\gateid{Meas},a)}}}}\\&}
			{\Monunit{\qubitt\multimap\treeNode{u}{\qubitt}{\bitt}}}
		\end{aligned}}
		}
		{
		\begin{aligned}
			\valjudgment{\emptycontext}{\emptycontext}{
				\lambda q_{\qubitt}.\return\lambda a_{\qubitt}.
				&\letin{q}{\apply{}(\gateid{H},q)}{\Monunit{\\
				&\letin{\_}{\apply{u}(\gateid{ML},q)}{\\
				&\treeNode{u}{\return a}{\apply{}(\gateid{Meas},a)}}}}\\&}
			{\qubitt\multimap\Monunit{\qubitt\multimap\treeNode{u}{\qubitt}{\bitt}}}
			\end{aligned}}	
		$$
	\end{minipage}}
	\caption{Type derivation for the \PQK\ program in Figure \protect\ref{fig: pqk one-way}. $\gateid{H},\gateid{Meas}$ and $\gateid{ML}$ are shorthand for the boxed circuits employed in Figure \protect\ref{fig: pqk one-way}. Arrow annotations are omitted for brevity. Sub-derivation $\Pi$ is given in Figure \protect\ref{fig: subder}.}
\end{figure}

\begin{figure}[t]
	\fbox{\begin{minipage}{.98\textwidth}
		$$
		\inference[let]
		{
			\inference[apply]
			{
				\valjudgment{\emptycontext}{\emptycontext}{\gateid{ML}}{\circt{\treeNode{\dlvalOne}{\emptytree}{\emptytree}}(\qubitt,\treeNode{\dlvalOne}{\unitt}{\unitt})}
				&
				\inference[var]{}
				{
					\valjudgment{q:\qubitt}{\emptycontext}{q}{\qubitt}
				}
			}{
				\compjudgment{q:\qubitt}{\emptycontext}{\treeNode{u}{\emptytree}{\emptytree}}{\apply{u}(\gateid{ML},q)}{\treeNode{\dlvalOne}{\unitt}{\unitt}}
			}
			&
			\Pi'
			&
			\Pi''
			\\
			\treeNode{u}{\return a}{\apply{}(\gateid{Meas},a)}\in\condmonad_{\treeNode{\dlvalOne}{\emptytree}{\emptytree}}(\termset)
		}{
			\begin{aligned}
				\compjudgment{q:\qubitt,a:\qubitt}{\emptycontext}{\treeNode{\dlvalOne}{\emptytree}{\emptytree}}{
					&\letin{\_}{\apply{u}(\gateid{ML},q)}{\\
						&\treeNode{u}{\return a}{\apply{}(\gateid{Meas},a)}}
				}{
					\treeNode{u}{\qubitt}{\bitt}
				}
			\end{aligned}
		}
		$$
	\end{minipage}}
	\caption{Sub-derivation $\Pi$. Sub-derivations $\Pi'$ and $\Pi''$ are given in figures \protect\ref{fig: subsubder0} and \protect\ref{fig: subsubder1}, respectively.}
	\label{fig: subder}
\end{figure}

\begin{figure}[ht]
	\fbox{\begin{minipage}{.98\textwidth}
			$$
			\inference[return]
			{
				\inference[var]{\void}
				{
					\valjudgment{\_:\unitt,a:\qubitt}{\emptycontext}{a}{\qubitt}
				}
			}
			{
				\compjudgment{\_:\unitt,a:\qubitt}{\emptycontext}{\emptytree}{\return a}{\Monunit{\qubitt}}
			}
			$$
		\end{minipage}
	}
	\caption{Sub-derivation $\Pi'$, corresponding to the $u=0$ branch of the lifted computational judgment $\Moncompjudgment{\_:\treeNode{u}{\unitt}{\unitt},a:\qubitt}{\emptycontext}{\treeNode{u}{\Monunit{\emptytree}}{\Monunit{\emptytree}}}{\treeNode{u}{\return a}{\apply{}(\gateid{Meas},a)}}{\treeNode{u}{\Monunit{\qubitt}}{\Monunit{\bitt}}}$}
	\label{fig: subsubder0}
\end{figure}

\begin{figure}[ht]
	\fbox{\begin{minipage}{.98\textwidth}
			$$
			\inference[apply]
			{
				\valjudgment{\_:\unitt}{\emptycontext}{\gateid{Meas}}{\circt{\emptytree}(\qubitt,\Monunit{\bitt})}
				&
				\inference[var]{\void}
				{
					\valjudgment{\_:\unitt,a:\qubitt}{\emptycontext}{a}{\qubitt}
				}
			}
			{
				\compjudgment{\_:\unitt,a:\qubitt}{\emptycontext}{\emptytree}{\apply{}(\gateid{Meas},a)}{\Monunit{\bitt}}
			}
			$$
	\end{minipage}}
	\caption{Sub-derivation $\Pi''$, corresponding to the $u=1$ branch of the lifted computational judgment $\Moncompjudgment{\_:\treeNode{u}{\unitt}{\unitt},a:\qubitt}{\emptycontext}{\treeNode{u}{\Monunit{\emptytree}}{\Monunit{\emptytree}}}{\treeNode{u}{\return a}{\apply{}(\gateid{Meas},a)}}{\treeNode{u}{\Monunit{\qubitt}}{\Monunit{\bitt}}}$}
	\label{fig: subsubder1}
\end{figure}

\end{document}